\documentstyle[12pt,fleqn]{book}

\font\twlgot =eufm10 scaled \magstep1 \font\egtgot =eufm8
\font\sevgot =eufm7

\font\twlmsb =msbm10 scaled \magstep1 \font\egtmsb =msbm8
\font\sevmsb =msbm7

\newfam\gotfam
\def\pgot{\fam\gotfam\twlgot}
\textfont\gotfam\twlgot \scriptfont\gotfam\egtgot
\scriptscriptfont\gotfam\sevgot
\def\got{\protect\pgot}

\newfam\msbfam
\textfont\msbfam\twlmsb \scriptfont\msbfam\egtmsb
\scriptscriptfont\msbfam\sevmsb
\def\Bbb{\protect\pBbb}
\def\pBbb{\relax\ifmmode\expandafter\Bb\else\typeout{You cann't use
Bbb in text mode}\fi}
\def\Bb #1{{\fam\msbfam\relax#1}}

\def\op#1{\mathop{{\it\fam0} #1}\limits}

\newcommand{\id}{{\rm Id\,}}

\newcommand{\pr}{{\rm pr}}

\newcommand{\di}{{\rm dim\,}}

\newcommand{\Id}{{\rm Id}}
\newcommand{\Ker}{{\rm Ker\,}}
\newcommand{\im}{{\rm Im\, }}

\newcommand{\nm}[1]{|{#1}|}

\newcommand{\dv}[1]{\dot{\rm{#1}}}

\newcommand{\bite}{\begin{itemize}}
\newcommand{\eite}{\end{itemize}}
\newcommand{\benu}{\begin{enumerate}}
\newcommand{\eenu}{\end{enumerate}}
\newcommand{\bde}{\begin{description}}
\newcommand{\ede}{\end{description}}
\newcommand{\bquo}{\begin{quote}}
\newcommand{\equo}{\end{quote}}
\newcommand{\bquot}{\begin{quotation}}
\newcommand{\equot}{\end{quotation}}

\newcommand{\eqref}[1]{(\ref{#1})}
\newcommand{\beq}{\begin{equation}}
\newcommand{\eeq}{\end{equation}}
\newcommand{\ben}{\begin{eqnarray}}
\newcommand{\een}{\end{eqnarray}}
\newcommand{\be}{\begin{eqnarray*}}
\newcommand{\ee}{\end{eqnarray*}}
\newcommand{\bea}{\begin{eqalph}}
\newcommand{\eea}{\end{eqalph}}

\newcommand{\gU}{{\got U}}

\newcommand{\gd}{{\got d}}

\newcommand{\gE}{{\got E}}

\newcommand{\gT}{{\got T}}

\newcommand{\gL}{{\got L}}

\newcommand{\cO}{{\cal O}}
\newcommand{\cT}{{\cal T}}

\newcommand{\cL}{{\cal L}}

\newcommand{\cE}{{\cal E}}
\newcommand{\cH}{{\cal H}}
\newcommand{\cF}{{\cal F}}
\newcommand{\cC}{{\cal C}}

\newcommand{\ccG}{{\cal G}}

\newcommand{\bL}{{\bf L}}

\newcommand{\bE}{{\bf E}}

\newcommand{\rA}{{\rm Ann\,}}

\newcommand{\rrq}{{\ol q}}

\newcommand{\al}{\alpha}
\newcommand{\bt}{\beta}
\newcommand{\dl}{\delta}
\newcommand{\la}{\lambda}
\newcommand{\La}{\Lambda}
\newcommand{\f}{\phi}

\newcommand{\F}{\Phi}
\newcommand{\p}{\pi}
\newcommand{\s}{\psi}

\newcommand{\Om}{\Omega}
\newcommand{\m}{\mu}
\newcommand{\n}{\nu}
\newcommand{\g}{\gamma}
\newcommand{\G}{\Gamma}
\newcommand{\e}{\epsilon}

\newcommand{\thh}{\theta}

\newcommand{\vr}{\varrho}
\newcommand{\up}{\upsilon}
\newcommand{\vt}{\vartheta}

\newcommand{\si}{\sigma}
\newcommand{\Si}{\Sigma}

\newcommand{\bom}{{\bf\Omega}}
\newcommand{\bth}{{\bf\Theta}}
\newcommand{\bT}{{\bf T}}

\newcommand{\Y}{Y\to X}
\newcommand{\w}{\wedge}

\newcommand{\wt}{\widetilde}
\newcommand{\wh}{\widehat}
\newcommand{\ol}{\overline}

\newcommand{\dr}{\partial}

\newcommand{\ar}{\op\longrightarrow}

\newcommand{\llra}{\longleftrightarrow}

\newcommand{\xx}{\times}
\newcommand{\ox}{\otimes}

\newcommand{\ot}{\otimes}
\newcommand{\ap}{\approx}

\let\ssection=\section
\renewcommand{\section}{\setcounter{equation}{0}\ssection}

\newcounter{eqalph}[section]
\newcounter{equationa}[section]
\newcounter{example}[section]
\newcounter{remark}[section]
\newcounter{theorem}[section]
\newcounter{proposition}[section]
\newcounter{lemma}[section]
\newcounter{corollary}[section]
\newcounter{definition}[section]

\def\theremark{\arabic{chapter}.\arabic{section}.\arabic{remark}}

\def\thedefinition{\arabic{chapter}.\arabic{section}.\arabic{definition}}

\newenvironment{proof}{\noindent {\it Proof:}\small
}{{\footnotesize\it QED}\medskip}
\newenvironment{ex}{\refstepcounter{remark}\medskip\noindent{\bf
Example \theremark:}\small }{$\diamondsuit$ \medskip}
\newenvironment{rem}{\refstepcounter{remark}\medskip\noindent{\bf
Remark \theremark:}\small }{$\diamondsuit$\medskip}
\newenvironment{theo}{\refstepcounter{definition}
\medskip\noindent{\sc Theorem \thedefinition}:}{$\Box$\medskip}
\newenvironment{prop}{\refstepcounter{definition}\medskip\noindent{\sc
Proposition \thedefinition}:}{$\Box$\medskip}
\newenvironment{lem}{\refstepcounter{definition}\medskip\noindent{\sc
Lemma \thedefinition}:}{ $\Box$\medskip }
\newenvironment{cor}{\refstepcounter{definition}\medskip\noindent{\sc
Corollary \thedefinition}:}{$\Box$\medskip}
\newenvironment{defi}{\refstepcounter{definition}\medskip\noindent{\sc
Definition \thedefinition}:}{$\Box$\medskip}

\newenvironment{eqalph}{\stepcounter{equation}
\setcounter{equationa}{\value{equation}} \setcounter{equation}{0}

\begin{eqnarray}}{\end{eqnarray}
\setcounter{equation}{\value{equationa}}}

\newcommand{\mar}[1]{}

\begin{document}

\hbox{}

\thispagestyle{empty}

\setcounter{page}{0}

\vskip 3cm

\begin{center}

{\large \bf Advanced Mechanics. Mathematical Introduction}

\bigskip
\bigskip
\bigskip

{\sc G. Giachetta, L. Mangiarotti}
\bigskip

Department of Mathematica and Informatics, University of Camerino,
Camerino, Italy

\bigskip

{\sc G. Sardanashvily}
\bigskip

Department of Theoretical Physics, Moscow State University,
Moscow, Russia

\bigskip
\bigskip
\bigskip
\bigskip
\bigskip
\bigskip

{\bf Abstract}
\bigskip
\end{center}

\noindent Classical non-relativistic mechanics in a general
setting of time-dependent transformations and reference frame
changes is formulated in the terms of fibre bundles over the
time-axis $\Bbb R$. Connections on fibre bundles are the main
ingredient in this formulation of mechanics which thus is
covariant under reference frame transformations. The basic notions
of a non-relativistic reference frame, a relative velocity, a free
motion equation, a relative acceleration, an external force are
formulated. Newtonian, Lagrangian, Hamiltonian mechanical systems
and the relations between them are defined. Lagrangian and
Hamiltonian conservation laws are considered.

\newpage

\bigskip
\bigskip
\bigskip

\noindent {\huge \bf Contents}

\bigskip
\bigskip
\bigskip

\contentsline {chapter}{Introduction}{3} \contentsline
{chapter}{\numberline {1}Dynamic equations}{5} \contentsline
{section}{\numberline {1.1}Preliminary. Fibre bundles over $\pBbb
R$}{5} \contentsline {section}{\numberline {1.2}Autonomous dynamic
equations}{10} \contentsline {section}{\numberline {1.3}Dynamic
equations}{12} \contentsline {section}{\numberline {1.4}Dynamic
connections}{14} \contentsline {section}{\numberline
{1.5}Non-relativistic geodesic equations}{17} \contentsline
{section}{\numberline {1.6}Reference frames}{22} \contentsline
{section}{\numberline {1.7}Free motion equations}{24}
\contentsline {section}{\numberline {1.8}Relative
acceleration}{26} \contentsline {section}{\numberline
{1.9}Newtonian systems}{29} \contentsline {section}{\numberline
{1.10}Integrals of motion}{31} \contentsline {chapter}{\numberline
{2}Lagrangian mechanics}{33} \contentsline {section}{\numberline
{2.1}Lagrangian formalism on $Q\to \pBbb R$}{33} \contentsline
{section}{\numberline {2.2}Cartan and Hamilton--De Donder
equations}{39} \contentsline {section}{\numberline {2.3}Lagrangian
and Newtonian systems}{41} \contentsline {section}{\numberline
{2.4}Lagrangian conservation laws}{43} \contentsline
{subsection}{\numberline {2.4.1}Generalized vector fields}{43}
\contentsline {subsection}{\numberline {2.4.2}First Noether
theorem}{45} \contentsline {subsection}{\numberline {2.4.3}Noether
conservation laws}{48} \contentsline {subsection}{\numberline
{2.4.4}Energy conservation laws}{49} \contentsline
{section}{\numberline {2.5}Gauge symmetries}{51} \contentsline
{chapter}{\numberline {3}Hamiltonian mechanics}{55} \contentsline
{section}{\numberline {3.1}Hamiltonian formalism on $Q\to \pBbb
R$}{55} \contentsline {section}{\numberline {3.2}Homogeneous
Hamiltonian formalism}{60} \contentsline {section}{\numberline
{3.3}Lagrangian form of Hamiltonian formalism}{61} \contentsline
{section}{\numberline {3.4}Associated Lagrangian and Hamiltonian
systems}{61} \contentsline {section}{\numberline {3.5}Hamiltonian
conservation laws}{65} \contentsline {chapter}{\numberline
{4}Appendixes}{71} \contentsline {section}{\numberline
{4.1}Geometry of fibre bundles}{71} \contentsline
{subsection}{\numberline {4.1.1}Fibred manifolds}{71}
\contentsline {subsection}{\numberline {4.1.2}Fibre bundles}{73}
\contentsline {subsection}{\numberline {4.1.3}Vector and affine
bundles}{76} \contentsline {subsection}{\numberline {4.1.4}Vector
and multivector fields}{80} \contentsline {subsection}{\numberline
{4.1.5}Differential forms}{82} \contentsline
{subsection}{\numberline {4.1.6}Distributions and foliations}{88}
\contentsline {section}{\numberline {4.2}Jet manifolds}{89}
\contentsline {subsection}{\numberline {4.2.1}First order jet
manifolds}{89} \contentsline {subsection}{\numberline
{4.2.2}Second order jet manifolds}{91} \contentsline
{subsection}{\numberline {4.2.3}Higher order jet manifolds}{92}
\contentsline {subsection}{\numberline {4.2.4}Differential
operators and differential equations}{93} \contentsline
{section}{\numberline {4.3}Connections on fibre bundles}{95}
\contentsline {subsection}{\numberline {4.3.1}Connections}{95}
\contentsline {subsection}{\numberline {4.3.2}Flat
connections}{98} \contentsline {subsection}{\numberline
{4.3.3}Linear connections}{98} \contentsline
{subsection}{\numberline {4.3.4}Composite connections}{100}
\contentsline {chapter}{Bibliography}{103}





\newpage

\chapter*{Introduction}

We address classical non-relativistic mechanics in a general
setting of arbitrary time-dependent coordinate and reference frame
transformations \cite{book10,book98}.

The technique of symplectic manifolds is well known to provide the
adequate Hamiltonian formulation of autonomous mechanics
\cite{abr,libe,vais}. Its familiar example is a mechanical system
whose configuration space is a manifold $M$ and whose phase space
is the cotangent bundle $T^*M$ of $M$ provided with the canonical
symplectic form
\mar{m83'}\beq
\Om= dp_i\w dq^i, \label{m83'}
\eeq
written with respect to the holonomic coordinates $(q^i, p_i=\dot
q_i)$ on $T^*M$.  A Hamiltonian $\cH$ of this mechanical system is
defined as a real function on a phase space $T^*M$. Any autonomous
Hamiltonian system locally is of this type.

However, this Hamiltonian formulation of autonomous mechanics is
not extended to mechanics under time-dependent transformations
because the symplectic form (\ref{m83'}) fails to be invariant
under these transformations. As a palliative variant, one develops
time-dependent mechanics on a configuration space $Q=\Bbb R\times
M$ where $\Bbb R$ is the time axis \cite{eche,leon}. Its phase
space $\Bbb R\times T^*M$ is provided with the presymplectic form
$\pr^*_2\Om =dp_i\w dq^i$ which is the pull-back of the canonical
symplectic form $\Om$ (\ref{m83'}) on $T^*M$. However, this
presymplectic form also is broken by time-dependent
transformations.

We consider non-relativistic mechanics whose configuration space
is a fibre bundle $Q\to \Bbb R$ over the time axis $\Bbb R$
endowed with the standard Cartesian coordinate $t$ possessing
transition functions $t'=t+$const. A velocity space of
non-relativistic mechanics is the first order jet manifold $J^1Q$
of sections of $Q\to \Bbb R$, and its phase space is the vertical
cotangent bundle $V^*Q$ of $Q\to\Bbb R$ endowed with the canonical
Poisson structure \cite{book98,book10,sard98}.

A fibre bundle $Q\to\Bbb R$ is always trivial. Its trivialization
defines both an appropriate coordinate systems and a connection on
this fibre bundle which is associated with a certain
non-relativistic reference frame. Formulated as theory on fibre
bundles over $\Bbb R$, non-relativistic mechanics is covariant
under gauge (atlas) transformations of these fibre bundles, i.e.,
time-dependent coordinate and reference frame transformations.

This formulation of mechanics is similar to that of classical
field theory on fibre bundles over a smooth manifold $X$, $\di
X>1$, \cite{book09}. A difference between mechanics and field
theory, however, lies in the fact that all connections on fibre
bundles over $\Bbb R$ are flat and, consequently, they are not
dynamic variables. Therefore, formulation of non-relativistic
mechanics is covariant, but not invariant under time-dependent
transformations.

Equations of motion of non-relativistic mechanics almost always
are first and second order dynamic equations. Second order dynamic
equations on a fibre bundle $Q\to\Bbb R$ are conventionally
defined as the holonomic connections on the jet bundle
$J^1Q\to\Bbb R$. These equations also are represented by
connections on the jet bundle $J^1Q\to Q$. The notions of a free
motion equation and a relative acceleration are formulated in
terms of connections on $J^1Q\to Q$ and $TQ\to Q$.

Generalizing the second Newton law, one introduces the notion of a
Newtonian system characterized by a mass tensor. If a mass tensor
is non-degenerate, an equation of motion of a Newtonian system is
equivalent to a dynamic equation. We also come to the definition
of an external force.

Lagrangian non-relativistic mechanics is formulated in the
framework of conventional Lagrangian formalism on fibre bundles
\cite{book,book09,book10,book98}. Its Lagrangian is defined as a
density on the velocity space $J^1Q$, and the corresponding
Lagrange equation is a second order differential equations on
$Q\to\Bbb R$. Besides Lagrange equations, the Cartan and
Hamilton--De Donder equations are considered in the framework of
Lagrangian formalism. Note that the Cartan equation, but not the
Lagrange one is associated to the Hamilton equation. The relations
between Lagrangian and Newtonian systems are established.
Lagrangian conservation laws are defined in accordance with the
first Noether theorem.

Hamiltonian mechanics on a phase space $V^*Q$ is not familiar
Poisson Hamiltonian theory on a Poisson manifold $V^*Q$ because
all Hamiltonian vector fields on $V^*Q$ are vertical. Hamiltonian
mechanics on $V^*Q$ is formulated as particular (polysymplectic)
Hamiltonian formalism on fibre bundles
\cite{book,book09,book10,book98}. Its Hamiltonian is a section of
the fibre bundle $T^*Q\to V^*Q$. The pull-back of the canonical
Liouville form on $T^*Q$ with respect to this section is a
Hamiltonian one-form on $V^*Q$. The corresponding Hamiltonian
connection on $V^*Q\to \Bbb R$ defines the first order Hamilton
equations on $V^*Q$.

Note that one can associate to any Hamiltonian system on $V^*Q$ an
autonomous symplectic Hamiltonian system on the cotangent bundle
$T^*Q$ such that the corresponding Hamilton equations on $V^*Q$
and $T^*Q$ are equivalent. Moreover, the Hamilton equations on
$V^*Q$ also are equivalent to the Lagrange equations of a certain
first order Lagrangian system on a configuration space $V^*Q$. As
a consequence, Hamiltonian conservation laws can be formulated as
the particular Lagrangian ones.

Lagrangian and Hamiltonian formulations of mechanics fail to be
equivalent, unless a Lagrangian is hyperregular. If a Lagrangian
$L$ on a velocity space $J^1Q$ is hyperregullar, one can associate
to $L$ an unique Hamiltonian form on a phase space $V^*Q$ such
that Lagarange equations on $Q$ and the Hamilton equations $V^*Q$
are equivalent. In general, different Hamiltonian forms are
associated to a non-regular Lagrangian. The comprehensive
relations between Lagrangian and Hamiltonian systems can be
established in the case of almost regular Lagrangians.

\chapter{Dynamic equations}

Equations of motion of non-relativistic mechanics are first and
second order differential equations on manifolds and fibre bundles
over $\Bbb R$. Almost always, they are dynamic equations. Their
solutions are called a motion.

This Chapter is devoted to theory of second order dynamic
equations in mechanics whose configuration space is a fibre bundle
$Q\to\Bbb R$. They are defined as the holonomic connections on the
jet bundle $J^1Q\to\Bbb R$ (Section 1.4). These equations are
represented by connections on the jet bundle $J^1Q\to Q$. Due to
the canonical imbedding $J^1Q\to TQ$ (\ref{z260}), they are proved
equivalent to non-relativistic geodesic equations on the tangent
bundle $TQ$ of $Q$ (Theorem \ref{jp50}).

The notions of a non-relativistic reference frame, a relative
velocity, a free motion equation and a relative acceleration are
formulated in terms of connections on $Q\to\Bbb R$, $J^1Q\to Q$
and $TQ\to Q$.

Generalizing the second Newton law, we introduce the notion of a
Newtonian system (Definition \ref{gn30}) characterized by a mass
tensor. If a mass tensor is non-degenerate, an equation of motion
of a Newtonian system is equivalent to a dynamic equation. The
notion of an external force also is formulated.

\section{Preliminary. Fibre bundles over $\Bbb R$}

This section summarizes some peculiarities of fibre bundles over
$\Bbb R$.

Let
\mar{gm360}\beq
\pi:Q\to \Bbb R \label{gm360}
\eeq
be a fibred manifold whose base is treated as a time axis.
Throughout the book, the time axis $\Bbb R$ is parameterized by
the Cartesian coordinate $t$ with the transition functions
$t'=t+$const. Relative to the Cartesian coordinate $t$, the time
axis $\Bbb R$ is provided with the standard vector field $\dr_t$
and  the standard one-form $dt$ which also is the volume element
on $\Bbb R$.  The symbol $dt$ also stands for any pull-back of the
standard one-form $dt$ onto a fibre bundle over $\Bbb R$.

In order that the dynamics of a mechanical system can be defined
at any instant $t\in\Bbb R$, we further assume that a fibred
manifold $Q\to \Bbb R$ is a fibre bundle with a typical fibre $M$.

\begin{rem} \label{047} \mar{047}
In accordance with Remark \ref{Ehresmann}, a fibred manifold
$Q\to\Bbb R$ is a fibre bundle iff it admits an Ehresmann
connection $\G$, i.e., the horizontal lift $\G\dr_t$ onto $Q$ of
the standard vector field $\dr_t$ on $\Bbb R$ is complete. By
virtue of Theorem \ref{11t3}, any fibre bundle $Q\to \Bbb R$ is
trivial. Its different trivializations
\mar{gm219}\beq
\psi: Q=  \Bbb R\times M \label{gm219}
\eeq
differ from each other in fibrations $Q\to M$.
\end{rem}

Given bundle coordinates $(t,q^i)$ on the fibre bundle $Q\to\Bbb
R$ (\ref{gm360}), the first order jet manifold $J^1Q$ of $Q\to\Bbb
R$  is provided with the adapted coordinates $(t,q^i,q^i_t)$
possessing transition functions (\ref{50}) which read
\mar{100}\beq
q'^i_t=(\dr_t + q^j_t\dr_j)q'^i. \label{100}
\eeq
In mechanics on a configuration space $Q\to\Bbb R$, the jet
manifold $J^1Q$ plays a role of the velocity space.

Note that, if $Q=\Bbb R\times M$ coordinated by $(t,\ol q^i)$,
there is the canonical isomorphism
\mar{gm220}\beq
J^1(\Bbb R\times M)=\Bbb R\times TM, \qquad \ol q^i_t= \dot{\ol
q}^i, \label{gm220}
\eeq
that one can justify by inspection of the transition functions of
the coordinates $\ol q^i_t$ and  $\dot{\ol q}^i$ when transition
functions of $q^i$ are time-independent. Due to the isomorphism
(\ref{gm220}), every trivialization (\ref{gm219}) yields the
corresponding trivialization of the jet manifold
\mar{jp2}\beq
J^1Q= \Bbb R\times TM. \label{jp2}
\eeq

The canonical imbedding (\ref{18}) of $J^1Q$ takes the form
\mar{z260,920}\ben
&& \la_{(1)}: J^1Q\to TQ, \quad
\la_{(1)}: (t,q^i,q^i_t) \to (t,q^i,\dot t=1, \dot q^i=q^i_t), \label{z260}\\
&& \la_{(1)}=d_t=\dr_t +q^i_t\dr_i, \label{z920}
\een
where by $d_t$ is meant the total derivative. From now on, a jet
manifold $J^1Q$ is identified with its image in $TQ$. Using the
morphism (\ref{z260}), one can define the contraction
\ben
&& J^1Q\op\times_Q T^*Q \op\to_Q Q\times \Bbb R, \nonumber \\
&& (q^i_t; \dot{\rm t}, \dot q_i) \to \la_{(1)}\rfloor(\dot{\rm t}dt + \dot
q_idq^i) = \dot{\rm t} + q^i_t\dot q_i, \label{gm280}
\een
where $(t,q^i,\dot{\rm t}, \dot q_i)$ are holonomic coordinates on
the cotangent bundle $T^*Q$.

A glance at the expression (\ref{z260}) shows that the affine jet
bundle $J^1Q\to Q$ is modelled over the vertical tangent bundle
$VQ$ of a fibre bundle $Q\to\Bbb R$.  As a consequence, there is
the following canonical splitting (\ref{48}) of the vertical
tangent bundle $V_QJ^1Q$ of the affine jet bundle $J^1Q\to Q$:
\mar{a1.4}\beq
\al:V_QJ^1Q = J^1Q\op\times_Q VQ, \qquad \al(\dr_i^t)=\dr_i,
\label{a1.4}
\eeq
together with the corresponding splitting of the vertical
cotangent bundle  $V^*_QJ^1Q$ of $J^1Q\to Q$:
\mar{gm382}\beq
\al^*:V^*_QJ^1Q= J^1Q\op\times_Q V^*Q, \qquad \al^*(\ol
dq^i_t)=\ol dq^i, \label{gm382}
\eeq
where $\ol dq^i_t$ and $\ol dq^i$ are the holonomic bases for
$V^*_QJ^1Q$ and $V^*Q$, respectively. Then the exact sequence
(\ref{63a}) of vertical bundles over the composite fibre bundle
\mar{gm361}\beq
J^1Q\ar Q\ar\Bbb R \label{gm361}
\eeq
reads
\be
&& \put(173,-11){$\rule{0.1mm}{4mm}$}
\put(50,0){\vector(0,-1){10}$\rule{18mm}{0.1mm}\,{}_{\al^{-1}} \,
\rule{18mm}{0.1mm}$} \\
&& 0\ar V_QJ^1Q\ar^i VJ^1Q \ar^{\pi_V} J^1Q\op\times_Q VQ\ar 0.
\ee
Hence, we obtain the following linear endomorphism over $J^1Q$ of
the vertical tangent bundle $VJ^1Q$ of the jet bundle $J^1Q\to\Bbb
R$:
\mar{a200}\ben
&& \wh v\op= i\circ\al^{-1}\circ\pi_V: VJ^1Q\to VJ^1Q, \label{a200}\\
&& \wh v(\dr_i)=\dr^t_i, \qquad \wh v(\dr^t_i)=0. \nonumber
\een
This endomorphism obeys the nilpotency rule $\wh v\circ \wh v=0$.

Combining the canonical horizontal splitting (\ref{48}), the
corresponding epimorphism
\be
&& \pr_2:J^1Q\op\times_Q TQ \to J^1Q\op\times_Q VQ= V_QJ^1Q,\\
&& \dr_t\to -q^i_t\dr^t_i,
\qquad \dr_i \to\dr_i^t,
\ee
and the monomorphism $VJ^1Q\to TJ^1Q$, one can extend the
endomorphism (\ref{a200}) to the tangent bundle $TJ^1Q$:
\mar{a1.7}\ben
&& \wh v: TJ^1Q\to TJ^1Q, \nonumber \\
&& \wh v(\dr_t) = -q^i_t\dr_i^t,  \qquad \wh v(\dr_i)=\dr^t_i, \qquad
\wh v(\dr^t_i)=0. \label{a1.7}
\een
This is called the vertical endomorphism. It inherits the
nilpotency property. The transpose of the vertical endomorphism
$\wh v$ (\ref{a1.7}) is
\mar{gm441}\ben
&& \wh v^*:T^*J^1Q\to T^*J^1Q, \nonumber \\
&& \wh v^*(dt)=0,\qquad \wh v^*(dq^i)=0, \qquad
\wh v^*(dq^i_t) =\thh^i, \label{gm441}
\een
where $\thh^i=dq^i-q^i_tdt$  are the contact forms (\ref{24}). The
nilpotency rule $\wh v^*\circ \wh v^*=0$ also is fulfilled. The
homomorphisms $\wh v$ and $\wh v^*$ are associated with the
tangent-valued one-form $\wh v= \thh^i\ot \dr_i^t$ in accordance
with the relations (\ref{29b}) -- (\ref{29b'}).

In view of the morphism $\la_{(1)}$ (\ref{z260}), any connection
\mar{z270}\beq
\G=dt\ot (\dr_t +\G^i\dr_i) \label{z270}
\eeq
on a fibre bundle $Q\to\Bbb R$  can be identified with a nowhere
vanishing  horizontal vector field
\mar{a1.10}\beq
\G = \dr_t + \G^i \dr_i \label{a1.10}
\eeq
on $Q$ which is the horizontal lift $\G\dr_t$ (\ref{b1.85}) of the
standard vector field $\dr_t$ on $\Bbb R$ by means of the
connection (\ref{z270}). Conversely, any vector field $\G$ on $Q$
such that $dt\rfloor\G =1$ defines a connection on $Q\to\Bbb R$.
Therefore, the connections (\ref{z270}) further are identified
with the vector fields (\ref{a1.10}). The integral curves of the
vector field (\ref{a1.10}) coincide with the integral sections for
the connection (\ref{z270}).

Connections on a fibre bundle $Q\to\Bbb R$ constitute an affine
space modelled over the vector space of vertical vector fields on
$Q\to\Bbb R$. Accordingly, the covariant differential
(\ref{2116}), associated with a connection $\G$ on $Q\to\Bbb R$,
takes its values into the vertical tangent bundle $VQ$ of
$Q\to\Bbb R$:
\mar{z279}\beq
D^\G: J^1Q\op\to_Q VQ, \qquad \dot q^i\circ D^\G =q^i_t-\G^i.
\label{z279}
\eeq

A connection $\G$ on a fibre bundle $Q\to\Bbb R$ is obviously
flat. It yields a horizontal distribution on $Q$. The integral
manifolds of this distribution are integral curves  of the vector
field (\ref{a1.10}) which are transversal to fibres of a fibre
bundle $Q\to\Bbb R$.

\begin{theo} \label{gn1} \mar{gn1}
By virtue of Theorem \ref{gena113}, every connection $\G$ on a
fibre bundle $Q\to\Bbb R$ defines an atlas of local constant
trivializations of $Q\to\Bbb R$ such that the associated bundle
coordinates $(t,q^i)$ on $Q$ possess the transition functions
$q^i\to q'^i(q^j)$ independent of $t$, and
\mar{z271}\beq
\G=\dr_t \label{z271}
\eeq
with respect to these coordinates. Conversely, every atlas of
local constant trivializations of the fibre bundle $Q\to\Bbb R$
determines a connection on  $Q\to\Bbb R$ which is equal to
(\ref{z271}) relative to this atlas.
\end{theo}

A connection $\G$ on a fibre bundle $Q\to \Bbb R$ is said to be
complete  if the horizontal vector field (\ref{a1.10}) is
complete. In accordance with Remark \ref{Ehresmann}, a connection
on a fibre bundle $Q\to \Bbb R$ is complete iff it is an Ehresmann
connection. The following holds \cite{book10,book98}.

\begin{theo}\label{compl} \mar{compl}
Every trivialization of a fibre bundle $Q\to \Bbb R$ yields a
complete connection on this fibre bundle. Conversely, every
complete connection $\G$ on $Q\to\Bbb R$ defines its
trivialization (\ref{gm219}) such that the horizontal vector field
(\ref{a1.10}) equals $\dr_t$ relative to the bundle coordinates
associated with this trivialization.
\end{theo}

Let $J^1J^1Q$ be the repeated jet manifold of a fibre bundle
$Q\to\Bbb R$ provided with the adapted coordinates $(t,q^i,q^i_t,
\wh q^i_t,q^i_{tt})$  possessing transition functions
\be
&& q'^i_t = d_t q'^i,
\qquad \wh q'^i_t = \wh d_t q'^i, \qquad
q'^i_{tt}= \wh d_t q'^i_t,\\
&& d_t = \dr_t +q^j_t\dr_j +q^j_{tt}\dr^t_j,
\qquad \wh d_t = \dr_t +\wh q^j_t\dr_j +q^j_{tt}\dr^t_j.
\ee
There  is the canonical isomorphism $k$ between the affine
fibrations $\pi_{11}$ (\ref{gm213}) and $J^1\pi_0^1$ (\ref{gm214})
of $J^1J^1Q$ over $J^1Q$, i.e.,
\be
\pi_{11}\circ k= J^1_0\pi_{01}, \qquad k\circ k= \id J^1J^1Q,
\ee
where
\mar{gm215}\beq
 q^i_t\circ k=\wh q^i_t, \qquad \wh q^i_t\circ k=q^i_t, \qquad
q^i_{tt}\circ k= q^i_{tt}. \label{gm215}
\eeq
In particular, the affine bundle $\pi_{11}$ (\ref{gm213}) is
modelled over the vertical tangent bundle $VJ^1Q$ of $J^1Q\to\Bbb
R$ which is canonically isomorphic to the underlying vector bundle
$J^1VQ\to J^1Q$  of the affine bundle $J^1\pi_0^1$ (\ref{gm214}).

For a fibre bundle $Q\to\Bbb R$, the sesquiholonomic jet manifold
$\wh J^2Q$ coincides with the second order jet manifold $J^2Q$
coordinated by $(t,q^i,q^i_t,q^i_{tt})$, possessing transition
functions
\mar{106}\beq
q'^i_t = d_t q'^i, \qquad q'^i_{tt}=  d_t q'^i_t. \label{106}
\eeq
The affine bundle $J^2Q\to J^1Q$ is modelled over the vertical
tangent bundle
\be
V_QJ^1Q= J^1Q\op\times_QVQ\to J^1Q
\ee
of the affine jet bundle $J^1Q\to Q$. There are the imbeddings
\mar{gm211,cqg80}\ben
&& J^2Q \ar^{\la_{(2)}} TJ^1Q \ar^{T\la}
V_QTQ= T^2Q\subset TTQ,
\nonumber \\
&& \la_{(2)}: (t,q^i,q^i_t,q^i_{tt})\to (t,q^i,q^i_t,\dot t=1,\dot
q^i=q^i_t,\dot q^i_t=q^i_{tt}), \label{gm211}\\
&& T\la_{(1)}\circ \la_{(2)}:(t,q^i,q^i_t,q^i_{tt})\to \label{cqg80}\\
&&\qquad  (t,q^i, \dot t =\dv t=1, \dot q^i =\dv q^i=q^i_t, \ddot t=0,
\ddot q^i=q^i_{tt}), \nonumber
\een
where:  (i) $(t,q^i,\dot t,\dot q^i,\dv t,\dv q^i, \ddot t, \ddot
q^i)$ are the coordinates on the double tangent bundle $TTQ$,(ii)
by $V_QTQ$ is meant the vertical tangent bundle of $TQ\to Q$, and
(iii) $T^2Q\subset TTQ$ is the second order tangent space given by
the coordinate relation $\dot t =\dv t$.

Due to the morphism (\ref{gm211}), any connection $\xi$ on the jet
bundle $J^1Q\to \Bbb R$ (defined as a section of the affine bundle
$\pi_{11}$ (\ref{gm213})) is represented by a horizontal vector
field on $J^1Q$ such that $\xi\rfloor dt=1$.

A connection $\G$ (\ref{a1.10}) on a fibre bundle $Q\to\Bbb R$ has
the jet prolongation to the section $J^1\G$ of the affine bundle
$J^1\pi_0^1$. By virtue of the isomorphism $k$ (\ref{gm215}),
every connection $\G$ on $Q\to\Bbb R$ gives rise to the connection
\mar{gm217'}\beq
J\G\op= k\circ J^1\G: J^1Q\to J^1J^1Q,\qquad J\G=\dr_t +\G^i\dr_i
+ d_t\G^i\dr^t_i, \label{gm217'}
\eeq
on the jet bundle $J^1Q\to \Bbb R$.

A connection on the jet bundle $J^1Q\to \Bbb R$ is said to be
holonomic if it is a section
\be
\xi=dt\ot(\dr_t + q^i_t\dr_i + \xi^i \dr_i^t)
\ee
of the holonomic subbundle $J^2Q\to J^1Q$ of $J^1J^1Q\to J^1Q$. In
view of the morphism (\ref{gm211}), a holonomic connection is
represented by a horizontal vector field
\mar{a1.30}\beq
\xi=\dr_t + q^i_t\dr_i + \xi^i \dr_i^t \label{a1.30}
\eeq
on $J^1Q$. Conversely, every vector field $\xi$ on $J^1Q$ such
that
\be
dt\rfloor\xi=1, \qquad \wh v(\xi)=0,
\ee
where $\wh v$ is the vertical endomorphism (\ref{a1.7}), is a
holonomic connection on the jet bundle $J^1Q\to\Bbb R$.

Holonomic connections (\ref{a1.30}) make up an affine space
modelled over the linear space of vertical vector fields on the
affine jet bundle $J^1Q\to Q$, i.e., which live in $V_QJ^1Q$.

A holonomic connection $\xi$ defines the corresponding covariant
differential (\ref{z279}) on the jet manifold $J^1Q$:
\be
&& D^\xi: J^1J^1Q\ar_{J^1Q} V_QJ^1Q\subset VJ^1Q, \\
&& \dot q^i \circ D^\xi =0, \qquad \dot q^i_t\circ D^\xi= q^i_{tt} - \xi^i,
\ee
which takes its values into the vertical tangent bundle $V_QJ^1Q$
of the jet bundle $J^1Q\to Q$.  Then by virtue of Theorem
\ref{1.5.3}, any integral section $\ol c: ()\to J^1Q$ for a
holonomic connection $\xi$ is holonomic, i.e., $\ol c=\dot c$
where $c$ is a curve in $Q$.

\section{Autonomous dynamic equations}

Let us start with dynamic equations on a manifold. From the
physical viewpoint, they are treated as autonomous dynamic
equations in autonomous mechanics.

Let $Z$, $\di Z>1$, be a smooth manifold coordinated by $(z^\la)$.

\begin{defi}\label{gena70} \mar{gena70}
Let $u$ be a vector field $u$ on $Z$. A closed subbundle $u(Z)$ of
the tangent bundle $TZ$ given by the coordinate relations
\mar{gm150}\beq
\dot z^\la=u^\la(z) \label{gm150}
\eeq
is said to be an autonomous first order dynamic equation on a
manifold $Z$. This is a system of first order differential
equations on a fibre bundle $\Bbb R\times Z\to\Bbb R$ in
accordance with Definition \ref{equa}.
\end{defi}

By a solution of the autonomous first order dynamic equation
(\ref{gm150}) is meant an integral curve of the vector field $u$.

\begin{defi}\label{gena72} \mar{gena72}
An autonomous  second order dynamic equation on a manifold $Z$ is
defined as a first order dynamic equation on the tangent bundle
$TZ$ which is associated with a holonomic vector field
\mar{gm200}\beq
\Xi = \dot z^\la\dr_\la +\Xi^\la(z^\m,\dot z^\m)\dot \dr_\la
\label{gm200}
\eeq
on $TZ$. This vector field, by definition, obeys the condition
$J(\Xi)=u_{TZ}$, where $J$ is the endomorphism (\ref{z117}) and
$u_{TZ}$ is the Liouville vector field (\ref{z112}) on $TZ$.
\end{defi}

The holonomic vector field (\ref{gm200}) also is called the
autonomous second order dynamic equation.

Let the double tangent bundle $TTZ$ be provided with coordinates
$(z^\la,\dot z^\la, \dv z^\la,\ddot z^\la)$. With respect to these
coordinates, the autonomous second order dynamic equation defined
by the holonomic vector field $\Xi$ (\ref{gm200}) reads
\mar{gm202}\beq
\dv z^\la=\dot z^\la, \qquad  \ddot z^\la =\Xi^\la(z^\m,\dot
z^\m). \label{gm202}
\eeq
By a solution  of the second order dynamic equation (\ref{gm202})
is meant a curve $c:(,)\to Z$ in a manifold $Z$ whose tangent
prolongation $\dot c:(,)\to TZ$ is an integral curve of the
holonomic vector field $\Xi$ or, equivalently, whose second order
tangent prolongation $\ddot c$ lives in the subbundle
(\ref{gm202}). It satisfies an autonomous second order
differential equation
\be
\ddot c^\la(t) = \Xi^\la(c^\m(t),\dot c^\m(t)).
\ee

Second order dynamic equations on a manifold $Z$ are exemplified
by geodesic equations on the tangent bundle $TZ$.

Given a connection
\mar{109}\beq
K= dz^\m\ot (\dr_\m +K^\nu_\m\dot\dr_\nu) \label{109}
\eeq
on the tangent bundle $TZ\to Z$, let
\mar{jp21}\beq
\wh K: TZ\op\times_Z TZ\to TTZ \label{jp21}
\eeq
be the corresponding linear bundle morphism over $TZ$ which splits
the exact sequence
\be
0\ar V_ZTZ\ar TTZ\ar TZ\op\times_Z TZ\ar 0.
\ee
Note that, in contrast with $K$ (\ref{B}), the connection $K$
(\ref{109}) need not be linear.

\begin{defi}
A geodesic equation on  $TZ$ with respect to the connection $K$
(\ref{109}) is defined as the range
\mar{jp20}\beq
\dv z^\la=\dot z^\la, \qquad  \ddot z^\m = K^\m_\nu\dot z^\nu
\label{jp20}
\eeq
of the morphism (\ref{jp21}) restricted to the diagonal $TZ\subset
TZ\times TZ$.
\end{defi}

By a solution of a geodesic equation on $TZ$ is meant a geodesic
curve  $c$  in $Z$ whose tangent prolongation $\dot c$ is an
integral section (a geodesic vector field) over $c\subset Z$ for a
connection $K$.

It is readily observed that the morphism $\wh K|_{TZ}$ is a
holonomic vector field on $TZ$. It follows that any geodesic
equation (\ref{jp21}) on $TZ$ is a second order equation on $Z$.
The converse is not true in general. Nevertheless, there is the
following theorem \cite{mora}.

\begin{theo} \label{jp30} \mar{jp30} Every second order
dynamic equation (\ref{gm202}) on a manifold $Z$ defines a
connection $K_\Xi$ on the tangent bundle $TZ\to Z$ whose
components are
\mar{jp32}\beq
K^\m_\nu = \frac12\dot\dr_\nu \Xi^\m. \label{jp32}
\eeq
\end{theo}

However, the autonomous second order dynamic equation
(\ref{gm202}) fails to be a geodesic equation with respect to the
connection (\ref{jp32}) in general. In particular, the geodesic
equation (\ref{jp20}) with respect to a connection $K$ determines
the connection (\ref{jp32}) on $TZ\to Z$ which does not
necessarily coincide with $K$.

\begin{theo} \label{y1} \mar{y1}
A second order equation $\Xi$ on $Z$ is a geodesic equation for
the connection (\ref{jp32}) iff $\Xi$ is a spray, i.e.,
$[u_{TZ},\Xi]=\Xi$, where $u_{TZ}$ is the Liouville vector field
(\ref{z112}) on $TZ$, i.e.,
\be
 \Xi^i=a_{ij}(q^k)\dot q^i\dot q^j
\ee
and the connection $K$ (\ref{jp32}) is linear.
\end{theo}

\section{Dynamic equations}

Let $Q\to X$ (\ref{gm360}) be a configuration space of
non-relativistic mechanics. Refereing to Definition \ref{equa} of
a differential equation on a fibre bundle, one defines a dynamic
equation on $Q\to\Bbb R$ as a differential equation which is
algebraically solved for the highest order derivatives.

\begin{defi} \label{fodeq2} \mar{fodeq2} Let $\G$ (\ref{a1.10}) be a
connection on a fibre bundle $Y\to \Bbb R$. The corresponding
covariant differential $D^\G$ (\ref{z279}) is a first order
differential operator on $Y$. Its kernel, given by the coordinate
equation
\mar{ggm155}\beq
q_t^i= \G^i(t,q^i), \label{ggm155}
\eeq
is a closed subbundle of the jet bundle $J^1Y\to\Bbb R$. By virtue
of Definition \ref{equa}, it is a first order differential
equation on a fibre bundle $Y\to \Bbb R$ called the  first order
dynamic equation on $Y\to\Bbb R$.
\end{defi}

Due to the canonical imbedding $J^1Q\to TQ$ (\ref{z260}), the
equation (\ref{ggm155}) is equivalent to the autonomous first
order dynamic equation
\mar{089}\beq
\dot t=1, \qquad \dot q^i=\G^i(t,q^i) \label{089}
\eeq
on a manifold $Y$ (Definition \ref{gena72}). It is defined by the
vector field (\ref{a1.10}). Solutions of the first order dynamic
equation (\ref{ggm155}) are integral sections for a connection
$\G$.

\begin{defi} \label{gena73} \mar{gena73}
Let us consider  the first order dynamic equation (\ref{ggm155})
on the jet bundle $J^1Q\to \Bbb R$, which is associated with  a
holonomic connection $\xi$ (\ref{a1.30}) on $J^1Q\to \Bbb R$. This
is a closed subbundle of the second order jet bundle $J^2Q\to \Bbb
R$ given by the coordinate relations
\mar{z273}\beq
q^i_{tt}=\xi^i(t,q^j,q^j_t). \label{z273}
\eeq
Consequently, it is a second order differential equation on a
fibre bundle $Q\to\Bbb R$ in accordance with Definition
\ref{equa}. This equation is called a second order dynamic
equation or, simply, a dynamic equation if there is no danger of
confusion. The corresponding horizontal vector field $\xi$
(\ref{a1.30}) also is termed a dynamic equation.
\end{defi}

The second order dynamic equation (\ref{z273}) possesses the
coordinate transformation law
\beq
q'^i_{tt} = \xi'^i, \qquad \xi'^i=(\xi^j\dr_j +
q^j_tq^k_t\dr_j\dr_k +2q^j_t\dr_j\dr_t +\dr_t^2)q'^i(t,q^j),
\label{z317}
\eeq
derived from the formula (\ref{106}).

A solution of the dynamic equation (\ref{z273}) is a curve $c$ in
$Q$ whose second order jet prolongation $\ddot c$ lives in
(\ref{z273}). Any integral section $\ol c$ for the holonomic
connection $\xi$ obviously is the jet prolongation $\dot c$ of a
solution $c$ of the dynamic equation (\ref{z273}), i.e.,
\mar{z274}\beq
\ddot c^i=\xi^i\circ \dot c, \label{z274}
\eeq
and {\it vice versa}.

\begin{rem}
By very definition, the second order dynamic equation (\ref{z273})
on a fibre bundle $Q\to\Bbb R$ is equivalent to the system  of
first order differential equations
\mar{z344}\beq
 \wh q^i_t=q^i_t, \qquad q^i_{tt} =\xi^i(t,q^j,q^j_t), \label{z344}
\eeq
on the jet bundle $J^1Q\to\Bbb R$. Any solution $\ol c$ of these
equations takes its values into $J^2Q$ and, by virtue of Theorem
\ref{1.5.3}, is holonomic, i.e., $\ol c=\dot c$. The equations
(\ref{z273}) and (\ref{z344}) are therefore equivalent.
\end{rem}

A dynamic equation $\xi$ on a fibre bundle $Q\to \Bbb R$ is said
to be conservative if there exist a trivialization (\ref{gm219})
of $Q$ and the corresponding trivialization (\ref{jp2}) of $J^1Q$
such that the vector field $\xi$ (\ref{a1.30}) on $J^1Q$ is
projectable over $M$. Then this projection
\be
\Xi_\xi=\dot q^i\dr_i +\xi^i(q^j,\dot q^j)\dot \dr_i
\ee
is an autonomous second order dynamic equation on the typical
fibre $M$ of $Q\to\Bbb R$ in accordance with Definition
\ref{gena72}. Conversely, every autonomous second order dynamic
equation $\Xi$ (\ref{gm200}) on a manifold $M$ can be seen as a
conservative dynamic equation
\mar{jp27}\beq
\xi_\Xi=\dr_t + \dot q^i\dr_i + \Xi^i\dot \dr_i \label{jp27}
\eeq
on the fibre bundle $\Bbb R\times M\to\Bbb R$ in accordance with
the isomorphism (\ref{jp2}).

The following theorem holds \cite{book10,book98}.

\begin{theo}\label{jp28} \mar{jp28}
Any dynamic equation $\xi$ (\ref{z273}) on a fibre bundle
$Q\to\Bbb R$ is equivalent to an autonomous second order dynamic
equation $\Xi$ on a manifold $Q$ which makes the diagram
\be
\begin{array}{rcccl}
& J^2Q & \ar & T^2Q & \\
_\xi &  \put(0,-10){\vector(0,1){20}} & &
\put(0,-10){\vector(0,1){20}}
& _\Xi\\
& J^1Q &\ar^{\la_{(1)}} & TQ &
\end{array}
\ee
commutative and obeys the relations
\be
\xi^i=\Xi^i(t,q^j,\dot t=1, \dot q^j=q^j_t), \qquad \Xi^t=0.
\ee
Accordingly, the dynamic equation (\ref{z273}) is written in the
form
\be
q^i_{tt}= \Xi^i\mid_{\dot t=1, \dot q^j=q^j_t},
\ee
which is equivalent to the autonomous second order dynamic
equation
\mar{jp35}\beq
\ddot t=0, \qquad \dot t=1, \qquad \ddot q^i= \Xi^i, \label{jp35}
\eeq
on $Q$.
\end{theo}

\section{Dynamic connections}

In order to say something more, let us consider the relationship
between the holonomic connections on the jet bundle $J^1Q\to\Bbb
R$ and the connections on the affine jet bundle $J^1Q\to Q$ (see
Propositions \ref{gena51} and \ref{gena52} below).

By $J^1_QJ^1Q$ throughout is meant the first order jet manifold of
the affine jet bundle $J^1Q\to Q$. The adapted coordinates on
$J^1_QJ^1Q$ are $(q^\la,q^i_t,q^i_{\la t})$, where we use the
compact notation $\la=(0,i)$, $q^0=t$. Let
\be
\g:J^1Q\to J^1_QJ^1Q
\ee
be a connection on the affine jet bundle $J^1Q\to Q$. It takes the
coordinate form
\mar{a1.38}\beq
 \g=dq^\la\ot (\dr_\la + \g^i_\la \dr_i^t),
\label{a1.38}
\eeq
together with the coordinate transformation law
\mar{m175}\beq
 \g'^i_\la = (\dr_jq'^i\g^j_\m
+\dr_\m q'^i_t)\frac{\dr q^\m}{\dr q'^\la}. \label{m175}
\eeq

\begin{rem}
In view of the canonical splitting (\ref{a1.4}), the curvature
(\ref{161'}) of the connection $\g$ (\ref{a1.38}) reads
\ben
&&R: J^1Q\to \op\w^2 T^*Q\op\ot_{J^1Q}VQ, \nonumber \\
&&R =\frac12 R^i_{\la\m} dq^\la\w dq^\m\otimes\dr_i =\left(\frac12 R^i_{kj}
dq^k\w dq^j +R^i_{0j}dt\w dq^j\right)\ot\dr_i,
\nonumber \\
&& R^i_{\la\m}=\dr_\la\g^i_\m -\dr_\m\g^i_\la +\g^j_\la\dr_j\g^i_\m
-\g^j_\m\dr_j\g^i_\la. \label{gm282}
\een
Using the contraction (\ref{gm280}), we obtain the soldering form
\be
\la_{(1)}\rfloor R= [(R^i_{kj}q^k_t + R^i_{0j})dq^j
-R^i_{0j}q^j_tdt]\ot\dr_i
\ee
on the affine jet bundle $J^1Q\to Q$. Its image by the canonical
projection $T^*Q\to V^*Q$ (\ref{b418'}) is the tensor field
\mar{gm283}\beq
\ol R: J^1Q\to V^*Q\op\ot_Q VQ,\qquad \ol R= (R^i_{kj}q^k_t +
R^i_{0j})\ol dq^j\ot\dr_i, \label{gm283}
\eeq
and then we come to the scalar field
\mar{gm284}\beq
\wt R: J^1Q \to \Bbb R,\qquad \wt R = R^i_{ki}q^k_t + R^i_{0i},
\label{gm284}
\eeq
on the jet manifold $J^1Q$.
\end{rem}

\begin{prop}\label{gena51} \mar{gena51}
Any connection $\g$ (\ref{a1.38}) on the affine jet bundle
$J^1Q\to Q$ defines the holonomic connection
\ben
&& \xi_\g=\rho\circ \g: J^1Q \to J^1_QJ^1Q
\to J^2Q, \label{z281}\\
&& \xi_\g = \dr_t + q^i_t\dr_i +(\g^i_0 +q^j_t\g^i_j)\dr_i^t, \nonumber
\een
 on the jet bundle
$J^1Q\to \Bbb R$.
\end{prop}

\begin{proof}
Let us consider the composite fibre bundle  (\ref{gm361}) and the
morphism $\rho$ (\ref{1.38}) which reads
\beq
\rho:  J^1_QJ^1Q \ni (q^\la,q^i_t,q^i_{\la t}) \mapsto
(q^\la,q^i_t,\wh q^i_t=q^i_t,q^i_{tt}=q^i_{0t} +q^j_tq^i_{jt})\in
J^2Q. \label{z298}
\eeq
A connection $\g$ (\ref{a1.38}) and the morphism $\rho$
(\ref{z298}) combine into the desired holonomic connection
$\xi_\g$ (\ref{z281}) on the jet bundle $J^1Q\to\Bbb R$.
\end{proof}

It follows that every connection $\g$ (\ref{a1.38}) on the affine
jet bundle $J^1Q\to Q$ yields the dynamic equation
\mar{z287}\beq
q^i_{tt}=\g^i_0 +q^j_t\g^i_j \label{z287}
\eeq
on the configuration bundle $Q\to\Bbb R$. This is precisely the
restriction to $J^2Q$ of the kernel $\Ker \wt D^\g$ of the
vertical covariant differential $\wt D^\g$ (\ref{7.10}) defined by
the connection $\g$:
\mar{gm388}\beq
\wt D^\g: J^1J^1Q\to V_QJ^1Q, \qquad \dot q^i_t\circ\wt D^\g=
q^i_{tt} -\g^i_0 - q^j_t\g^i_j. \label{gm388}
\eeq
Therefore, connections on the jet bundle $J^1Q\to Q$ are called
the dynamic connections. The corresponding equation (\ref{z274})
can be written in the form
\be
\ddot c^i=\rho\circ \g\circ \dot c,
\ee
where $\rho$ is the morphism (\ref{z298}).

Of course, different dynamic connections can lead to the same
dynamic equation (\ref{z287}).

\begin{prop}\label{gena52}
Any holonomic connection $\xi$ (\ref{a1.30}) on the jet bundle
$J^1Q\to \Bbb R$ yields the dynamic connection
\mar{z286}\beq
\g_\xi =dt\ot\left[\dr_t+(\xi^i-\frac12
q^j_t\dr_j^t\xi^i)\dr_i^t\right] + dq^j\ot\left[\dr_j
+\frac12\dr_j^t\xi^i \dr_i^t\right] \label{z286}
\eeq
 on the affine
jet bundle $J^1Q\to Q$ \cite{book10,book98}.
\end{prop}

It is readily observed that the dynamic connection $\g_\xi$
(\ref{z286}), defined by a dynamic equation,  possesses the
property
\mar{a1.69}\beq
\g^k_i = \dr_i^t\g^k_0 +  q^j_t\dr_i^t\g^k_j, \label{a1.69}
\eeq
which implies the relation $\dr_j^t\g^k_i = \dr_i^t\g^k_j$.
Therefore, a dynamic connection $\g$, obeying the condition
(\ref{a1.69}), is said to be symmetric. The torsion of a dynamic
connection $\g$ is defined as the tensor field
\mar{gm273}\ben
&&T: J^1Q\to V^*Q\op\ot_Q VQ, \nonumber\\
&& T=T^k_i \ol dq^i\ot\dr_k, \qquad
T^k_i = \g^k_i - \dr_i^t\g^k_0 - q^j_t\dr_i^t\g^k_j. \label{gm273}
\een
It follows at once that a dynamic connection is symmetric iff its
torsion vanishes.

Let $\g$ be the dynamic connection (\ref{a1.38}) and $\xi_\g$ the
corresponding dynamic equation (\ref{z281}). Then the dynamic
connection (\ref{z286}) associated with the dynamic equation
$\xi_\g$ takes the form
\be
\g_{\xi_\g}{}^k_i = \frac{1}{2} (\g^k_i + \dr_i^t\g^k_0 +
q^j_t\dr_i^t\g^k_j), \qquad \g_{\xi_\g}{}^k_0 = \g^k_0+
q^j_t\g^k_j - q^i_t\g_{\xi_\g}{}^k_i.
\ee
It is readily observed that $\g = \g_{\xi_\g}$ iff the torsion $T$
(\ref{gm273}) of the dynamic connection $\g$ vanishes.

\begin{ex}
Since a jet bundle $J^1Q\to Q$ is affine, it admits an affine
connection
\beq
 \g=dq^\la\ot [\dr_\la + (\g^i_{\la 0}(q^\m)+ \g^i_{\la j}(q^\m)q^j_t)\dr_i^t].
\label{z299}
\eeq
This connection is symmetric iff $\g^i_{\la \m}=\g^i_{\m\la}$. One
can easily justify that an affine dynamic connection generates a
quadratic dynamic equation, and {\it vice versa}. Nevertheless, a
non-affine dynamic connection, whose symmetric part is affine,
also defines a quadratic dynamic equation.
\end{ex}

Using the notion of a dynamic connection, we can modify Theorem
\ref{jp30} as follows. Let $\Xi$ be an autonomous  second order
dynamic equation on a manifold $M$, and let $\xi_\Xi$ (\ref{jp27})
be the corresponding conservative dynamic equation on the bundle
$\Bbb R\times M\to\Bbb R$. The latter yields the dynamic
connection $\g$ (\ref{z286}) on a fibre bundle
\be
\Bbb R\times TM\to \Bbb R\times M.
\ee
Its components $\g^i_j$ are exactly those of the connection
(\ref{jp32}) on the tangent bundle $TM\to M$ in Theorem
\ref{jp30}, while $\g^i_0$ make up a vertical vector field
\mar{jp38}\beq
e=\g^i_0\dot \dr_i = \left(\Xi^i -\frac12 \dot q^j\dot\dr_j
\Xi^i\right)\dot \dr_i \label{jp38}
\eeq
on $TM\to M$. Thus, we have shown the following.

\begin{prop}\label{gn23} \mar{gn23}
Every autonomous second order dynamic equation $\Xi$ (\ref{gm202})
on a manifold $M$ admits the decomposition
\be
\Xi^i= K^i_j\dot q^j + e^i
\ee
where $K$ is the connection (\ref{jp32}) on the tangent bundle
$TM\to M$, and $e$ is the vertical vector field (\ref{jp38}) on
$TM\to M$.
\end{prop}

\section{Non-relativistic geodesic equations}

In this Section, we aim to show that every dynamic equation on a
configuration bundle $Q\to \Bbb R$ is equivalent to a geodesic
equation on the tangent bundle $TQ\to Q$.

We start with the relation between the dynamic connections $\g$ on
the affine jet bundle $J^1Q\to Q$ and the connections
\mar{z290}\beq
K= dq^\la \ot (\dr_\la +K^\m_\la\dot\dr_\m) \label{z290}
\eeq
on the tangent bundle $TQ\to Q$ of the configuration space $Q$.
Note that they need not be linear. We follow the compact notation
(\ref{vvv}).

Let us consider the diagram
\beq
\begin{array}{rcccl}
& J^1_QJ^1Q & \ar^{J^1\la_{(1)}} & J^1_QTQ & \\
_\g &  \put(0,-10){\vector(0,1){20}} & &
\put(0,-10){\vector(0,1){20}}
& _K\\
& J^1Q &\ar^{\la_{(1)}} & TQ &
\end{array} \label{z291}
\eeq
where $J^1_QTQ$ is the first order jet manifold of the tangent
bundle $TQ\to Q$, coordinated by
\be
(t,q^i,\dot t,\dot q^i, (\dot t)_\m, (\dot q^i)_\m).
\ee
The jet prolongation over $Q$ of the canonical imbedding
$\la_{(1)}$ (\ref{z260}) reads
\be
J^1\la_{(1)}: (t,q^i,q^i_t, q^i_{\m t}) \mapsto (t,q^i,\dot
t=1,\dot q^i=q^i_t, (\dot t)_\m=0, (\dot q^i)_\m=q^i_{\m t}).
\ee
Then we have
\be
&& J^1\la_{(1)}\circ \g: (t,q^i,q^i_t) \mapsto
(t,q^i,\dot t=1,\dot q^i=q^i_t, (\dot t)_\m=0,
(\dot q^i)_\m=\g^i_\m ),\\
&& K\circ \la_{(1)}: (t,q^i,q^i_t) \mapsto
(t,q^i,\dot t=1,\dot q^i=q^i_t, (\dot t)_\m=K_\m^0, (\dot
q^i)_\m=K^i_\m).
\ee
It follows that the diagram (\ref{z291}) can be commutative only
if the components $K^0_\m$ of the connection $K$ (\ref{z290}) on
the tangent bundle $TQ\to Q$ vanish.

Since the transition functions $t\to t'$ are independent of $q^i$,
a connection
\mar{z292}\beq
\wt K= dq^\la\ot (\dr_\la +K^i_\la\dot\dr_i) \label{z292}
\eeq
with $K^0_\m=0$ may exist on the tangent bundle $TQ\to Q$ in
accordance with the transformation law
\mar{z293}\beq
{K'}_\la^i=(\dr_j q'^i K^j_\m + \dr_\m\dot q'^i) \frac{\dr
q^\m}{\dr q'^\la}.  \label{z293}
\eeq
Now the diagram (\ref{z291}) becomes commutative if the
connections $\g$ and $\wt K$ fulfill the relation
\mar{z294}\beq
\g^i_\m= K^i_\m\circ \la_{(1)}=K^i_\m(t,q^i,\dot t=1, \dot
q^i=q^i_t). \label{z294}
\eeq
It is easily seen that this relation holds globally because the
substitution of $\dot q^i=q^i_t$ in (\ref{z293}) restates the
transformation law (\ref{m175}) of a connection on the affine jet
bundle $J^1Q\to Q$. In accordance with the relation (\ref{z294}),
the desired connection $\wt K$ is an extension  of the section
$J^1\la\circ \g$ of the affine jet bundle $J^1_QTQ\to TQ$ over the
closed submanifold $J^1Q\subset TQ$ to a global section. Such an
extension always exists by virtue of Theorem \ref{mos9}, but it is
not unique. Thus, we have proved the following.

\begin{prop}\label{motion1}
In accordance with the relation (\ref{z294}), every dynamic
equation on a configuration bundle $Q\to\Bbb R$ can be written in
the form
\beq
q^i_{tt} = K^i_0\circ\la_{(1)} +q^j_t K^i_j\circ\la_{(1)},
\label{gm340}
\eeq
where $\wt K$ is the connection (\ref{z292}) on the tangent bundle
$TQ\to Q$. Conversely,
 each connection $\wt K$ (\ref{z292}) on $TQ\to Q$
defines the dynamic connection $\g$ (\ref{z294}) on the affine jet
bundle $J^1Q\to Q$ and the dynamic equation (\ref{gm340}) on a
configuration bundle  $Q\to\Bbb R$.
\end{prop}

Then we come to the following theorem.

\begin{theo} \label{jp50} \mar{jp50}
Every dynamic equation (\ref{z273}) on a configuration bundle
$Q\to\Bbb R$ is equivalent to the geodesic equation
\ben
&&\ddot q^0=0, \qquad \dot q^0=1,\nonumber\\
&& \ddot q^i = K^i_\la(q^\m,\dot q^\m)\dot q^\la, \label{cqg11}
\een
on the tangent bundle $TQ$ relative to a connection $\wt K$ with
the components $K^0_\la=0$ and $K^i_\la$ (\ref{z294}). Its
solution is a geodesic curve in $Q$ which also obeys the dynamic
equation (\ref{gm340}), and {\it vice versa}.
\end{theo}

In accordance with this theorem,  the autonomous second order
equation (\ref{jp35}) in Theorem \ref{jp28} can be chosen as a
geodesic equation. It should be emphasized that, written in the
bundle coordinates $(t,q^i)$, the geodesic equation (\ref{cqg11})
and the connection $\wt K$ (\ref{z294}) are well defined with
respect to any coordinates on $Q$.

From the physical viewpoint, the most relevant dynamic equations
are the quadratic ones
\beq
\xi^i = a^i_{jk}(q^\m)q^j_t q^k_t + b^i_j(q^\m)q^j_t + f^i(q^\m).
\label{cqg12}
\eeq
This property is global due to the transformation law
(\ref{z317}). Then one can use the following two facts.

\begin{prop}\label{aff} \mar{aff}
There is one-to-one correspondence between the affine connections
$\g$ on the affine jet bundle $J^1Q\to Q$ and the linear
connections $K$ (\ref{z292}) on the tangent bundle $TQ\to Q$.
\end{prop}

\begin{proof}
This correspondence is given by the relation (\ref{z294}), written
in the form
\be
&&\g^i_\m=\g^i_{\m 0} + \g^i_{\m j}q^j_t =K_\m{}^i{}_0(q^\nu)\dot t
+ K_\m{}^i{}_j(q^\nu)\dot q^j|_{\dot t=1, \dot q^i=q^i_t}=\\
&& \qquad K_\m{}^i{}_0(q^\nu) + K_\m{}^i{}_j(q^\nu)q^j_t,
\ee
i.e., $\g^i_{\m\la}= K_\m{}^i{}_\la$.
\end{proof}

In particular, if an affine dynamic connection $\g$ is symmetric,
so is the corresponding linear connection $K$.

\begin{cor} \label{c2}
Every quadratic dynamic equation (\ref{cqg12}) on a configuration
bundle $Q\to\Bbb R$ of mechanics gives rise to the geodesic
equation
\ben
&& \ddot q^0= 0, \qquad \dot q^0=1,\nonumber\\
&& \ddot q^i=
a^i_{jk}(q^\m)\dot q^j \dot q^k + b^i_j(q^\m)\dot q^j\dot q^0 +
f^i(q^\m) \dot q^0\dot q^0 \label{cqg17}
\een
on the tangent bundle $TQ$ with respect to the symmetric linear
connection
\beq
K_\la{}^0{}_\nu=0, \quad K_0{}^i{}_0= f^i, \quad
K_0{}^i{}_j=\frac12 b^i_j, \quad K_k{}^i{}_j= a^i_{kj}
\label{cqg13}
\eeq
on the tangent bundle $TQ\to Q$.
\end{cor}

The geodesic equation (\ref{cqg17}), however, is not unique for
the dynamic equation (\ref{cqg12}).

\begin{prop} \label{jp40}
Any quadratic dynamic equation (\ref{cqg12}), being equivalent to
the geodesic equation with respect to the symmetric linear
connection $\wt K$ (\ref{cqg13}), also is equivalent to the
geodesic equation with respect to an affine connection $K'$ on
$TQ\to Q$ which differs from $\wt K$ (\ref{cqg13}) in a soldering
form $\si$ on $TQ\to Q$ with the components
\be
\si^0_\la= 0, \qquad \si^i_k= h^i_k+(s-1) h^i_k\dot q^0, \qquad
\si^i_0= -s h^i_k\dot q^k -h^i_0\dot q^0 + h^i_0,
\ee
where $s$ and $h^i_\la$ are local functions on $Q$.
\end{prop}

Proposition \ref{jp40} also can be deduced from the following
lemma.

\begin{lem}\label{c3}
Every affine vertical vector field
\beq
\si= [f^i(q^\m) + b^i_j(q^\m)q^j_t]\dr_i^t \label{gm364}
\eeq
on the affine jet bundle $J^1Q\to Q$ is extended to the soldering
form
\beq
\si=(f^idt + b^i_kdq^k)\ot\dot \dr_i \label{gm481}
\eeq
on the tangent bundle $TQ\to Q$.
\end{lem}

Now let us extend our inspection of dynamic equations to
connections on the tangent bundle $TM\to M$ of the typical fibre
$M$ of a configuration bundle $Q\to \Bbb R$. In this case, the
relationship fails to be  canonical, but depends on a
trivialization (\ref{gm219}) of $Q\to\Bbb R$.

Given such a trivialization, let $(t,\rrq^i)$ be the associated
coordinates on $Q$, where $\rrq^i$ are coordinates on $M$ with
transition functions independent of $t$.  The corresponding
trivialization (\ref{jp2}) of $J^1Q\to \Bbb R$ takes place in the
coordinates $(t,\rrq^i,\dot\rrq^i)$, where $\dot\rrq^i$ are
coordinates on $TM$. With respect to these coordinates, the
transformation law (\ref{m175}) of a dynamic connection $\g$ on
the affine jet bundle $J^1Q\to Q$ reads
\be
{\g'}_0^i =\frac{\dr\rrq'^i}{\dr\rrq^j}\g_0^j \qquad \g'^i_k =
\left(\frac{\dr\rrq'^i}{\dr \rrq^j}\g^j_n +\frac{\dr \dot
\rrq'^i}{\dr \rrq^n}\right)\frac{\dr\rrq^n}{\dr{\rrq'}^k}.
\ee
It follows that, given a trivialization of $Q\to\Bbb R$, a
connection $\g$ on $J^1Q\to Q$ defines the  time-dependent
vertical vector field
\be
\g^i_0(t,\rrq^j,\dot\rrq^j)\frac{\dr}{\dr \dot\rrq^i}:\Bbb R\times
TM\to VTM
\ee
 and the time-dependent connection
\beq
d\rrq^k\ot\left(\frac{\dr}{\dr \rrq^k} +\g^i_k(t,\rrq^j,
\dot\rrq^j) \frac{\dr}{\dr \dot\rrq^i}\right): \Bbb R\times TM\to
J^1TM\subset TTM \label{jp45}
\eeq
 on the tangent
bundle $TM\to M$.

Conversely, let us consider a connection
\be
\ol K=d\rrq^k\ot\left(\frac{\dr}{\dr \rrq^k} +\ol K^i_k(\rrq^j,
\dot\rrq^j) \frac{\dr}{\dr \dot\rrq^i}\right)
\ee
on the tangent bundle $TM\to M$. Given the above-mentioned
trivialization of the configuration bundle $Q\to\Bbb R$, the
connection $\ol K$ defines the connection $\wt K$ (\ref{z292})
with the components
\be
K_0^i=0, \qquad K^i_k=\ol K^i_k,
\ee
on the tangent bundle $TQ\to Q$. The corresponding dynamic
connection $\g$ on the affine jet bundle $J^1Q\to Q$ reads
\mar{m177}\beq
\g_0^i=0, \qquad \g^i_k=\ol K^i_k. \label{m177}
\eeq

Using the transformation law (\ref{m175}), one can extend the
expression (\ref{m177}) to arbitrary bundle coordinates $(t,q^i)$
on the configuration space $Q$ as follows:
\ben
&&\g^i_k=\left[\frac{\dr q^i}{\dr\rrq^j}\ol K^j_n(\rrq^j(q^r),
\dot\rrq^j(q^r,q^r_t)) +\frac{\dr^2 q^i}{\dr\rrq^n\dr\rrq^j}
\dot\rrq^j +
\frac{\dr\G^i}{\dr\rrq^n}\right]\dr_k\rrq^n,  \label{m178}\\
&&\g_0^i = \dr_t\G^i +\dr_j\G^iq^j_t - \g^i_k\G^k, \nonumber
\een
where $\G^i=\dr_tq^i(t,\rrq^j)$ is the connection on $Q\to\Bbb R$,
corresponding to a given trivialization of $Q$, i.e., $\G^i=0$
relative to $(t,\rrq^i)$. The dynamic equation on $Q$ defined by
the dynamic connection (\ref{m178})  takes the form
\mar{m179}\beq
q^i_{tt}=\dr_t\G^i +q^j_t\dr_j\G^i + \g^i_k(q^k_t-\G^k).
\label{m179}
\eeq
By construction, it is a conservative dynamic equation. Thus, we
have proved the following.

\begin{prop}\label{jp46} \mar{jp46}
Any connection $\ol K$ on the typical fibre $M$ of a configuration
bundle $Q\to \Bbb R$ yields a conservative dynamic equation
(\ref{m179}) on $Q$.
\end{prop}

\section{Reference frames}

From the physical viewpoint, a reference frame in non-relativistic
mechanics determines a tangent vector at each point of a
configuration space $Q$, which characterizes the velocity of an
observer at this point. This speculation leads to the following
mathematical definition of a reference frame in mechanics
\cite{book10,book98,massa,sard98}.

\begin{defi}\label{gn10} \mar{gn10}
A non-relativistic reference frame is a connection $\G$ on a
configuration space $Q\to\Bbb R$.
\end{defi}

By virtue of this definition, one can think of the horizontal
vector field (\ref{a1.10}) associated with a connection $\G$ on
$Q\to\Bbb R$ as being a family of observers, while the
corresponding covariant differential (\ref{z279}):
\be
\dot q^i_\G= D^\G(q^i_t)= q^i_t-\G^i,
\ee
determines the relative velocity with respect to a reference frame
$\G$. Accordingly, $q^i_t$ are regarded as the absolute
velocities.

In particular, given a motion $c:\Bbb R\to Q$, its covariant
derivative $\nabla^\G c$ (\ref{+190}) with respect to a connection
$\G$ is a velocity of this motion relative to a reference frame
$\G$. For instance, if $c$ is an integral section for a connection
$\G$, a velocity of the motion $c$ relative to a reference frame
$\G$ is equal to 0. Conversely, every motion $c:\Bbb R\to Q$
defines a reference frame $\G_c$  such that a velocity of $c$
relative to $\G_c$ vanishes. This reference frame $\G_c$ is an
extension of a section $c(\Bbb R)\to J^1Q$ of an affine jet bundle
$J^1Q\to Q$ over the closed submanifold $c(\Bbb R)\in Q$ to a
global section in accordance with Theorem \ref{mos9}.

By virtue of Theorem \ref{gn1}, any reference frame $\G$ on a
configuration bundle $Q\to\Bbb R$ is associated with an atlas of
local constant trivializations, and {\it vice versa}. A connection
$\G$ takes the form $\G=\dr_t$ (\ref{z271}) with respect to the
corresponding coordinates $(t,\rrq^i)$, whose transition functions
$\rrq^i\to \rrq'^i$ are independent of time. One can think of
these coordinates as also being a reference frame, corresponding
to the connection (\ref{z271}). They are called the adapted
coordinates to a reference frame $\G$. Thus, we come to the
following definition, equivalent to Definition \ref{gn10}.

\begin{defi}\label{gn11}
In mechanics, a reference frame is an atlas of local constant
trivializations of a configuration bundle $Q\to\Bbb R$.
\end{defi}

In particular, with respect to the coordinates $\rrq^i$ adapted to
a reference frame $\G$, the velocities relative to this reference
frame coincide with the absolute ones
\be
D^\G(\rrq^i_t)={\dot \rrq}^i_\G=\rrq^i_t.
\ee

\begin{rem} \label{016} \mar{016} By analogy with gauge field theory,
we agree to call transformations of bundle atlases of a fibre
bundle $Q\to\Bbb R$ the gauge transformations. To be precise, one
should call them passive gauge transformations, while by active
gauge transformations are meant automorphisms of a fibre bundle.
In mechanics, gauge transformations also are reference frame
transformations in accordance with Theorem \ref{gn1}. An object on
a fibre bundle is said to be gauge covariant  or, simply,
covariant if its definition is atlas independent. It is called
gauge invariant if its form is maintained under atlas
transformations.
\end{rem}

A reference frame is said to be complete if the associated
connection $\G$ is complete. By virtue of Proposition \ref{compl},
every complete reference frame defines a trivialization  of a
bundle $Q\to\Bbb R$, and {\it vice versa}.

\begin{rem}
Given a reference frame $\G$, one should solve the equations
\mar{gm300}\ben
&& \G^i(t, q^j(t,\rrq^a))=\frac{\dr q^i(t,\rrq^a)}{\dr t}, \label{gm300a}\\
&& \frac{\dr \rrq^a(t,q^j)}{\dr q^i}\G^i(t,q^j) +\frac{\dr \rrq^a(t,q^j)}{\dr
t}=0 \label{gm300b}
\een
in order to find the coordinates $(t,\rrq^a)$ adapted to $\G$. Let
$(t,q^a_1)$ and $(t,q^i_2)$ be the adapted coordinates   for
reference frames $\G_1$ and $\G_2$, respectively. In accordance
with the equality (\ref{gm300b}),  the components $\G^i_1$ of the
connection $\G_1$ with respect to the coordinates $(t,q^i_2)$ and
the components $\G^a_2$ of the connection $\G_2$ with respect to
the coordinates $(t,q^a_1)$ fulfill the relation
\be
\frac{\dr q^a_1}{\dr q^i_2}\G^i_1 +\G^a_2=0.
\ee
\end{rem}

Using the relations (\ref{gm300a}) -- (\ref{gm300b}), one can
rewrite the coordinate transformation law (\ref{z317}) of dynamic
equations as follows. Let
\mar{gm302}\beq
\rrq^a_{tt}=\ol\xi^a \label{gm302}
\eeq
be a dynamic equation on a configuration space $Q$ written with
respect to a reference frame $(t,\rrq^n)$. Then, relative to
arbitrary bundle coordinates $(t,q^i)$ on $Q\to\Bbb R$, the
dynamic equation (\ref{gm302}) takes the form
\beq
 q^i_{tt}=d_t\G^i +\dr_j\G^i(q^j_t-\G^j) -
\frac{\dr q^i}{\dr\rrq^a}\frac{\dr\rrq^a}{\dr q^j\dr
q^k}(q^j_t-\G^j) (q^k_t-\G^k) + \frac{\dr q^i}{\dr\rrq^a}\ol\xi^a,
\label{gm304}
\eeq
where $\G$ is a connection corresponding to the reference frame
$(t,\rrq^n)$. The dynamic equation (\ref{gm304}) can be expressed
in the relative velocities $\dot q^i_\G=q^i_t-\G^i$ with respect
to the initial reference frame $(t,\rrq^a)$. We have
\beq
 d_t\dot q^i_\G=\dr_j\G^i\dot q^j_\G -
\frac{\dr q^i}{\dr\rrq^a}\frac{\dr\rrq^a}{\dr q^j\dr q^k}\dot
q^j_\G \dot q^k_\G + \frac{\dr q^i}{\dr\rrq^a}\ol\xi^a(t,q^j,\dot
q^j_\G). \label{gm307}
\eeq
Accordingly, any dynamic equation (\ref{z273}) can be expressed in
the relative velocities $\dot q^i_\G=q^i_t-\G^i$ with respect to
an arbitrary reference frame $\G$ as follows:
\beq
d_t\dot q^i_\G= (\xi-J\G)^i_t=\xi^i-d_t\G, \label{gm308}
\eeq
where $J\G$ is the prolongation (\ref{gm217'}) of a connection
$\G$ onto the jet bundle $J^1Q\to\Bbb R$.

For instance, let us consider the following particular reference
frame $\G$ for a dynamic equation $\xi$. The covariant derivative
of a reference frame $\G$ with respect to the corresponding
dynamic connection $\g_\xi$ (\ref{z286}) reads
\ben
&&\nabla^\g\G=  Q\to T^*Q\times V_QJ^1Q, \label{jp57}\\
&& \nabla^\g\G= \nabla^\g_\la\G^k dq^\la\ot\dr_k, \qquad
\nabla^\g_\la\G^k = \dr_\la\G^k - \g^k_\la\circ\G. \nonumber
\een
A connection $\G$ is called a geodesic reference frame for the
dynamic equation $\xi$ if
\beq
\G\rfloor\nabla^\g \G= \G^\la(\dr_\la\G^k -
\g^k_\la\circ\G)=(d_t\G^i-\xi^i\circ \G)\dr_i=0. \label{gm310}
\eeq

\begin{prop}
Integral sections $c$ for a reference frame $\G$ are solutions of
a dynamic equation $\xi$ iff $\G$ is a geodesic reference frame
for $\xi$.
\end{prop}

\begin{rem}
The left- and right-hand sides of the equation (\ref{gm308})
separately are not well-behaved objects.  This equation is brought
into the covariant form (\ref{gm389}).
\end{rem}

Reference frames play a prominent role in many constructions of
mechanics. They enable us to write the covariant forms:
(\ref{gm390}) -- (\ref{gm389}) of dynamic equations and
(\ref{m46'}) of Hamiltonians of mechanics.

With a reference frame, we obtain the converse of Theorem
\ref{jp50}.

\begin{theo}\label{jp51} \mar{jp51}
Given a reference frame $\G$, any connection $K$ (\ref{z290}) on
the tangent bundle $TQ\to Q$ defines a dynamic equation
\be
\xi^i= (K^i_\la -\G^i K^0_\la)\dot q^\la\mid_{\dot q^0=1,\dot
q^j=q^j_t}.
\ee
\end{theo}

This theorem is a corollary of Proposition \ref{motion1} and the
following lemma.

\begin{lem} \label{jp51'}
Given a connection $\G$ on a fibre bundle $Q\to\Bbb R$ and a
connection $K$ on the tangent bundle $TQ\to Q$, there is the
connection $\wt K$ on $TQ\to Q$ with the components
\be
\wt K^0_\la =0, \qquad \wt K^i_\la = K^i_\la - \G^iK^0_\la.
\ee
\end{lem}

\section{Free motion equations}

Let us point out the following interesting class of dynamic
equations which we agree to call the free motion equations.

\begin{defi}
We say that the dynamic equation (\ref{z273}) is a free motion
equation if there exists a reference frame $(t,\ol q^i)$ on the
configuration space $Q$
 such that this equation reads
\beq
\ol q^i_{tt}=0. \label{z280}
\eeq
\end{defi}

With respect to arbitrary bundle coordinates $(t,q^i)$, a free
motion equation takes the form
\beq
 q^i_{tt}=d_t\G^i +\dr_j\G^i(q^j_t-\G^j) -
\frac{\dr q^i}{\dr\rrq^m}\frac{\dr\rrq^m}{\dr q^j\dr
q^k}(q^j_t-\G^j) (q^k_t-\G^k),  \label{m188}
\eeq
where $\G^i=\dr_t q^i(t,\ol q^j)$ is the connection associated
with the initial frame $(t,\ol q^i)$ (cf. (\ref{gm304})).  One can
think of the right-hand side of the equation (\ref{m188}) as being
the general coordinate expression for an inertial force in
mechanics. The corresponding dynamic connection $\g_\xi$ on the
affine jet bundle $J^1Q\to Q$ reads
\ben
&& \g^i_k=\dr_k\G^i  -
\frac{\dr q^i}{\dr\rrq^m}\frac{\dr\rrq^m}{\dr q^j\dr
q^k}(q^j_t-\G^j),
\label{gm366} \\
&& \g^i_0= \dr_t\G^i +\dr_j\G^iq^j_t -\g^i_k\G^k. \nonumber
\een
It is affine. By virtue of Proposition \ref{aff}, this dynamic
connection defines a linear connection $K$ on the tangent bundle
$TQ\to Q$, whose curvature  necessarily vanishes.  Thus, we come
to the following criterion of a dynamic equation to be a free
motion equation.

\begin{prop}
If $\xi$ is a free motion equation on a configuration space $Q$,
it is quadratic, and the corresponding symmetric linear connection
(\ref{cqg13}) on the tangent bundle $TQ\to Q$ is a curvature-free
connection.
\end{prop}

This criterion is not a sufficient condition because it may happen
that the components of a curvature-free  symmetric linear
connection on $TQ\to Q$ vanish with respect to the coordinates on
$Q$
 which are not compatible with a fibration $Q\to\Bbb R$.

The similar criterion involves the curvature of a dynamic
connection (\ref{gm366}) of a free motion equation.

\begin{prop} \label{gena120}
If $\xi$ is a free motion equation, then the curvature $R$
(\ref{gm282}) of the corresponding dynamic connection $\g_\xi$ is
equal to 0, and so are the tensor field $\ol R$ (\ref{gm283}) and
the scalar field $\wt R$ (\ref{gm284}).
\end{prop}

Proposition \ref{gena120} also fails to be a sufficient condition.
If the curvature $R$ (\ref{gm282}) of a dynamic connection
$\g_\xi$ vanishes, it may happen that components of $\g_\xi$ are
equal to 0 with respect to non-holonomic bundle coordinates on an
affine jet bundle $J^1Q\to Q$.

Nevertheless, we can formulate the necessary and sufficient
condition of the existence of a free motion equation on a
configuration space $Q$.

\begin{prop} \label{gena110}
A free motion equation on a fibre bundle $Q\to\Bbb R$ exists iff
the typical fibre $M$ of $Q$ admits a curvature-free symmetric
linear connection.
\end{prop}

\begin{proof}
Let a free motion equation take the form (\ref{z280}) with respect
to some atlas of local constant trivializations of a fibre bundle
$Q\to\Bbb R$. By virtue of Proposition \ref{gena52}, there exists
an affine dynamic connection $\g$ on the affine jet bundle
$J^1Q\to Q$ whose components relative to this atlas are equal to
0. Given a trivialization chart of this atlas, the connection $\g$
defines the curvature-free symmetric linear connection
(\ref{jp45}) on $M$.  The converse statement follows at once from
Proposition \ref{jp46}.
\end{proof}

The free motion equation (\ref{m188}) is simplified if the
coordinate transition functions $\ol q^i\to q^i$ are affine in
coordinates $\ol q^i$. Then we have
\beq
q^i_{tt}=\dr_t\G^i -\G^j\dr_j\G^i +2q^j_t\dr_j\G^i. \label{m182}
\eeq

The following lemma shows that the free motion equation
(\ref{m182}) is affine in the coordinates $q^i$ and $q^i_t$
\cite{book10,book98}.

\begin{lem}\label{gena115}
Let $(t,\rrq^a)$  be a reference frame on a configuration bundle
$Q\to \Bbb R$ and $\G$ the corresponding connection. Components
$\G^i$ of this connection with respect to another coordinate
system $(t,q^i)$ are affine functions in the coordinates $q^i$ iff
the transition functions between the coordinates $\rrq^a$ and
$q^i$ are affine.
\end{lem}

One can easily find the geodesic reference frames for the free
motion equation
\beq
q^i_{tt}=0. \label{z318}
\eeq
They are $\G^i=v^i=$ const. By virtue of Lemma \ref{gena115},
these reference frames define the adapted coordinates
\beq
\ol q^i =k^i_jq^j-v^it-a^i, \quad k^i_j={\rm const.}, \quad
v^i={\rm const.}, \quad a^i={\rm const.}\label{z320}
\eeq
The  equation (\ref{z318})
 obviously keeps its free motion form under the transformations (\ref{z320})
between the geodesic reference frames. It is readily observed that
these transformations are precisely the elements of the Galilei
group.

\section{Relative acceleration}

In comparison with the notion of a relative velocity, the
definition of a relative acceleration is more intricate.

To consider a relative acceleration with respect to a reference
frame $\G$,  one should prolong a connection $\G$ on a
configuration space $Q\to\Bbb R$ to a holonomic connection
$\xi_\G$ on the jet bundle $J^1Q\to\Bbb R$. Note that the jet
prolongation $J\G$ (\ref{gm217'}) of $\G$ onto $J^1Q\to\Bbb R$ is
not holonomic. We can construct the desired prolongation by means
of a dynamic connection $\g$ on an affine jet bundle $J^1Q\to Q$.

\begin{lem}\label{gn3}
Let us consider the composite bundle (\ref{gm361}). Given a
reference frame $\G$ on $Q\to\Bbb R$ and a dynamic connections
$\g$ on $J^1Q\to Q$, there exists a dynamic connection $\wt\g$ on
$J^1Q\to Q$ with the components
\beq
\wt \g^i_k=\g^i_k, \qquad \wt \g^i_0=d_t\G^i-\g^i_k\G^k.
\label{gm367}
\eeq
\end{lem}

Now, we construct a certain soldering form on an affine jet bundle
$J^1Q\to Q$ and add it to this connection. Let us apply
 the canonical projection $T^*Q\to V^*Q$ and then the imbedding
$\G:V^*Q\to T^*Q$ to the covariant derivative (\ref{jp57}) of the
reference frame $\G$ with respect to the dynamic connection $\g$.
We obtain the $V_QJ^1Q$-valued one-form
\be
\si= [-\G^i(\dr_i\G^k - \g^k_i\circ\G)dt +(\dr_i\G^k -
\g^k_i\circ\G)dq^i]\ot\dr_k^t
\ee
on $Q$ whose pull-back onto $J^1Q$ is a desired soldering form.
The sum
\be
\g_\G\op=\wt \g +\si,
\ee
called the frame connection, reads
\ben
 &&\g_\G{}^i_0= d_t\G^i - \g^i_k\G^k -\G^k(\dr_k\G^i -
\g^i_k\circ\G),\label{jp68}\\
 && \g_\G{}^i_k= \g^i_k +\dr_k\G^i - \g^i_k\circ\G. \nonumber
\een
This connection yields the desired holonomic connection
\be
\xi_\G^i= d_t\G^i +(\dr_k\G^i +\g^i_k - \g^i_k\circ\G)(q^k_t-\G^k)
\ee
on the jet bundle $J^1Q\to \Bbb R$.

 Let $\xi$ be a dynamic equation and
 $\g=\g_\xi$ the connection (\ref{z286}) associated with
$\xi$. Then one can think of the vertical vector field
\beq
a_\G\op=\xi-\xi_\G=(\xi^i-\xi_\G^i)\dr^t_i \label{gm369}
\eeq
 on the affine jet bundle $J^1Q\to Q$  as being a relative
acceleration with respect to the reference frame $\G$ in
comparison with the absolute acceleration $\xi$.

For instance, let us consider a reference frame $\G$ which is
geodesic for the dynamic equation $\xi$, i.e., the relation
(\ref{gm310}) holds. Then the relative acceleration of a motion
$c$ with respect to a reference frame $\G$ is
\be
(\xi-\xi_\G)\circ\G=0.
\ee

Let $\xi$ now be an arbitrary dynamic equation, written with
respect to coordinates $(t,q^i)$ adapted to a reference frame
$\G$, i.e., $\G^i=0$. In these coordinates, the relative
acceleration with respect to a reference frame $\G$  is
\beq
a^i_\G = \xi^i(t,q^j,q^j_t) -\frac12 q^k_t(\dr_k \xi^i - \dr_k
\xi^i\mid_{q^j_t=0}). \label{jp64}
\eeq
Given another bundle coordinates $(t, q'^i)$ on $Q\to\Bbb R$, this
dynamic equation takes the form (\ref{gm307}), while the relative
acceleration (\ref{jp64}) with respect to a reference frame $\G$
reads
\be
a'^i_\G =\dr_jq'^i a^j_\G.
\ee
Then we can write the dynamic equation (\ref{z273}) in the form
which is covariant under coordinate transformations:
\mar{gm390}\beq
\wt D_{\g_\G} q^i_t = d_t q^i_t -\xi^i_\G=a_\G, \label{gm390}
\eeq
where $\wt D_{\g_\G}$ is the vertical covariant differential
(\ref{gm388}) with respect to the frame connection $\g_\G$
(\ref{jp68}) on an affine jet bundle $J^1Q\to Q$.

In particular, if $\xi$ is a free motion equation which takes the
form (\ref{z280}) with respect to a reference frame $\G$, then
\be
\wt D_{\g_\G} q^i_t=0
\ee
relative to arbitrary bundle coordinates on the configuration
bundle $Q\to\Bbb R$.

The left-hand side of the dynamic equation (\ref{gm390}) also can
be expressed in the relative velocities such that this dynamic
equation takes the form
\mar{gm389}\beq
d_t\dot q^i_\G -\g_\G{}^i_k\dot q^k_\G = a_\G  \label{gm389}
\eeq
 which is the covariant form of the equation
(\ref{gm308}).

The concept of a relative acceleration is understood better when
we deal with a quadratic dynamic equation $\xi$, and the
corresponding dynamic connection $\g$ is affine.

\begin{lem}\label{gn4}
If a dynamic connection $\g$ is affine, i.e.,
\be
\g^i_\la=\g^i_{\la 0} + \g^i_{\la k}q^k_t,
\ee
so is a frame connection $\g_\G$ for any frame $\G$.
\end{lem}

In particular,   we obtain
\be
\g_\G{}^i_{jk}=\g^i_{jk}, \qquad
\g_\G{}^i_{0k}=\g_\G{}^i_{k0}=\g_\G{}^i_{00}=0
\ee
relative to the coordinates adapted to a reference frame $\G$.

\begin{cor} \label{gn5}
If a dynamic equation $\xi$ is quadratic, the relative
acceleration $a_\G$ (\ref{gm369}) is always affine, and it admits
the decomposition
\beq
a_\G^i= -(\G^\la\nabla^\g_\la\G^i + 2\dot q^\la_\G
\nabla^\g_\la\G^i), \label{gm371}
\eeq
where $\g=\g_\xi$ is the dynamic connection (\ref{z286}), and
\be
\dot q^\la_\G=q^\la_t -\G^\la, \qquad q^0_t=1, \qquad \G^0=1,
\ee
is the relative velocity with respect to the reference frame $\G$.
\end{cor}

Note that the splitting (\ref{gm371}) gives a generalized
 Coriolis theorem.  In
particular, the well-known analogy between inertial and
electromagnetic forces is restated. Corollary \ref{gn5} shows that
this analogy can be extended to an arbitrary quadratic dynamic
equation.

\section{Newtonian systems}

Equations of motion of non-relativistic mechanics need not be
exactly dynamic equations. For instance, the second Newton law of
point mechanics contains a mass. The notion of a Newtonian system
generalizes the second Newton law as follows.

Let $m$ be a fibre metric (bilinear form) in the vertical tangent
bundle $V_QJ^1Q\to J^1Q$ of $J^1Q\to Q$. It reads
\mar{xx5}\beq
m:J^1Q\to\op\vee^2_{J^1Q} V_Q^*J^1Q,\qquad m=\frac12 m_{ij} \ol
dq^i_t\vee \ol dq^j_t, \label{xx5}
\eeq
where $\ol dq^i_t$ are the holonomic bases for the vertical
cotangent bundle $V^*_QJ^1Q$ of $J^1Q\to Q$. It defines the map
\be
\wh m:  V_QJ^1Q\to  V_Q^*J^1Q.
\ee

\begin{defi}\label{gn30}  \mar{gn30} Let $Q\to\Bbb R$ be a fibre bundle
together with:

(i) a fibre metric $m$ (\ref{xx5}) satisfying the symmetry
condition
\mar{a1.103}\beq
\dr_k^t m_{ij}=\dr_j^t m_{ik}, \label{a1.103}
\eeq

(ii) and a holonomic connection $\xi$ (\ref{a1.30}) on a jet
bundle $J^1Q\to \Bbb R$ related to the fibre metric $m$ by the
compatibility condition
\mar{a1.95}\beq
\xi\rfloor dm_{ij} + \frac{1}{2}m_{ik}\dr_j^t\xi^k +
m_{jk}\dr_i^t\xi^k = 0. \label{a1.95}
\eeq
A triple $(Q,m,\xi)$ is called the Newtonian system.
\end{defi}

We agree to call a metric $m$ in Definition \ref{gn30} the mass
tensor of a Newtonian system $(Q,m,\xi)$. The equation of motion
of this Newtonian system is defined to be
\mar{z355}\beq
\wh m(D^\xi)=0, \qquad m_{ik}(q^k_{tt}-\xi^k)=0. \label{z355}
\eeq
Due to the conditions (\ref{a1.103}) and (\ref{a1.95}), it is
brought into the form
\be
d_t(m_{ik}q^k_t) - m_{ik}\xi^k=0.
\ee
Therefore, one can think of this equation as being a
generalization of the second Newton law.

If a mass tensor $m$ (\ref{xx5}) is non-degenerate, the equation
of motion (\ref{z355}) is equivalent to the dynamic equation
\be
D^\xi=0, \qquad q^k_{tt}-\xi^k=0.
\ee

Because of the canonical vertical splitting (\ref{gm382}), the
mass tensor (\ref{xx5}) also is a map
\mar{xx11}\beq
m:J^1Q\to\op\vee^2_{J^1Q} V^*Q,\qquad m=\frac12 m_{ij} dq^i\vee
dq^j, \label{xx11}
\eeq

A Newtonian system $(Q,m,\xi)$ is said to be standard, if its mass
tensor $m$ is the pull-back onto $V_QJ^1Q$ of a fibre metric
\mar{xx12}\beq
m:Q\to\op\vee^2_Q V^*Q \label{xx12}
\eeq
in the vertical tangent bundle $VQ\to Q$ in accordance with the
isomorphisms (\ref{a1.4}) and (\ref{gm382}), i.e., $m$ is
independent of the velocity coordinates $q^i_t$.

Given a mass tensor, one can introduce the notion of an external
force.

\begin{defi} An external force is
defined as a section of the vertical cotangent bundle
$V_Q^*J^1Q\to J^1Q$. Let us also bear in mind the isomorphism
(\ref{gm382}).
\end{defi}

It should be emphasized that there are no canonical isomorphisms
between the vertical cotangent bundle $V_Q^*J^1Q$ and the vertical
tangent bundle $V_QJ^1Q$ of $J^1Q$. One must therefore distinguish
forces and accelerations which are related by means of a mass
tensor.

Let $(Q,\wh m,\xi)$ be a Newtonian system and $f$ an external
force. Then
\mar{gm387}\beq
\xi^i_f\op= \xi^i + (m^{-1})^{ik}f_k \label{gm387}
\eeq
is a dynamic equation, but the triple $(Q,m,\xi_f)$ is not a
Newtonian system in general. As it follows from a direct
computation, iff an external force possesses the property
\mar{gm383}\beq
\dr_i^tf_j+\dr_j^tf_i=0,  \label{gm383}
\eeq
then $\xi_f$ (\ref{gm387}) fulfills the compatibility condition
(\ref{a1.95}), and $(Q,\wh m,\xi_f)$ also is a Newtonian system.

\section{Integrals of motion}

Let an equation of motion of a mechanical system on a fibre bundle
$Y\to\Bbb R$ be described by an $r$-order differential equation
$\gE$ given by a closed subbundle of the jet bundle $J^rY\to\Bbb
R$ in accordance with Definition \ref{equa}.

\begin{defi} \label{026} \mar{026}
An integral of motion of this mechanical system is defined as a
$(k<r)$-order differential operator $\Phi$ on $Y$ such that $\gE$
belongs to the kernel of an $r$-order jet prolongation of the
differential operator $d_t\Phi$, i.e.,
\mar{021}\beq
J^{r-k-1}(d_t\Phi)|_{\gE}=0. \label{021}
\eeq
\end{defi}

It follows that an integral of motion $\Phi$ is constant on
solutions $s$ of a differential equation $\gE$, i.e., there is the
differential conservation law
\mar{020}\beq
(J^ks)^*\Phi={\rm const}., \qquad (J^{k+1}s)^*d_t\Phi=0.
\label{020}
\eeq

We agree to write the condition (\ref{021}) as the weak equality
\mar{022}\beq
J^{r-k-1}(d_t\Phi)\ap 0, \label{022}
\eeq
which holds on-shell, i.e., on solutions of a differential
equation $\gE$ by the formula (\ref{020}).

In mechanics, we can restrict our consideration to integrals of
motion $\Phi$ which are functions on $J^kY$. As was mentioned
above, equations of motion of mechanics mainly are of first or
second order. Accordingly, their integrals of motion are functions
on $Y$ or $J^kY$. In this case, the corresponding weak equality
(\ref{021}) takes the form
\mar{027}\beq
d_t\Phi\ap 0 \label{027}
\eeq
of a weak conservation law or, simply, a conservation law.

Different integrals of motion need not be independent. Let
integrals of motion $\Phi_1,\ldots, \Phi_m$ of a mechanical system
on $Y$ be functions on $J^kY$. They are called independent if
\mar{024}\beq
d\Phi_1\w\cdots\w d\Phi_m\neq 0 \label{024}
\eeq
everywhere on $J^kY$. In this case, any motion $J^ks$ of this
mechanical system lies in the common level surfaces of functions
$\Phi_1,\ldots, \Phi_m$ which bring $J^kY$ into a fibred manifold.

Integrals of motion can come from symmetries. This is the case of
Lagrangian and Hamiltonian mechanics (Sections 2.4 and 3.5).

\begin{defi} \label{025} \mar{025}
Let an equation of motion of a mechanical system be an $r$-order
differential equation $\gE\subset J^rY$. Its infinitesimal
symmetry (or, simply, a symmetry) is defined as a vector field on
$J^rY$ whose restriction to $\gE$ is tangent to $\gE$.
\end{defi}

For instance, let us consider first order dynamic equations.

\begin{prop} \label{080} \mar{080}
Let $\gE$ be the autonomous first order dynamic equation
(\ref{gm150}) given by a vector field $u$ on a manifold $Z$. A
vector field $\vt$ on $Z$ is its symmetry iff $[u,\vt]\ap 0$.
\end{prop}

One can show that a smooth real function $F$ on a manifold $Z$ is
an integral of motion of the autonomous first order dynamic
equation (\ref{gm150}) (i.e., it is constant on solutions of this
equation) iff its Lie derivative along $u$ vanishes:
\mar{0115}\beq
\bL_u F=u^\la\dr_\la \Phi=0. \label{0115}
\eeq

\begin{prop} \label{081} \mar{081}
Let $\gE$ be the first order dynamic equation (\ref{ggm155}) given
by a connection $\G$ (\ref{a1.10}) on a fibre bundle $Y\to\Bbb R$.
Then a vector field $\vt$ on $Y$ is its symmetry iff $[\G,\vt]\ap
0$.
\end{prop}

A smooth real function $\Phi$ on $Y$ is an integral of motion of
the first order dynamic equation (\ref{ggm155}) in accordance with
the equality (\ref{027}) iff
\mar{0116}\beq
\bL_\G \Phi=(\dr_t +\G^i\dr_i)\Phi=0. \label{0116}
\eeq

Following Definition \ref{025}, let us introduce the notion of a
symmetry of differential operators in the following relevant case.
Let us consider an $r$-order differential operator on a fibre
bundle $Y\to\Bbb R$ which is represented by an exterior form $\cE$
on $J^rY$ (Definition \ref{oper}). Let its kernel $\Ker\cE$ be an
$r$-order differential equation on $Y\to\Bbb R$.

\begin{prop} \label{082} \mar{082}
It is readily justified that a vector field $\vt$ on $J^rY$ is a
symmetry of the equation $\Ker\cE$ in accordance with Definition
\ref{025} iff
\mar{083}\beq
\bL_\vt \cE\ap 0. \label{083}
\eeq
\end{prop}

Motivated by Proposition \ref{082}, we come to the following.

\begin{defi} \label{084} \mar{084} Let $\cE$ be the above
mentioned differential operator. A vector field $\vt$ on $J^rY$ is
called a symmetry of a differential operator $\cE$ if the Lie
derivative $\bL_\vt \cE$ vanishes.
\end{defi}

By virtue of Proposition \ref{082}, a symmetry of a differential
operator $\cE$ also is a symmetry of the differential equation
$\Ker\cE$.

Note that there exist integrals of motion which are not associated
with symmetries of an equation of motion (Example \ref{048}).

\chapter{Lagrangian mechanics}

Lagrangian mechanics on a velocity phase space is formulated in
the framework of Lagrangian formalism on fibre bundles
\cite{book,book09,book10,book98}. This formulation is based on the
variational bicomplex and the first variational formula, without
appealing to the variational principle. Besides Lagrange
equations, the Cartan and Hamilton--De Donder equations are
considered in the framework of Lagrangian formalism. Note that the
Cartan equation, but not the Lagrange one is associated to the
Hamilton equation (Section 3.4). The relations between Lagrangian
and Newtonian systems are investigated. Lagrangian conservation
laws are defined by means of the first Noether theorem.

\section{Lagrangian formalism on $Q\to\Bbb R$}

Let $\pi: Q\to \Bbb R$ be a fibre bundle (\ref{gm360}). The finite
order jet manifolds $J^kQ$ of $Q\to \Bbb R$ form the inverse
sequence
\mar{j1}\beq
Q\op\longleftarrow^{\pi^1_0} J^1Q \longleftarrow \cdots J^{r-1}Q
\op\longleftarrow^{\pi^r_{r-1}} J^rQ\longleftarrow\cdots,
\label{j1}
\eeq
where $\pi^r_{r-1}$ are affine bundles. Its projective limit
$J^\infty Q$ is a paracompact Fr\'echet manifold. One can think of
its elements as being infinite order jets of sections of $Q\to
\Bbb R$ identified by their Taylor series at points of $\Bbb R$.
Therefore, $J^\infty Q$ is called the infinite order jet manifold.
A bundle coordinate atlas $(t,q^i)$ of $Q\to \Bbb R$ provides
$J^\infty Q$ with the manifold coordinate atlas
\mar{j3'}\beq
(t,q^i, q^i_t,q^i_{tt},\ldots), \qquad {q'}^i_{t\La}=d_t q'^i_\La,
\label{j3'}
\eeq
where $\La=(t\cdots t)$ denotes a multi-index of length $|\La|$
and
\be
d_t=\dr_t +q^i_t\dr_i +q^i_{tt}\dr^t_i+\cdots
+q^i_{t\La}\dr_i^\La+\cdots
\ee
is the total derivative.

Let $\cO_r^*=\cO^*(J^rQ)$ be a graded differential algebra of
exterior forms on a jet manifold $J^rQ$. The inverse sequence
(\ref{j1}) of jet manifolds yields the direct sequence of
differential graded algebras $\cO_r^*$:
\mar{5.7}\beq
\cO^*(Q) \op\longrightarrow^{\pi^1_0{}^*} \cO_1^* \longrightarrow
\cdots \cO_{r-1}^*\op\longrightarrow^{\pi^r_{r-1}{}^*}
 \cO_r^* \longrightarrow\cdots, \label{5.7}
\eeq
where $\pi^r_{r-1}{}^*$ are the pull-back monomorphisms. Its
direct limit
\mar{01}\beq
\cO^*_\infty Q=\op\lim^\to \cO^*_r \label{01}
\eeq
(or, simply, $\cO^*_\infty$) consists of all exterior forms on
finite order jet manifolds modulo the pull-back identification. In
particular, $\cO^0_\infty$ is the ring of all smooth functions on
finite order jet manifolds. The $\cO^*_\infty$ (\ref{01}) is a
differential graded algebra which inherits the operations of the
exterior differential $d$ and exterior product $\w$ of exterior
algebras $\cO^*_r$.

\begin{theo} \label{j4} \mar{j4} The cohomology $H^*(\cO_\infty^*)$ of
the de Rham complex
\mar{5.13} \beq
0\longrightarrow \Bbb R\longrightarrow \cO^0_\infty
\op\longrightarrow^d\cO^1_\infty \op\longrightarrow^d \cdots
\label{5.13}
\eeq
of the differential graded algebra $\cO^*_\infty$ equals the de
Rham cohomology $H^*_{\rm DR}(Q)$ of a fibre bundle $Q$
\cite{book09}.
\end{theo}

\begin{cor} \label{02} \mar{02}
Since $Q$ (\ref{gm360}) is a trivial fibre bundle over $\Bbb R$,
the de Rham cohomology $H^*(\cO_\infty^*)$ of $Q$ equals the de
Rham cohomology of its typical fibre $M$ in accordance with the
well-known K\"unneth formula. Therefore, the cohomology
$H^*(\cO_\infty^*)$ of the de Rham complex (\ref{5.13}) equals the
de Rham cohomology $H^*_{\rm DR}(M)$ of $M$.
\end{cor}

Since elements of the differential graded algebra $\cO^*_\infty$
(\ref{01}) are exterior forms on finite order jet manifolds, this
$\cO^0_\infty$-algebra is locally generated by the horizontal form
$dt$ and contact one-forms
\be
\thh^i_\La=dq^i_\La -q^i_{t\La}dt.
\ee
Moreover, there is the canonical decomposition
\be
\cO^*_\infty =\oplus\cO^{k,m}_\infty , \qquad m=0,1,
\ee
of $\cO^*_\infty$ into $\cO^0_\infty $-modules $\cO^{k,m}_\infty$
of $k$-contact and $(m=0,1)$-horizontal forms together with the
corresponding projectors
\be
h_k:\cO^*_\infty\to \cO^{k,*}_\infty, \qquad h^m:\cO^*_\infty \to
\cO^{*,m}_\infty.
\ee
Accordingly, the exterior differential on $\cO_\infty^*$ is
decomposed into the sum $d=d_V+d_H$ of the vertical differential
\be
&& d_V: \cO_\infty^{k,m}\to \cO_\infty^{k+1,m}, \qquad d_V \circ h^m=h^m\circ d\circ h^m,\\
&& d_V(\f)=\thh^i_\La \w \dr^\La_i\f, \qquad \f\in\cO^*_\infty,
\ee
and the total differential
\mar{xx28}\ben
&& d_H: \cO_\infty^{k,m}\to \cO_\infty^{k,m+1}, \quad d_H\circ h_k=h_k\circ d\circ
h_k, \quad d_H\circ h_0=h_0\circ d, \nonumber\\
&& d_H(\f)= dt\w d_t\f, \qquad \f\in\cO^*_\infty. \label{xx28}
\een
These differentials obey the nilpotent conditions
\be
d_H\circ d_H=0, \qquad d_V\circ d_V=0, \qquad d_H\circ
d_V+d_V\circ d_H=0,
\ee
and make $\cO^{*,*}_\infty$ into a bicomplex.

One introduces the following two additional operators acting on
$\cO^{*,n}_\infty$.

(i) There exists an $\Bbb R$-module endomorphism
\mar{r12}\ben
&& \vr=\op\sum_{k>0} \frac1k\ol\vr\circ h_k\circ h^1:
\cO^{*>0,1}_\infty\to \cO^{*>0,1}_\infty,
\label{r12}\\
&& \ol\vr(\f)= \op\sum_{0\leq|\La|} (-1)^{|\La|}\thh^i\w
[d_\La(\dr^\La_i\rfloor\f)], \qquad \f\in \cO^{>0,1}_\infty,
\nonumber
\een
possessing the following properties.

\begin{lem} \label{21t1} \mar{21t1}
For any $\f\in \cO^{>0,1}_\infty$, the form $\f-\vr(\f)$ is
locally $d_H$-exact on each coordinate chart (\ref{j3'}). The
operator $\vr$ obeys the relation
\mar{21f5}\beq
(\vr\circ d_H)(\psi)=0, \qquad \psi\in \cO^{>0,0}_\infty.
\label{21f5}
\eeq
\end{lem}

It follows from Lemma \ref{21t1} that $\vr$ (\ref{r12}) is a
projector, i.e., $\vr\circ \vr=\vr$.

(ii) One defines the variational operator
\mar{21f1}\beq
\dl=\vr\circ d : \cO^{*,1}_\infty \to \cO^{*+1,1}_\infty.
\label{21f1}
\eeq

\begin{lem} \label{21t3} \mar{21t3} The variational operator $\dl$ (\ref{21f1}) is
nilpotent, i.e., $\dl\circ\dl=0$, and it obeys the relation
\mar{xx27}\beq
\dl\circ\vr=\dl. \label{xx27}
\eeq
\end{lem}

With operators $\vr$ (\ref{r12}) and $\dl$ (\ref{21f1}), the
bicomplex $\cO^{*,*}$ is brought into the variational bicomplex.
Let us denote $\bE_k=\vr(\cO^{k,1}_\infty)$. We have
\mar{7}\beq
\begin{array}{ccccrlcrlcrl}
 & &  &  & & \vdots & & & \vdots  & &   & \vdots \\
& & & & _{d_V} & \put(0,-7){\vector(0,1){14}} & & _{d_V} &
\put(0,-7){\vector(0,1){14}}  & & _{-\dl} & \put(0,-7){\vector(0,1){14}} \\
 &  & 0 & \to & &\cO^{1,0}_\infty &\op\to^{d_H} & &
\cO^{1,1}_\infty &\op\to^\vr &  & \bE_1\to  0\\
& & & & _{d_V} &\put(0,-7){\vector(0,1){14}} & & _{d_V} &
\put(0,-7){\vector(0,1){14}}
 & & _{-\dl} & \put(0,-7){\vector(0,1){14}} \\
0 & \to & \Bbb R & \to & & \cO^0_\infty &\op\to^{d_H} & &
\cO^{0,1}_\infty & \equiv &  & \cO^{0,1}_\infty
\end{array}
\label{7}
\eeq
This variational bicomplex possesses the following cohomology
\cite{book09}.

\begin{theo} \label{g90} \mar{g90}
The bottom row and the last column of the variational bicomplex
(\ref{7}) make up the variational complex
\mar{b317}\beq
0\to\Bbb R\to \cO^0_\infty \ar^{d_H}\cO^{0,1}_\infty
\op\longrightarrow^\dl \bE_1 \op\longrightarrow^\dl \bE_2 \ar
\cdots\,. \label{b317}
\eeq
Its cohomology is isomorphic to the de Rham cohomology of a fibre
bundle $Q$ and, consequently, the de Rham cohomology of its
typical fibre $M$ (Corollary \ref{02}).
\end{theo}

\begin{theo} \label{21t7} \mar{21t7}
The rows of contact forms of the variational bicomplex (\ref{7})
are exact sequences.
\end{theo}

Note that Theorem \ref{21t7} gives something more. Due to the
relations (\ref{xx28}) and (\ref{xx27}), we have the cochain
morphism
\be
\begin{array}{rlcrlcrlcrlcl}
 & \cO^0_\infty & \op\to^d & & \cO^1_\infty & \op\to^d & &
\cO^2_\infty & \op\to^d & &
\cO^3_\infty & \to & \cdots \\
_{h_0} & \put(0,10){\vector(0,-1){20}} & & _{h_0} &
\put(0,10){\vector(0,-1){20}}
 & & _\vr & \put(0,10){\vector(0,-1){20}} & & _\vr
& \put(0,10){\vector(0,-1){20}} & &\\
 & \cO^{0,0}_\infty & \op\to^{d_H} & & \cO^{0,1}_\infty &
\op\to^\dl & & \bE_1 & \op\to^\dl & &
\bE_2 & \longrightarrow & \cdots \\
\end{array}
\ee
of the de Rham complex (\ref{5.13}) of the differential graded
algebra $\cO^*_\infty $ to its variational complex (\ref{b317}).
By virtue of Theorems \ref{j4} and \ref{g90}, the corresponding
homomorphism of their cohomology groups is an isomorphism. A
consequence of this fact is the following.

\begin{theo} \label{t41} \mar{t41}
Any $\dl$-closed form $\f\in\cO^{k,1}$, $k=0,1$, is split into the
sum
\mar{t42a,b}\ben
&& \f=h_0\si +
d_H\xi, \qquad k=0, \qquad \xi\in \cO^{0,0}_\infty ,
\label{t42a}\\
&& \f=\vr(\si) +\dl(\xi), \qquad k=1, \qquad \xi\in
\cO^{0,1}_\infty , \label{t42b}
\een
where $\si$ is a closed $(1+k)$-form on $Q$.
\end{theo}

In Lagrangian formalism on a fibre bundle $Q\to\Bbb R$, a finite
order Lagrangian and its Lagrange operator are defined as elements
\mar{21f10-12}\ben
&& L=\cL dt\in \cO^{0,1}_\infty, \label{21f10}\\
&& \cE_L=\dl L=\cE_i\thh^i\w dt \in \bE_1, \label{21f11}\\
&& \cE_i=\op\sum_{0\leq|\La|}(-1)^{|\La|}d_\La(\dr^\La_i \cL),
\label{21f12}
\een
of the variational complex (\ref{b317}). Components $\cE_i$
(\ref{21f12}) of the Lagrange operator (\ref{21f11}) are called
the variational derivatives. Elements of $\bE_1$ are called the
Lagrange-type operators.

We agree to call a pair $(\cO^*_\infty, L)$ the Lagrangian system.

\begin{cor} \label{lmp112'} \mar{lmp112'}
A finite order Lagrangian $L$ (\ref{21f10}) is variationally
trivial, i.e., $\dl(L)=0$ iff
\mar{tams3}\beq
L=h_0\si + d_H \xi, \qquad \xi\in \cO^{0,0}_\infty, \label{tams3}
\eeq
where $\si$ is a closed one-form on $Q$.
\end{cor}

\begin{cor} \label{21t12} \mar{21t12}
A finite order Lagrange-type operator $\cE\in \bE_1$ satisfies the
Helmholtz condition $\dl(\cE)=0$ iff
\mar{xx30}\beq
\cE=\dl L + \vr(\si), \qquad L\in\cO^{0,1}_\infty, \label{xx30}
\eeq
where $\si$ is a closed two-form on $Q$.
\end{cor}

Given a Lagrangian $L$ (\ref{21f10}) and its Lagrange operator
$\dl L$ (\ref{21f11}), the kernel $\Ker\dl L\subset J^{2r}Q$ of
$\dl L$ is called the Lagrange equation. It is locally given by
the equalities
\mar{21f50}\beq
\cE_i=\op\sum_{0\leq|\La|}(-1)^{|\La|}d_\La(\dr^\La_i \cL)=0.
\label{21f50}
\eeq
However, it may happen that the Lagrange equation is not a
differential equation in accordance with Definition \ref{mos04}
because $\Ker\dl L$ need not be a closed subbundle of $J^{2r}Q\to
\Bbb R$.

\begin{ex}
Let $Q=\Bbb R^2\to\Bbb R$ be a configuration space, coordinated by
$(t,q)$. The corresponding velocity phase space $J^1Q$ is equipped
with the adapted coordinates $(t,q,q_t)$. The Lagrangian
\be
L=\frac12 q^2q_t^2dt
\ee
on $J^1Q$ leads to the Lagrange operator
\be
\cE_L=[qq^2_t-d_t(q^2q_t)]\ol dq\w dt
\ee
whose kernel is not a submanifold at the point $q=0$.
\end{ex}

\begin{theo} \label{g93} \mar{g93}
Owing to the exactness of the row of one-contact forms of the
variational bicomplex (\ref{7}) at the term $\cO^{1,1}_\infty$,
there is the decomposition
\mar{+421}\beq
dL=\dl L-d_H\gL, \label{+421}
\eeq
where a one-form $\gL$  is a Lepage equivalent of a Lagrangian $L$
\cite{book09}.
\end{theo}

Let us restrict our consideration to first order Lagrangian theory
on a fibre bundle $Q\to \Bbb R$. This is the case of Lagrangian
non-relativistic mechanics.

A first order Lagrangian is defined as a density
\mar{23f2}\beq
L=\cL dt, \qquad \cL: J^1Q\to \Bbb R, \label{23f2}
\eeq
on a velocity space $J^1Q$. The corresponding second-order
Lagrange operator (\ref{21f11}) reads
\mar{305}\beq
\dl L= (\dr_i\cL- d_t\dr^t_i\cL) \thh^i\w dt. \label{305}
\eeq
Let us further use the notation
\mar{03}\beq
\pi_i=\dr^t_i\cL, \qquad \pi_{ji}=\dr_j^t\dr_i^t\cL. \label{03}
\eeq

The kernel $\Ker\dl L\subset J^2Q$ of the Lagrange operator
defines the second order Lagrange equation
\mar{b327}\beq
(\dr_i- d_t\dr^t_i)\cL=0. \label{b327}
\eeq
Its solutions are (local) sections $c$ of the fibre bundle
$Q\to\Bbb R$ whose second order jet prolongations $\ddot c$ live
in (\ref{b327}). They obey the equations
\beq
\dr_i\cL\circ \dot c- \frac{d}{dt}(\pi_i\circ \dot c)=0.
\label{z333}
\eeq

\begin{defi}
Given a Lagrangian $L$, a holonomic connection
\be
\xi_L=\dr_t + q^i_t\dr_i + \xi^i \dr_i^t
\ee
on the jet bundle $J^1Q\to \Bbb R$ is said to be the Lagrangian
connection if it takes its values into the kernel of the Lagrange
operator $\dl L$, i.e., if it satisfies the relation
\mar{2234}\beq
\dr_i\cL - \dr_t\pi_i - q_t^j\dr_j\pi_i - \xi^j\pi_{ji} =0.
\label{2234}
\eeq
\end{defi}

A Lagrangian connection need not be unique.

Let us bring the relation (\ref{2234}) into the form
\mar{2234'}\beq
\dr_i\cL - d_t\pi_i + (q^i_{tt} - \xi^j)\pi_{ji} =0. \label{2234'}
\eeq
If a Lagrangian connection $\xi_L$ exists, it defines the second
order dynamic equation
\mar{cc205}\beq
q^i_{tt}=\xi_L \label{cc205}
\eeq
on $Q\to \Bbb R$, whose solutions also are solutions of the
Lagrange equation (\ref{b327}) by virtue of the relation
(\ref{2234'}). Conversely, since the jet bundle $J^2Q\to J^1Q$ is
affine, every solution $c$ of the Lagrange equation also is an
integral section for a holonomic connection $\xi$, which is a
global extension of the local section $J^1c(\Bbb R)\to J^2c(\Bbb
R)$ of this jet bundle over the closed imbedded submanifold
$J^1c(\Bbb R)\subset J^1Q$. Hence, every solution of the Lagrange
equation also is a solution of some second order dynamic equation,
but it is not necessarily a Lagrangian connection.

Every first order Lagrangian $L$ (\ref{23f2}) yields the bundle
morphism
\mar{a303}\beq
\wh L:J^1Q\ar_Q V^*Q,\qquad  p_i \circ\wh L = \pi_i, \label{a303}
\eeq
where $(t,q^i,p_i)$ are holonomic coordinates on the vertical
cotangent bundle $V^*Q$ of $Q\to\Bbb R$. This morphism is called
the Legendre map, and
\mar{z400}\beq
\pi_\Pi:V^*Q\to Q, \label{z400}
\eeq
is called the Legendre bundle. As was mentioned above, the
vertical cotangent bundle $V^*Q$ plays a role of the phase space
of mechanics on a configuration space $Q\to\Bbb R$. The range
$N_L=\wh L(J^1Q)$ of the Legendre map (\ref{a303}) is called the
Lagrangian constraint space.

\begin{defi} \label{d11} \mar{d11}
A Lagrangian $L$ is said to be:

$\bullet$ hyperregular if the Legendre map $\wh L$ is a
diffeomorphism;

$\bullet$ regular if $\wh L$ is a local diffeomorphism, i.e.,
$\det(\pi_{ij})\neq 0$;

$\bullet$ semiregular if the inverse image $\wh L^{-1}(p)$ of any
point $p\in N_L$ is a connected submanifold of $J^1Q$;

$\bullet$ almost regular if a Lagrangian constraint space $N_L$ is
a closed imbedded subbundle $i_N:N_L\to V^*Q$ of the Legendre
bundle $V^*Q\to Q$ and the Legendre map
\mar{cmp12}\beq
\wh L:J^1Q\to N_L \label{cmp12}
\eeq
is a fibred manifold with connected fibres (i.e., a Lagrangian is
semiregular).
\end{defi}

\begin{rem} \label{xx34} \mar{xx34}
A glance at the equation (\ref{2234}) shows that a regular
Lagrangian $L$ admits a unique Lagrangian connection
\mar{04}\beq
\xi_L=(\pi^{-1})^{ij}(-\dr_i\cL +\dr_t\pi_i +q^k_t\dr_k\pi_i).
\label{04}
\eeq
In this case, the Lagrange equation (\ref{b327}) for $L$ is
equivalent to the second order dynamic equation associated to the
Lagrangian connection (\ref{04}).
\end{rem}

\section{Cartan and Hamilton--De Donder equations}

Given a first order Lagrangian $L$, its Lepage equivalent $\gL$ in
the decomposition (\ref{+421}) is the Poincar\'e--Cartan form
\mar{303}\beq
H_L=\pi_i dq^i -(\pi_iq^i_t-\cL)dt \label{303}
\eeq
(see the notation (\ref{03})). This form takes its values into the
subbundle $J^1Q\op\times_Q T^*Q$ of $T^*J^1Q$. Hence, we have a
morphism
\mar{N41}\beq
\wh H_L: J^1Q\to T^*Q, \label{N41}
\eeq
whose range
\mar{23f10}\beq
Z_L= \wh H_L(J^1Q) \label{23f10}
\eeq
is an imbedded subbundle $i_L:Z_L\to T^*Q$ of the cotangent bundle
$T^*Q$. This morphism is called the homogeneous Legendre map. Let
$(t,q^i,p_0,p_i)$ denote the holonomic coordinates of $T^*Q$
possessing transition functions
\mar{2.3'}\beq
{p'}_i = \frac{\dr q^j}{\dr{q'}^i}p_j, \qquad p'_0=
\left(p_0+\frac{\dr q^j}{\dr t}p_j\right). \label{2.3'}
\eeq
With respect to these coordinates, the morphism $\wh H_L$
(\ref{N41}) reads
\be
(p_0,p_i)\circ \wh H_L =(\cL-q^i_t\p_i, \p_i).
\ee

A glance at the transition functions (\ref{2.3'}) shows that
$T^*Q$ is a one-dimensional affine bundle
\mar{b418'}\beq
\zeta:T^*Q\to V^*Q\label{b418'}
\eeq
over the vertical cotangent bundle $V^*Q$ (cf. (\ref{z11z})).
Moreover, the Legendre map $\wh L$ (\ref{a303}) is exactly the
composition of morphisms
\mar{m11'}\beq
\wh L=\zeta\circ H_L:J^1Q \op\to_Q V^*Q. \label{m11'}
\eeq

It is readily observed that the Poincar\'e--Cartan form $H_L$
(\ref{303}) also is the Poincar\'e--Cartan form $H_L=H_{\wt L}$ of
the first order Lagrangian
\mar{cmp80}\beq
\wt L=\wh h_0(H_L) = (\cL + (q_{(t)}^i - q_t^i)\p_i)dt, \qquad \wh
h_0(dq^i)=q^i_{(t)} dt,\label{cmp80}
\eeq
on the repeated jet manifold $J^1J^1Y$ \cite{book,book09}. The
Lagrange operator for $\wt L$ (called the Lagrange--Cartan
operator) reads
\mar{2237}\beq
\dl\wt L  = [(\dr_i\cL - \wh d_t\p_i + \dr_i\p_j(q_{(t)}^j -
q_t^j))dq^i +  \dr_i^t\p_j(q_{(t)}^j - q_t^j) dq_t^i]\w dt.
\label{2237}
\eeq
Its kernel $\Ker\dl \ol L\subset J^1J^1Q$ defines the Cartan
equation
\mar{b336}\ben
&& \dr_i^t\p_j(q_{(t)}^j - q_t^j)=0, \label{b336a}\\
&& \dr_i\cL - \wh d_t\p_i
+ \dr_i\p_j(q_{(t)}^j - q_t^j)=0 \label{b336b}
\een
on $J^1Q$. Since $\dl\wt L|_{J^2Q}=\dl L$, the Cartan equation
(\ref{b336a}) -- (\ref{b336b}) is equivalent to the Lagrange
equation (\ref{b327}) on integrable sections of $J^1Q\to X$. It is
readily observed that these equations are equivalent if a
Lagrangian $L$ is regular.

The Cartan equation (\ref{b336a}) -- (\ref{b336b}) on sections
$\ol c: \Bbb R\to J^1Q$ is equivalent to the relation
\mar{C28}\beq
\ol c^*(u\rfloor dH_L)=0, \label{C28}
\eeq
which is assumed to hold for all vertical vector fields $u$ on
$J^1Q\to \Bbb R$.

The cotangent bundle $T^*Q$ admits the Liouville form
\mar{N43}\beq
\Xi= p_0dt + p_i dq^i. \label{N43}
\eeq
Accordingly, its  imbedded subbundle $Z_L$ (\ref{23f10}) is
provided with the pull-back De Donder form $\Xi_L=i^*_L\Xi$. There
is the equality
\mar{cmp14}\beq
H_L=\wh H_L^*\Xi_L=\wh H_L^*(i_L^*\Xi).  \label{cmp14}
\eeq
By analogy with the Cartan equation (\ref{C28}), the Hamilton--De
Donder equation for sections $\ol r$ of $Z_L\to \Bbb R$ is written
as
\mar{N46}\beq
\ol r^*(u\rfloor d\Xi_L)=0, \label{N46}
\eeq
where $u$ is an arbitrary vertical vector field on $Z_L\to \Bbb
R$.

\begin{theo}\label{ddd} \mar{ddd} Let the homogeneous Legendre map
$\wh H_L$ be a submersion. Then a section $\ol c$ of $J^1Q\to \Bbb
R$ is a solution of the Cartan equation (\ref{C28}) iff $\wh
H_L\circ\ol c$ is a solution of the Hamilton--De Donder equation
(\ref{N46}), i.e., the Cartan and Hamilton--De Donder equations
are quasi-equivalent \cite{book09,got91}.
\end{theo}

\section{Lagrangian and Newtonian systems}

Let $L$ be a Lagrangian on a velocity space $J^1Q$ and $\wh L$ the
Legendre map (\ref{a303}). Due to the vertical splitting
(\ref{48}) of $VV^*Q$, the vertical tangent map $V\wh L$ to $\wh
L$ reads
\be
V\wh L: V_QJ^1Q\to V^*Q\op\times_QV^*Q.
\ee
It yields the linear bundle morphism
\mar{z350}\beq
\wh m= (\Id_{J^1Q},\pr_2\circ V\wh L): V_QJ^1Q\to V_Q^*J^1Q, \quad
\wh m: \dr_i^t\mapsto \pi_{ij}\ol dq^j_t, \label{z350}
\eeq
and consequently a fibre metric
\be
m:J^1Q\to \op\vee^2_{J^1Q}V_Q^*J^1Q
\ee
in the vertical tangent bundle $V_QJ^1Q\to J^1Q$. This fibre
metric $m$ obviously satisfies the symmetry condition
(\ref{a1.103}).

Let a Lagrangian $L$ be regular. Then the fibre metric $m$
(\ref{z350}) is non-degenerate. In accordance with Remark
\ref{xx34}, if a Lagrangian $L$ is regular, there exists a unique
Lagrangian connection $\xi_L$ for $L$ which obeys the equality
\mar{a1.91}\beq
m_{ik}\xi^k_L  +\dr_t\pi_i + \dr_j\pi_iq^j_t -\dr_i\cL=0.
\label{a1.91}
\eeq
The derivation of this equality with respect to $q^j_t$ results in
the relation (\ref{a1.95}). Thus, any regular Lagrangian $L$
defines a Newtonian system characterized by the mass tensor
$m_{ij}=\pi_{ij}$.

Now let us investigate the conditions for a Newtonian system to be
the Lagrangian one.

The equation (\ref{z355}) is the kernel of the second order
differential Lagrange type operator
\mar{z361}\beq
\cE: J^2Q\to V^*Q, \qquad \cE=m_{ik}(\xi^k-q^k_{tt})\thh^i\w
dt.\label{z361}
\eeq
A glance at the variational complex (\ref{b317}) shows that this
operator is a Lagrange operator of some Lagrangian only if it
obeys the Helmholtz condition
\be
&&\dl(\cE_i\thh^i\w dt)
= [(2\dr_j -d_t\dr^t_j +d_t^2\dr^{tt}_j)\cE_i\thh^j\w\thh^i + \\
&& \qquad
(\dr^t_j\cE_i + \dr^t_i\cE_j -
2d_t\dr^{tt}_j\cE_i)\thh^i_t\w\thh^j + (\dr^{tt}_j\cE_i
-\dr^{tt}_i\cE_j)\thh^j_{tt}\w\thh^i]\w dt =0.
\ee
This condition falls into the equalities
\mar{b361}\ben
&& \dr_j\cE_i- \dr_i\cE_j +\frac12d_t(\dr^t_i\cE_j-\dr^t_j\cE_i) =0,
\label{b361a} \\
&& \dr^t_j\cE_i + \dr^t_i\cE_j -2d_t\dr^{tt}_j\cE_i=0, \label{b361b}\\
&& \dr^{tt}_j\cE_i -\dr^{tt}_i\cE_j=0. \label{b361c}
\een
It is readily observed, that the condition (\ref{b361c}) is
satisfied since the mass tensor is symmetric. The condition
(\ref{b361b}) holds due to the equality (\ref{a1.95}) and the
property (\ref{a1.103}). Thus, it is necessary to verify the
condition (\ref{b361a}) for a Newtonian system to be a Lagrangian
one. If this condition holds, the operator $\cE$ (\ref{z361})
takes the form (\ref{xx30}) in accordance with Corollary
\ref{21t12}. If the second de Rham cohomology of $Q$ (or,
equivalently, $M$) vanishes, this operator is a Lagrange operator.

\begin{ex} \label{qqq} \mar{qqq}
Let us consider a one-dimensional motion of a point mass $m_0$
subject to friction. It is described by the equation
\beq
m_0q_{tt}=-kq_t, \qquad k> 0, \label{z369}
\eeq
on the configuration space $\Bbb R^2\to\Bbb R$ coordinated by
$(t,q)$. This mechanical system is characterized by the mass
function $m=m_0$ and the holonomic connection
\beq
\xi=\dr_t +q_t\dr_q -\frac{k}{m}q_t\dr^t_q, \label{z371}
\eeq
but it  is neither a Newtonian nor a Lagrangian system. The
conditions (\ref{b361a}) and (\ref{b361c}) are satisfied for an
arbitrary mass function $m(t,q,q_t)$, whereas the conditions
(\ref{a1.95}) and (\ref{b361b}) take the form
\beq
-kq_t\dr_q^tm -km +\dr_tm +q_t\dr_qm=0. \label{zzzz}
\eeq
The mass function  $m=$ const. fails to satisfy this relation.
Nevertheless, the equation (\ref{zzzz}) has a solution
\beq
m=m_0\exp\left[\frac{k}{m_0}t\right]. \label{z370}
\eeq
The mechanical system characterized by the mass function
(\ref{z370}) and the holonomic connection (\ref{z371}) is both a
Newtonian and Lagrangian system with the Havas Lagrangian
\mar{gm22}\beq
\cL=\frac12 m_0\exp\left[\frac{k}{m_0}t\right]q_t^2  \label{gm22}
\eeq
\cite{rie}. The corresponding Lagrange equation is equivalent to
the equation of motion (\ref{z369}).
\end{ex}

In conclusion, let us mention mechanical systems whose motion
equations are Lagrange equations plus additional non-Lagrangian
external forces. They read
\mar{z012}\beq
(\dr_i- d_t\dr^t_i)\cL +f_i(t,q^j,q^j_t)=0. \label{z012}
\eeq
Let a Lagrangian system be the Newtonian one, and let an external
force $f$ satisfy the condition (\ref{gm383}). Then the equation
(\ref{z012}) describe a Newtonian system.

\section{Lagrangian conservation laws}

In Lagrangian mechanics, integrals of motion come from variational
symmetries of a Lagrangian (Theorem \ref{035'}) in accordance with
the first Noether theorem (Theorem \ref{j22}). However, not all
integrals of motion are of this type (Example \ref{048}).

\subsection{Generalized vector fields}

Given a Lagrangian system $(\cO^*_\infty, L)$ on a fibre bundle
$Q\to\Bbb R$, its infinitesimal transformations are defined to be
contact derivations of the real ring $\cO^0_\infty$
\cite{cmp04,book09,book10}.

Let us consider the $\cO^0_\infty$-module $\gd\cO^0_\infty$ of
derivations of the real ring $\cO^0_\infty$. This module is
isomorphic to the $\cO^0_\infty$-dual $(\cO^1_\infty)^*$ of the
module of one-forms $\cO^1_\infty$. Let $\vt\rfloor\f$, $\vt\in
\gd\cO^0_\infty$, $\f\in \cO^1_\infty$, be the corresponding
interior product.  Extended to a differential graded algebra
$\cO^*_\infty$, it obeys the rule (\ref{031}).

Restricted to the coordinate chart (\ref{j3'}), any derivation of
a real ring $\cO^0_\infty$ takes the coordinate form
\be
\vt=\vt^t \dr_t + \vt^i\dr_i + \op\sum_{0<|\La|}\vt^i_\La
\dr^\La_i,
\ee
where
\be
\dr^\La_i(q_\Si^j)=\dr^\La_i\rfloor dq_\Si^j=\dl_i^j\dl^\La_\Si.
\ee
Not considering time reparametrization, we can restrict our
consideration to derivations
\mar{g3}\beq
\vt=u^t\dr_t + \vt^i\dr_i + \op\sum_{0<|\La|}\vt^i_\La \dr^\La_i,
\qquad u^t=0,1. \label{g3}
\eeq
Their coefficients $\vt^i$, $\vt^i_\La$ possess the transformation
law
\be
\vt'^i=\frac{\dr q'^i}{\dr q^j}\vt^j + \frac{\dr q'^i}{\dr t}u^t,
\qquad
&& \vt'^i_\La=\op\sum_{|\Si|\leq|\La|}\frac{\dr q'^i_\La}{\dr
q^j_\Si}\vt^j_\Si + \frac{\dr q'^i_\La}{\dr t}u^t.
\ee

Any derivation $\vt$ (\ref{g3}) of a ring $\cO^0_\infty$ yields a
derivation (a Lie derivative) $\bL_\vt$ of a differential graded
algebra $\cO^*_\infty$ which obeys the relations (\ref{034}) --
(\ref{035}).

A derivation $\vt\in \gd \cO^0_\infty$ (\ref{g3}) is called
contact if the Lie derivative $\bL_\vt$ preserves an ideal of
contact forms of a differential graded algebra $\cO^*_\infty$,
i.e., the Lie derivative $\bL_\vt$ of a contact form is a contact
form.

\begin{lem}
A derivation $\vt$ (\ref{g3}) is contact iff it takes the form
\mar{g4}\beq
\vt=u^t\dr_t +u^i\dr_i +\op\sum_{0<|\La|}
[d_\La(u^i-q^i_tu^t)+q^i_{t\La}u^t]\dr^\La_i. \label{g4}
\eeq
\end{lem}

A glance at the expression (\ref{g4}) enables one to regard a
contact derivation $\vt$ as an infinite order jet prolongation
$\vt=J^\infty u$ of its restriction
\mar{j15}\beq
u=u^t\dr_t +u^i(t,q^i,q^i_\La)\dr_i, \qquad u^t=1,0, \label{j15}
\eeq
to a ring $C^\infty(Q)$. Since coefficients $u^i$ of $u$
(\ref{j15}) generally depend on jet coordinates $q^i_\La$,
$0<|\La|\leq r$, one calls $u$ (\ref{j15}) the generalized vector
field. It can be represented as a section of the pull-back bundle
\be
J^rQ\op\times_Q TQ \to J^rQ.
\ee

In particular, let $u$ (\ref{j15}) be a vector field
\mar{z372}\beq
u=u^t\dr_t +u^i(t,q^i)\dr_i, \qquad u^t=0,1, \label{z372}
\eeq
on a configuration space $Q\to\Bbb R$. One can think of this
vector field as being an infinitesimal generator of a local
one-parameter group of local automorphisms of a fibre bundle
$Q\to\Bbb R$. If $u^t=0$, the vertical vector field (\ref{z372})
is an infinitesimal generator of a local one-parameter group of
local vertical automorphisms of $Q\to\Bbb R$. If $u^t=1$, the
vector field $u$ (\ref{z372}) is projected onto the standard
vector field $\dr_t$ on a base $\Bbb R$ which is an infinitesimal
generator of a group of translations of $\Bbb R$.

Any contact derivation $\vt$ (\ref{g4}) admits the horizontal
splitting
\mar{g5,'}\ben
&&\vt=\vt_H +\vt_V=u^t d_t + \left[u_V^i\dr_i +
\op\sum_{0<|\La|}
d_\La u_V^i\dr_i^\La\right], \label{g5}\\
&& u=u_H+ u_V= u^t(\dr_t +q^i_t\dr_i)+(u^i-q^i_tu^t)\dr_i.
\label{g5'}
\een

\begin{lem}
Any vertical contact derivation
\mar{gg6}\beq
\vt=u^i\dr_i +\op\sum_{0<|\La|} d_\La u^i\dr_i^\La \label{gg6}
\eeq
obeys the relations
\mar{g6}\beq
\vt\rfloor d_H\f=-d_H(\vt\rfloor\f), \qquad
\bL_\vt(d_H\f)=d_H(\bL_\vt\f), \qquad \f\in\cO^*_\infty.
\label{g6}
\eeq
\end{lem}

We restrict our consideration to first order Lagrangian mechanics.
In this case, contact derivations (\ref{j4}) can be reduced to the
first order jet prolongation
\mar{23f41}\beq
\vt=J^1u= u^t\dr_t + u^i\dr_i + d_t u^i\dr^t_i \label{23f41}
\eeq
of the generalized vector fields $u$ (\ref{j15}).

\subsection{First Noether theorem}

Let $L$ be a Lagrangian (\ref{23f2}) on a velocity space $J^1Q$.
Let us consider its Lie derivative $\bL_\vt L$ along the contact
derivation $\vt$ (\ref{23f41}).

\begin{theo}  \label{g75} \mar{g75}
The Lie derivative $\bL_\vt L$ fulfils the first variational
formula
\mar{23f42}\beq
\bL_{J^1u}L= u_V\rfloor\dl L + d_H(u\rfloor H_L), \label{23f42}
\eeq
where $\gL=H_L$ is the Poincar\'e--Cartan form (\ref{303}). Its
coordinate expression reads
\mar{J4}\beq
[u^t\dr_t+ u^i\dr_i +d_t u^i\dr^t_i]\cL = (u^i-q^i_t u^t)\cE_i +
d_t[\pi_i(u^i-u^t q^i_t) +u^t\cL]. \label{J4}
\eeq
\end{theo}

The generalized vector field $u$ (\ref{j15}) is said to be the
variational symmetry of a Lagrangian $L$ if the Lie derivative
$\bL_{J^1u} L$ is $d_H$-exact, i.e.,
\mar{22f1}\beq
\bL_{J^1u} L=d_H\si. \label{22f1}
\eeq
Variational symmetries of $L$ constitute a real vector space which
we denote $\ccG_L$.

\begin{prop} \label{0133} \mar{0133} A glance at the first
variational formula (\ref{J4}) shows that a generalized vector
field $u$ is a variational symmetry iff the exterior form
\mar{0134}\beq
u_V\rfloor\dl L=(u^i-q^i_t u^t)\cE_i dt \label{0134}
\eeq
is $d_H$-exact.
\end{prop}

\begin{prop} \label{0127} \mar{0127}
The generalized vector field $u$ (\ref{j15}) is a variational
symmetry of a Lagrangian $L$ iff its vertical part $u_V$
(\ref{g5'}) also is a variational symmetry.
\end{prop}

\begin{proof} A direct computation shows that
\mar{0128}\beq
\bL_{J^1u} L = \bL_{J^1u_V} L +d_H(u^t\cL). \label{0128}
\eeq
\end{proof}

A corollary of the first variational formula (\ref{23f42}) is the
first Noether theorem.

\begin{theo} \label{j22} \mar{j22} If a contact derivation $\vt$
(\ref{g4}) is a variational symmetry (\ref{22f1}) of a Lagrangian
$L$, the first variational formula (\ref{23f42}) restricted to the
kernel of the Lagrange operator $\Ker\dl L$ yields a weak
conservation law
\mar{22f2,'}\ben
&& 0\ap d_H(u\rfloor H_L -\si), \label{22f2}\\
&& 0\ap d_t(\pi_i(u^i- u^t q^i_t) +u^t\cL -\si), \label{22f2'}
\een
of the generalized symmetry current
\mar{m225'}\beq
\gT_u=u\rfloor H_L-\si =\pi_i(u^i-u^tq^i_t) + u^t\cL-\si
\label{m225'}
\eeq
along a generalized vector field $u$. The generalized symmetry
current (\ref{m225'}) obviously is defined with the accuracy of a
constant summand.
\end{theo}

The weak conservation law (\ref{22f2}) on the shell $\dl L=0$ is
called the Lagrangian conservation law. It leads to the
differential conservation law (\ref{020}):
\be
0 = \frac{d}{dt}[\gT_u\circ J^{r+1}c],
\ee
on solutions $c$ of the Lagrange equation (\ref{z333}).

\begin{prop} \label{0130} \mar{0130} Let $u$ be a variational symmetry of a Lagrangian $L$.
By virtue of Proposition \ref{0127}, its vertical part $u_V$ is
so. It follows from the equality (\ref{0128}) that the conserved
generalized symmetry current $\gT_u$ (\ref{m225'}) along $u$
equals that $\gT_{u_V}$ along $u_V$.
\end{prop}

A glance at the conservation law (\ref{22f2'}) shows the
following.

\begin{theo} \label{035'} \mar{035'}
If a variational symmetry $u$ is a generalized vector field
independent of higher order jets $q^i_\La$, $|\La|>1$, the
conserved generalized current $\gT_u$ (\ref{m225'}) along $u$
plays a role of an integral of motion.
\end{theo}

Therefore, we further restrict our consideration to variational
symmetries at most of first jet order for the purpose of obtaining
integrals of motion. However, it may happen that a Lagrangian
system possesses integrals of motion which do not come from
variational symmetries (Example \ref{048}).

A variational symmetry $u$ of a Lagrangian $L$ is called its exact
symmetry if
\mar{xx45}\beq
\bL_{J^1u}L=0. \label{xx45}
\eeq
In  this case, the first variational formula (\ref{23f42}) takes
the form
\mar{xx46}\beq
0= u_V\rfloor\dl L + d_H(u\rfloor H_L). \label{xx46}
\eeq
It leads to the weak conservation law (\ref{22f2}):
\mar{gm488}\beq
0\ap d_t\gT_u, \label{gm488}
\eeq
of the symmetry current
\mar{m225}\beq
\gT_u=u\rfloor H_L=\pi_i(u^i-u^tq^i_t) + u^t\cL \label{m225}
\eeq
along a generalized vector field $u$.

\begin{rem} In accordance with the standard terminology,
if variational and exact symmetries are generalized vector fields
(\ref{j15}), they are called generalized symmetries
\cite{bry,fat,ibr,olv}. Accordingly, by variational and exact
symmetries one means only vector fields $u$ (\ref{z372}) on $Q$.
We agree to call them classical symmetries. Classical exact
symmetries are symmetries of a Lagrangian, and they are named the
Lagrangian symmetries.
\end{rem}

\begin{rem} \label{060} \mar{060}
Let us describe the relation between symmetries of a Lagrangian
and and symmetries of the corresponding Lagrange equation.  Let
$u$ be the vector field (\ref{z372}) and
\be
J^2u =u^t\dr_t + u^i\dr_i + d_tu^i\dr_i^t +d_{tt}u^i\dr_i^{tt}
\ee
its second order jet prolongation. Given a Lagrangian $L$ on
$J^1Q$, the relation
\mar{22f9}\beq
\bL_{J^2u}\dl L=\dl(\bL_{J^1u} L) \label{22f9}
\eeq
holds \cite{book,olv}. Note that this equality need not be true in
the case of a generalized vector field $u$. A vector field $u$ is
called the local variational symmetry of a Lagrangian $L$ if the
Lie derivative $\bL_{J^1u} L$ of $L$ along $u$ is variationally
trivial, i.e.,
\be
\dl(\bL_{J^1u} L) =0.
\ee
Then it follows from the equality (\ref{22f9}) that a local
(classical) variational symmetry of $L$ also is a symmetry of the
Lagrange operator $\dl L$, i.e.,
\be
\bL_{J^2u}\dl L=0,
\ee
and {\it vice versa}. Consequently, any local classical
variational symmetry $u$ of a Lagrangian $L$ is a symmetry of the
Lagrange equation (\ref{b327}) in accordance with Proposition
\ref{082}. By virtue of Theorem \ref{g90}, any local classical
variational symmetry is a classical variational symmetry if a
typical fibre $M$ of $Q$ is simply connected.
\end{rem}

\subsection{Noether conservation laws}

It is readily observed that the first variational formula
(\ref{J4}) is linear in a generalized vector field $u$. Therefore,
one can consider superposition of the identities (\ref{J4}) for
different generalized vector fields.

For instance, if $u$ and $u'$ are generalized vector fields
(\ref{j15}), projected onto the standard vector field $\dr_t$ on
$\Bbb R$, the difference of the corresponding identities
(\ref{J4}) results in the first variational formula (\ref{J4}) for
the vertical generalized vector field $u-u'$.

Conversely, every generalized vector field $u$ (\ref{z372}),
projected onto $\dr_t$, can be written as the sum
\mar{gm490}\beq
u=\G +v \label{gm490}
\eeq
of some reference frame
\mar{gm489}\beq
\G=\dr_t +\G^i\dr_i \label{gm489}
\eeq
and a vertical generalized vector field $v$ on $Q$.

It follows that the first variational formula (\ref{J4}) for the
generalized vector field $u$ (\ref{z372}) can be represented as a
superposition of those for a reference frame $\G$ (\ref{gm489})
and a vertical generalized vector field $v$.

If $u=v$ is a vertical generalized vector field, the first
variational formula (\ref{J4}) reads
\be
(v^i\dr_i +d_t v^i \dr^t_i)\cL = v^i\cE_i + d_t(\pi_i v^i).
\ee
If $v$ is an exact symmetry of $L$, we obtain from (\ref{gm488})
the weak conservation law
\mar{z383}\beq
 0\ap d_t(\pi_i v^i). \label{z383}
\eeq
By analogy with field theory \cite{book09,book10}, it is called
the Noether conservation law of the Noether current
\mar{z384}\beq
\gT_v=\pi_i v^i. \label{z384}
\eeq
If a generalized vector field $v$ is independent of higher order
jets $q^i_\La$, $|\La|>1$, the Noether current (\ref{z384}) is an
integral of motion by virtue of Theorem \ref{035'}.

\begin{ex} \label{055} \mar{055}
Let us consider a free motion on a configuration space $Q$. It is
described by a Lagrangian
\mar{gm475}\beq
L=\left(\frac12 \ol m_{ij}\ol q^i_t\ol q^j_t\right)dt, \qquad  \ol
m_{ij}={\rm const.}, \label{gm475}
\eeq
written with respect to a reference frame $(t,\ol q^i)$ such that
the free motion dynamic equation takes the form (\ref{z280}). This
Lagrangian admits $\di Q-1$ independent integrals of motion
$\p_i$.
\end{ex}

\begin{ex} \label{041} \mar{041} Let us consider a point mass
in the presence of a central potential. Its configuration space is
\mar{0155}\beq
Q=\Bbb R\times\Bbb R^3\to\Bbb R \label{0155}
\eeq
endowed with the Cartesian coordinates $(t,q^i)$. A Lagrangian of
this mechanical system reads
\mar{042}\beq
\cL=\frac12 \left(\op\sum_i(q^i_t)^2\right) - V(r), \qquad
r=\left(\op\sum_i(q^i)^2\right)^{1/2}. \label{042}
\eeq
The vector fields
\mar{043}\beq
v^a_b=q^a\dr_b -q^b\dr_a \label{043}
\eeq
are infinitesimal generators of the group $SO(3)$ acting in $\Bbb
R^3$. Their jet prolongation (\ref{23f41}) reads
\mar{044}\beq
J^1v^a_b=q^a\dr_b -q^b\dr_a + q^a_t\dr_b^t -q^b_t\dr_a^t.
\label{044}
\eeq
It is readily observed that vector fields (\ref{043}) are
symmetries of the Lagrangian (\ref{042}). The corresponding
conserved Noether currents (\ref{z384}) are orbital momenta
\mar{045}\beq
M^a_b=\gT^b_a =(q^a \pi_b - q^b\pi_a)=q^aq^b_t- q^b q^a_t.
\label{045}
\eeq
They are integrals of motion, which however fail to be
independent.
\end{ex}

\begin{ex} \label{048} \mar{048}
Let us consider the Lagrangian system in Example \ref{041} where
\mar{049}\beq
V(r)=-\frac1r \label{049}
\eeq
is the Kepler potential. This Lagrangian system possesses the
integrals of motion
\mar{050}\beq
A^a=\op\sum_b(q^aq^b_t -q^bq^a_t)q^b_t -\frac{q^a}{r}, \label{050}
\eeq
besides the orbital momenta (\ref{045}). They are components of
the Rung--Lenz vector. There is no Lagrangian symmetry whose
generalized symmetry currents are $A^a$ (\ref{050}).
\end{ex}

\subsection{Energy conservation laws}

In the case of a reference frame $\G$ (\ref{gm489}), where
$u^t=1$, the first variational formula (\ref{J4}) reads
\mar{m227}\beq
(\dr_t +\G^i\dr_i +d_t\G^i\dr_i^t)\cL = (\G^i-q^i_t)\cE_i -
d_t(\pi_i(q^i_t-\G^i) -\cL), \label{m227}
\eeq
where
\mar{m228}\beq
E_\G=-\gT_\G= \pi_i(q^i_t -\G^i) -\cL \label{m228}
\eeq
is the energy function relative to a reference frame $\G$
\cite{eche95,book10,book98,sard98}.

With respect to the coordinates adapted to a reference frame $\G$,
the first variational formula (\ref{m227}) takes the form
\mar{m229}\beq
\dr_t\cL = (\G^i-q^i_t)\cE_i - d_t(\pi_iq^i_t -\cL),\label{m229}
\eeq
and $E_\G$ (\ref{m228}) coincides with the canonical energy
function
\be
E_L=\pi_iq^i_t -\cL.
\ee
A glance at the expression (\ref{m229}) shows that the vector
field $\G$ (\ref{gm489}) is an exact symmetry of a Lagrangian $L$
iff, written with respect to coordinates adapted to $\G$, this
Lagrangian is independent on the time $t$. In this case, the
energy function $E_\G$ (\ref{m229}) relative to a reference frame
$\G$ is conserved:
\mar{043'}\beq
0\ap -d_t E_\G. \label{043'}
\eeq
It is an integral of motion in accordance with Theorem \ref{035'}.

\begin{ex}
Let us consider a free motion on a configuration space $Q$
described by the Lagrangian (\ref{gm475}) written with respect to
a reference frame $(t,\ol q^i)$ such that the free motion dynamic
equation takes the form (\ref{z280}). Let $\G$ be the associated
connection. Then the conserved energy function $E_\G$ (\ref{m228})
relative  to this reference frame $\G$ is precisely the kinetic
energy of this free motion. With respect to arbitrary bundle
coordinates $(t,q^i)$ on $Q$, it takes the form
\be
E_\G= \pi_i(q^i_t -\G^i) -\cL=\frac12m_{ij}(t,q^k)(
q^i_t-\G^i)(q^j_t-\G^j).
\ee
\end{ex}

\begin{ex}
Let us consider a one-dimensional motion of a point mass $m_0$
subject to friction on the configuration space $\Bbb R^2\to\Bbb
R$, coordinated by $(t,q)$ (Example \ref{qqq}). It is described by
the dynamic equation (\ref{z369}) which is the Lagrange equation
for the Lagrangian $L$ (\ref{gm22}). It is readily observed that
the Lie derivative of this Lagrangian along the vector field
\beq
\G=\dr_t- \frac12\frac{k}{m_0}q\dr_q \label{gm23}
\eeq
vanishes. Consequently, we have the conserved energy function
(\ref{m228}) with respect to the reference frame $\G$
(\ref{gm23}). This energy function reads
\be
E_\G=\frac12m_0\exp\left[\frac{k}{m_0}t\right]
q_t(q_t+\frac{k}{m_0}q)= \frac12m\dot q_\G^2
-\frac{mk^2}{8m_0^2}q^2,
\ee
where $m$ is the mass function (\ref{z370}).
\end{ex}

Since any generalized vector field $u$ (\ref{j15}) can be
represented as the sum (\ref{gm490}) of a reference frame $\G$
(\ref{gm489}) and a vertical generalized vector field $v$, the
symmetry current (\ref{m225}) along the generalized vector field
$u$ (\ref{z372}) is the difference
\be
\gT_u=\gT_v-E_\G
\ee
of the  Noether current $\gT_v$ (\ref{z384}) along the vertical
generalized vector field $v$ and the energy function $E_\G$
(\ref{m228}) relative to a reference frame $\G$
\cite{eche95,book10,sard98}. Conversely, energy functions relative
to different reference frames $\G$ and $\G'$ differ from each
other in the Noether current along the vertical vector field
$\G'-\G$:
\be
E_\G-E_{\G'}=\gT_{\G-\G'}.
\ee
One can regard this vector field $\G'-\G$ as the relative velocity
of a reference frame $\G'$ with respect to $\G$.

\section{Gauge symmetries}

Treating gauge symmetries of Lagrangian field theory, one is
traditionally based on an example of the Yang--Mills gauge theory
of principal connections on a principal bundle. This notion of
gauge symmetries is generalized to Lagrangian theory on an
arbitrary fibre bundle \cite{jmp09,book09}, including mechanics on
a fibre bundle $Q\to\Bbb R$.

\begin{defi} \mar{s7} \label{s7}
Let $E\to \Bbb R$ be a vector bundle and $E(\Bbb R)$ the
$C^\infty(\Bbb R)$ module of sections $\chi$ of $E\to \Bbb R$. Let
$\zeta$ be a linear differential operator on $E(\Bbb R)$ taking
its values into the vector space $\ccG_L$ of variational
symmetries of a Lagrangian $L$. Elements
\mar{gg1}\beq
u_\chi=\zeta(\chi) \label{gg1}
\eeq
of $\im\zeta$ are called the gauge symmetry of a Lagrangian $L$
parameterized by sections $\chi$ of $E\to \Bbb R$. These sections
are called the gauge parameters.
\end{defi}

\begin{rem} \label{22r100} \mar{22r100}
The differential operator $\zeta$ in Definition \ref{s7} takes its
values into the vector space $\ccG_L$  as a subspace of the
$C^\infty(\Bbb R)$-module $\gd \cO^0_\infty$, but it sends the
$C^\infty(\Bbb R)$-module $E(\Bbb R)$ into the real vector space
$\ccG_L\subset \gd \cO^0_\infty$. The differential operator
$\zeta$ is assumed to be at least of first order (Remark
\ref{g1g}).
\end{rem}

Equivalently, the gauge symmetry (\ref{gg1}) is given by a section
$\wt\zeta$ of the fibre bundle
\be
(J^rQ\op\times_Q J^mE)\op\times_Q TQ\to J^rQ\op\times_Q J^mE
\ee
(see Definition \ref{ch538}) such that
\be
u_\chi=\zeta(\chi)=\wt\zeta\circ\chi
\ee
for any section $\chi$ of $E\to \Bbb R$. Hence, it is a
generalized vector field $u_\zeta$ on the product $Q\times E$
represented by a section of the pull-back bundle
\be
J^k(Q\op\times_{\Bbb R} E)\op\times_Q T(Q\op\times_{\Bbb R} E)\to
J^k(Q\op\times_{\Bbb R} E), \qquad k={\rm max}(r,m),
\ee
which lives in
\be
TQ\subset T(Q\times E).
\ee
This generalized vector field yields the contact derivation
$J^\infty u_\zeta$ (\ref{g4}) of the real ring
$\cO^0_\infty[Q\times E]$ which obeys the following condition.

Given a Lagrangian
\be
L\in \cO^{0,n}_\infty E\subset \cO^{0,n}_\infty[Q\times E],
\ee
let us consider its Lie derivative
\mar{gg5}\beq
\bL_{J^\infty u_\zeta} L=J^\infty u_\zeta\rfloor dL + d(J^\infty
u_\zeta\rfloor L) \label{gg5}
\eeq
where $d$ is the exterior differential of $\cO^*_\infty[Q\times
E]$. Then for any section $\chi$ of $E\to \Bbb R$, the pull-back
$\chi^* \bL_{J^\infty u_\zeta} L$ is $d_H$-exact.

It follows at once from the first variational formula
(\ref{23f42}) for the Lie derivative (\ref{gg5}) that the above
mentioned condition holds only if $u_\zeta$ is projected onto a
generalized vector field on $Q$ and, in this case, iff the density
$(u_\zeta)_V\rfloor \cE$ is $d_H$-exact (Proposition \ref{0133}).
Thus, we come to the following equivalent definition of gauge
symmetries.

\begin{defi} \label{s7'} \mar{s7'} Let $E\to \Bbb R$ be a vector
bundle. A gauge symmetry of a Lagrangian $L$ parameterized by
sections $\chi$ of $E\to \Bbb R$ is defined as a contact
derivation $\vt=J^\infty u$ of the real ring $\cO^0_\infty[Q\times
E]$ such that:

(i) it vanishes on the subring $\cO^0_\infty E$,

(ii) the generalized vector field $u$ is linear in coordinates
$\chi^a_\La$ on $J^\infty E$, and it is projected onto a
generalized vector field on $Q$, i.e., it takes the form
\mar{gg2}\beq
u=\dr_t + \left(\op\sum_{0\leq|\La|\leq m}
u^{i\La}_a(t,q^j_\Si)\chi^a_\La\right)\dr_i, \label{gg2}
\eeq

(iii) the vertical part of $u$ (\ref{gg2}) obeys the equality
\mar{gg10}\beq
u_V\rfloor \dl L=d_H\si. \label{gg10}
\eeq
\end{defi}

For the sake of convenience, the generalized vector field
(\ref{gg2}) also is called the gauge symmetry. In accordance with
Proposition \ref{0127}, the $u$ (\ref{gg2}) is a gauge symmetry
iff its vertical part is so. Owing to this fact and Proposition
\ref{0130}, we can restrict our consideration to vertical gauge
symmetries
\mar{gg2'}\beq
u=\left(\op\sum_{0\leq|\La|\leq m}
u^{i\La}_a(t,q^j_\Si)\chi^a_\La\right)\dr_i. \label{gg2'}
\eeq

Gauge symmetries possess the following particular properties.

(i) Let $E'\to \Bbb R$ be another vector bundle and $\zeta'$ a
linear $E(\Bbb R)$-valued differential operator on a
$C^\infty(\Bbb R)$-module $E'(\Bbb R)$ of sections of $E'\to \Bbb
R$. Then
\be
u_{\zeta'(\chi')} =(\zeta\circ\zeta')(\chi')
\ee
also is a gauge symmetry of $L$ parameterized by sections $\chi'$
of $E'\to \Bbb R$. It factorizes through the gauge symmetry
$u_\chi$ (\ref{gg1}).

(ii) Given a gauge symmetry, the corresponding conserved symmetry
current $\gT_u$ (\ref{m225'}) vanishes on-shell (Theorem
\ref{supp} below).

(iii) The second Noether theorem associates to a gauge symmetry of
a Lagrangian $L$ the Noether identities of its Lagrange operator
$\dl L$.

\begin{theo} \label{ggg3} \mar{ggg3}
Let $u$ (\ref{gg2'}) be a gauge symmetry of a Lagrangian $L$, then
its Lagrange operator $\dl L$ obeys the Noether identities
(\ref{gg11}).
\end{theo}

\begin{proof}
The density (\ref{gg10}) is variationally trivial and, therefore,
its variational derivatives with respect to variables $\chi^a$
vanish, i.e.,
\mar{gg11}\beq
\cE_a=\op\sum_{0\leq |\La|}(-1)^{|\La|}d_\La(u^{i\La}_a\cE_i)=0.
\label{gg11}
\eeq
These are the Noether identities for the Lagrange operator $\dl L$
\cite{book09}.
\end{proof}

For instance, if the gauge symmetry $u$ (\ref{gg2}) is of second
jet order in gauge parameters, i.e.,
\mar{0656}\beq
u=(u_a^i\chi^a +u^{it}_a\chi^a_t + u_a^{itt}\chi^a_{tt})\dr_i,
\label{0656}
\eeq
the corresponding Noether identities (\ref{gg11}) take the form
\mar{0657}\beq
u^i_a\cE_i - d_t(u^{it}_a\cE_i) + d_{tt}(u_a^{itt} \cE_i)=0.
\label{0657}
\eeq

\begin{rem} \label{g1g} \mar{g1g}
A glance at the expression (\ref{0657}) shows that, if a gauge
symmetry is independent of derivatives of gauge parameters (i.e.,
the differential operator $\zeta$ in Definition \ref{s7} is of
zero order), then all variational derivatives of a Lagrangian
equals zero, i.e., this Lagrangian is variationally trivial.
Therefore, such gauge symmetries usually are not considered.
\end{rem}

If a Lagrangian $L$ admits a gauge symmetry $u$ (\ref{gg2'}),
i.e., $\bL_{J^1u}L=\si$, the weak conservation law (\ref{22f2'})
of the corresponding generalized symmetry current $\gT_u$
(\ref{m225'}) holds. We call it the gauge conservation law.
Because gauge symmetries depend on derivatives of gauge
parameters, all gauge conservation laws in first order Lagrangian
mechanics possess the following peculiarity.

\begin{theo} \label{supp} \mar{supp}
If $u$ (\ref{gg2'}) is a gauge symmetry of a first order
Lagrangian $L$, the corresponding conserved generalized symmetry
current $\gT_u$ (\ref{m225'}) vanishes on-shell, i.e., $\gT_u\ap
0$ \cite{book09,book10}.
\end{theo}

Note that the statement of Theorem \ref{supp} is a particular case
of the fact that symmetry currents of gauge symmetries in field
theory are reduced to a superpotential \cite{book09,ijgmmp09}.

\chapter{Hamiltonian mechanics}

As was mentioned above, a phase space of non-relativistic
mechanics is the vertical cotangent bundle $V^*Q$ of its
configuration space $Q\to\Bbb R$. This phase space is provided
with the canonical Poisson structure (\ref{m72}). However,
Hamiltonian mechanics on a phase space $V^*Q$ is not familiar
Poisson Hamiltonian theory on a Poisson manifold $V^*Q$ because
all Hamiltonian vector fields on $V^*Q$ are vertical. Hamiltonian
mechanics on $V^*Q$ is formulated as particular (polysymplectic)
Hamiltonian formalism on fibre bundles
\cite{book,book09,book10,book98}. Its Hamiltonian is a section of
the fibre bundle $T^*Q\to V^*Q$ (\ref{b418'}). The pull-back of
the canonical Liouville form (\ref{N43}) on $T^*Q$ with respect to
this section is a Hamiltonian one-form on $V^*Q$. The
corresponding Hamiltonian connection (\ref{z3}) on $V^*Q\to \Bbb
R$ defines a first order Hamilton equations on $V^*Q$.

Note that one can associate to any Hamiltonian system on $V^*Q$ an
autonomous symplectic Hamiltonian system on the cotangent bundle
$T^*Q$ such that the corresponding Hamilton equations on $V^*Q$
and $T^*Q$ are equivalent (Section 3.2). Moreover, a Hamilton
equations on $V^*Q$ also ia equivalent to the Lagrange equation of
a certain first order Lagrangian on a configuration space $V^*Q$
(Section 3.3).

Lagrangian and Hamiltonian formulations of mechanics fail to be
equivalent, unless a Lagrangian is hyperregular. The comprehensive
relations between Lagrangian and Hamiltonian systems can be
established in the case of almost regular Lagrangians (Section
3.4).

\section{Hamiltonian formalism on $Q\to\Bbb R$}

As was mentioned above, a phase space of mechanics on a
configuration space $Q\to\Bbb R$ is the vertical cotangent bundle
(\ref{z400}):
\be
V^*Q\ar^{\pi_\Pi} Q \ar^\pi \Bbb R,
\ee
of $Q\to\Bbb R$ equipped with the holonomic coordinates $(t,
q^i,p_i=\dot q_i)$ with respect to the fibre bases $\{\ol dq^i\}$
for the bundle $V^*Q\to Q$ \cite{book10,book98}.

The cotangent bundle $T^*Q$ of the configuration space $Q$ is
endowed  with the holonomic coordinates $(t,q^i,p_0,p_i)$,
possessing the transition functions (\ref{2.3'}). It admits the
Liouville form $\Xi$ (\ref{N43}), the symplectic form
\mar{m91'}\beq
\Om_T=d\Xi=dp_0\w dt +dp_i\w dq^i, \label{m91'}
\eeq
and the corresponding Poisson bracket
\mar{m116}\beq
\{f,g\}_T =\dr^0f\dr_tg - \dr^0g\dr_tf +\dr^if\dr_ig-\dr^ig\dr_if,
\quad f,g\in C^\infty(T^*Q). \label{m116}
\eeq
Provided with the structures (\ref{m91'}) -- (\ref{m116}), the
cotangent bundle $T^*Q$ of $Q$ plays a role of the homogeneous
phase space of Hamiltonian mechanics.

There is the canonical one-dimensional affine bundle
(\ref{b418'}):
\mar{z11'}\beq
\zeta:T^*Q\to V^*Q. \label{z11'}
\eeq
A glance at the transformation law (\ref{2.3'}) shows that it is a
trivial affine bundle. Indeed, given a global section $h$ of
$\zeta$, one can equip $T^*Q$ with the global fibre coordinate
\mar{09151}\beq
I_0=p_0-h, \qquad I_0\circ h=0, \label{09151}
\eeq
possessing the identity transition functions. With respect to the
coordinates
\mar{09150}\beq
(t,q^i,I_0,p_i), \qquad i=1,\ldots,m, \label{09150}
\eeq
the fibration (\ref{z11'}) reads
\mar{z11}\beq
\zeta: \Bbb R\times V^*Q \ni (t,q^i,I_0,p_i)\to (t,q^i,p_i)\in
V^*Q. \label{z11}
\eeq

Let us consider the subring of $C^\infty(T^*Q)$ which comprises
the pull-back $\zeta^*f$ onto $T^*Q$ of functions $f$ on the
vertical cotangent bundle $V^*Q$ by the fibration $\zeta$
(\ref{z11'}).  This subring is closed under the Poisson bracket
(\ref{m116}). Then by virtue of the well known theorem, there
exists the degenerate Poisson structure
\mar{m72}\beq
\{f,g\}_V = \dr^if\dr_ig-\dr^ig\dr_if, \qquad f,g\in
C^\infty(V^*Q), \label{m72}
\eeq
on a phase space $V^*Q$ such that
\mar{m72'}\beq
\zeta^*\{f,g\}_V=\{\zeta^*f,\zeta^*g\}_T.\label{m72'}
\eeq
The holonomic coordinates on $V^*Q$ are canonical for the Poisson
structure (\ref{m72}).

With respect to the Poisson bracket (\ref{m72}), the Hamiltonian
vector fields of functions on $V^*Q$ read
\mar{m73}\beq
\vt_f = \dr^if\dr_i- \dr_if\dr^i, \qquad f\in C^\infty(V^*Q).
\label{m73}
\eeq
They are vertical vector fields on $V^*Q\to \Bbb R$. Accordingly,
the characteristic distribution of the Poisson structure
(\ref{m72}) is the vertical tangent bundle $VV^*Q\subset TV^*Q$ of
a fibre bundle $V^*Q\to \Bbb R$. The corresponding symplectic
foliation on the phase space $V^*Q$ coincides with the fibration
$V^*Q\to \Bbb R$.

It is readily observed that the ring $\cC(V^*Q)$ of Casimir
functions on a Poisson manifold $V^*Q$ consists of the pull-back
onto $V^*Q$ of functions on $\Bbb R$. Therefore, the Poisson
algebra $C^\infty(V^*Q)$ is a Lie $C^\infty(\Bbb R)$-algebra.

\begin{rem} \label{ws529} \mar{ws529}
The Poisson structure (\ref{m72}) can be introduced in a different
way \cite{book10,book98}. Given  any section $h$ of the fibre
bundle (\ref{z11'}), let us consider the pull-back forms
\mar{z401}\ben
&& \bth=h^*(\Xi\w dt)=p_idq^i\w dt, \nonumber\\
&& \bom=h^*(d\Xi\w dt)=dp_i\w dq^i\w dt \label{z401}
\een
on $V^*Q$. They are independent of the choice of $h$. With $\bom$
(\ref{z401}), the Hamiltonian vector field $\vt_f$ (\ref{m73}) for
a function $f$ on $V^*Q$ is given by the relation
\be
\vt_f\rfloor\bom = -df\w dt,
\ee
while the Poisson bracket (\ref{m72}) is written as
\be
\{f,g\}_Vdt=\vt_g\rfloor\vt_f\rfloor\bom.
\ee
Moreover, one can show that a projectable vector field $\vt$ on
$V^*Q$ such that $\vt\rfloor dt=$const. is a canonical vector
field for the Poisson structure (\ref{m72}) iff
\mar{0100}\beq
\bL_\vt\bom=d(\vt\rfloor\bom)=0. \label{0100}
\eeq
\end{rem}

In contrast with autonomous Hamiltonian mechanics, the Poisson
structure (\ref{m72}) fails to provide any dynamic equation on a
fibre bundle $V^*Q\to\Bbb R$ because Hamiltonian vector fields
(\ref{m73}) of functions on $V^*Q$ are vertical vector fields, but
not connections on $V^*Q\to\Bbb R$ (see Definition \ref{fodeq2}).
Hamiltonian dynamics on $V^*Q$ is described as a particular
Hamiltonian dynamics on fibre bundles \cite{book10,book98,sard98}.

A Hamiltonian on a phase space $V^*Q\to\Bbb R$ of mechanics is
defined as a global section
\mar{ws513}\beq
h:V^*Q\to T^*Q, \qquad p_0\circ h=\cH(t,q^j,p_j), \label{ws513}
\eeq
of the affine bundle $\zeta$ (\ref{z11'}). Given the Liouville
form $\Xi$ (\ref{N43}) on $T^*Q$, this section yields the
pull-back Hamiltonian form
\mar{b4210}\beq
H=(-h)^*\Xi= p_k dq^k -\cH dt  \label{b4210}
\eeq
on $V^*Q$. This is the well-known invariant of Poincar\'e--Cartan
\cite{arn}.

It should be emphasized that, in contrast with a Hamiltonian in
autonomous mechanics, the Hamiltonian $\cH$ (\ref{ws513}) is not a
function on $V^*Q$, but it obeys the transformation law
\mar{0144}\beq
\cH'(t,q'^i,p'_i)=\cH(t,q^i,p_i)+ p'_i\dr_t q'^i. \label{0144}
\eeq

\begin{rem} \label{ws512} \mar{ws512}
Any connection $\G$ (\ref{a1.10}) on a configuration bundle
$Q\to\Bbb R$ defines the global section $h_\G=p_i\G^i$
(\ref{ws513}) of the affine bundle $\zeta$ (\ref{z11'}) and the
corresponding Hamiltonian form
\mar{ws515}\beq
H_\G= p_k dq^k -\cH_\G dt= p_k dq^k -p_i\G^i dt. \label{ws515}
\eeq
Furthermore, given a connection $\G$, any Hamiltonian form
(\ref{b4210}) admits the splitting
\mar{m46'}\beq
H= H_\G -\cE_\G dt, \label{m46'}
\eeq
where
\mar{xx60}\beq
\cE_\G=\cH-\cH_\G=\cH- p_i\G^i \label{xx60}
\eeq
is a function on $V^*Q$. One can think of $\cE_\G$ (\ref{xx60}) as
being an energy function relative to a reference frame $\G$
\cite{book10,jmp07}. With respect to the coordinates adapted to a
reference frame $\G$, we have $\cE_\G=\cH$. Given different
reference frames $\G$ and $\G'$, the decomposition (\ref{m46'})
leads at once to the relation
\mar{0200}\beq
\cE_{\G'}=\cE_\G + \cH_\G -\cH_{\G'}=\cE_\G + (\G^i -\G'^i)p_i
\label{0200}
\eeq
between the energy functions with respect to different reference
frames.
\end{rem}

Given a Hamiltonian form $H$ (\ref{b4210}), there exists a unique
horizontal vector field (\ref{a1.10}):
\be
\g_H=\dr_t -\g^i\dr_i -\g_i\dr^i,
\ee
on $V^*Q$ (i.e., a connection on $V^*Q\to \Bbb R$) such that
\mar{w255}\beq
\g_H\rfloor dH=0. \label{w255}
\eeq
This vector field, called the Hamilton vector field, reads
\mar{z3}\beq
\g_H=\dr_t + \dr^k\cH\dr_k- \dr_k\cH\dr^k. \label{z3}
\eeq
In a different way (Remark \ref{ws529}), the Hamilton vector field
$\g_H$ is defined by the relation
\be
\g_H\rfloor\bom=dH.
\ee
Consequently, it is canonical for the Poisson structure $\{,\}_V$
(\ref{m72}). This vector field  yields the first order dynamic
Hamilton equation
\mar{z20a,b}\ben
&& q^k_t=\dr^k\cH, \label{z20a}\\
&&  p_{tk}=-\dr_k\cH \label{z20b}
\een
on $V^*Q\to\Bbb R$ (Definition \ref{fodeq2}), where
$(t,q^k,p_k,q^k_t,\dot p_{tk})$ are the adapted coordinates on the
first order jet manifold $J^1V^*Q$ of $V^*Q\to\Bbb R$.

Due to the canonical imbedding $J^1V^*Q\to TV^*Q$ (\ref{z260}),
the Hamilton equation (\ref{z20a}) -- (\ref{z20b}) is equivalent
to the autonomous first order dynamic equation
\mar{z20}\beq
\dot t=1, \qquad \dot q^i=\dr^i\cH, \qquad \dot p_i=-\dr_i\cH
\label{z20}
\eeq
on a manifold $V^*Q$ (Definition \ref{gena70}).

A solution of the Hamilton equation (\ref{z20a}) -- (\ref{z20b})
is an integral section $r$ for the connection $\g_H$.

\begin{rem}
Similarly to the Cartan equation (\ref{C28}), the Hamilton
equation (\ref{z20a}) -- (\ref{z20b}) is equivalent to the
condition
\mar{N7}\beq
r^*(u\rfloor dH)= 0 \label{N7}
\eeq
for any vertical vector field $u$ on $V^*Q\to \Bbb R$.
\end{rem}

We agree to call $(V^*Q,H)$ the Hamiltonian system of $m=\di Q-1$
degrees of freedom.

In order to describe evolution of a Hamiltonian system at any
instant, the Hamilton vector field $\g_H$ (\ref{z3}) is assumed to
be complete, i.e., it is an Ehressmann connection (Remark
\ref{047}). In this case, the Hamilton equation (\ref{z20a}) --
(\ref{z20b}) admits a unique global solution through each point of
the phase space $V^*Q$. By virtue of Theorem \ref{compl}, there
exists a trivialization of a fibre bundle $V^*Q\to \Bbb R$ (not
necessarily compatible with its fibration $V^*Q\to Q$) such that
\mar{0102}\beq
\g_H=\dr_t, \qquad H=\ol p_id\ol q^i \label{0102}
\eeq
with respect to the associated coordinates $(t,\ol q^i, \ol p_i)$.
A direct computation shows that the Hamilton vector field $\g_H$
(\ref{z3}) satisfies the relation (\ref{0100}) and, consequently,
it is an infinitesimal generator of a one-parameter group of
automorphisms of the Poisson manifold $(V^*Q,\{,\}_V)$. Then one
can show that $(t,\ol q^i,\ol p_i)$ are canonical coordinates for
the Poisson manifold $(V^*Q,\{,\}_V)$ \cite{book98}, i.e.,
\be
w=\frac{\dr}{\dr \ol p_i}\w \frac{\dr}{\dr \ol q^i}.
\ee
Since $\cH=0$, the Hamilton equation (\ref{z20a}) -- (\ref{z20b})
in these coordinates takes the form
\be
\ol q^i_t=0, \qquad \ol p_{ti}=0,
\ee
i.e., $(t,\ol q^i,\ol p_i)$ are the initial data coordinates.

\section{Homogeneous Hamiltonian formalism}

As was mentioned above,  one can associate to any Hamiltonian
system on a phase space $V^*Q$ an equivalent autonomous symplectic
Hamiltonian system on the cotangent bundle $T^*Q$ (Theorem
\ref{09121}).

Given a Hamiltonian system $(V^*Q,H)$, its Hamiltonian $\cH$
(\ref{ws513}) defines the function
\mar{mm16}\beq
\cH^*=\dr_t\rfloor(\Xi-\zeta^* (-h)^*\Xi))=p_0+h=p_0+\cH
\label{mm16}
\eeq
on $T^*Q$. Let us regard $\cH^*$ (\ref{mm16}) as a Hamiltonian of
an autonomous Hamiltonian system on the symplectic manifold
$(T^*Q,\Om_T)$. The corresponding autonomous Hamilton equation on
$T^*Q$ takes the form
\mar{z20'}\beq
\dot t=1, \qquad \dot p_0=-\dr_t\cH, \qquad \dot q^i=\dr^i\cH,
\qquad \dot p_i=-\dr_i\cH. \label{z20'}
\eeq

\begin{rem} \label{0170} \mar{0170}
Let us note that the splitting $\cH^*=p_0+\cH$ (\ref{mm16}) is ill
defined. At the same time, any reference frame $\G$ yields the
decomposition
\mar{j3}\beq
\cH^*=(p_0+\cH_\G) + (\cH-\cH_\G) = \cH^*_\G +\cE_\G, \label{j3}
\eeq
where $\cH_\G$ is the Hamiltonian (\ref{ws515}) and $\cE_\G$
(\ref{xx60}) is the energy function relative to a reference frame
$\G$.
\end{rem}

The Hamiltonian vector field $\vt_{\cH^*}$ of $\cH^*$ (\ref{mm16})
on $T^*Q$ is
\mar{z5}\beq
\vt_{\cH^*}=\dr_t -\dr_t\cH\dr^0+ \dr^i\cH\dr_i- \dr_i\cH\dr^i.
\label{z5}
\eeq
Written relative to the coordinates (\ref{09150}), this vector
field reads
\mar{z5'}\beq
\vt_{\cH^*}=\dr_t + \dr^i\cH\dr_i- \dr_i\cH\dr^i. \label{z5'}
\eeq
It is identically projected onto the Hamilton vector field $\g_H$
(\ref{z3}) on $V^*Q$ such that
\mar{ws525}\beq
\zeta^*(\bL_{\g_H}f)=\{\cH^*,\zeta^*f\}_T, \qquad f\in
C^\infty(V^*Q). \label{ws525}
\eeq
Therefore, the Hamilton equation (\ref{z20a}) -- (\ref{z20b}) is
equivalent to the autonomous Hamilton equation (\ref{z20'}).

Obviously, the Hamiltonian vector field $\vt_{\cH^*}$ (\ref{z5'})
is complete if the Hamilton vector field $\g_H$ (\ref{z3}) is
complete.

Thus, the following has been proved \cite{dew,book10,mang00}.

\begin{theo} \label{09121} \mar{09121} A Hamiltonian system $(V^*Q,H)$
of $m$ degrees of freedom is equivalent to an autonomous
Hamiltonian system $(T^*Q,\cH^*)$ of $m+1$ degrees of freedom on a
symplectic manifold $(T^*Q,\Om)$ whose Hamiltonian is the function
$\cH^*$ (\ref{mm16}).
\end{theo}

We agree to call $(T^*Q,\cH^*)$ the homogeneous Hamiltonian system
and $\cH^*$ (\ref{mm16}) the homogeneous Hamiltonian.

\section{Lagrangian form of Hamiltonian formalism}

It is readily observed that the Hamiltonian form $H$ (\ref{b4210})
is the Poincar\'e--Cartan form of the Lagrangian
\mar{Q33}\beq
L_H=h_0(H) = (p_iq^i_t - \cH)dt \label{Q33}
\eeq
on the jet manifold $J^1V^*Q$ of $V^*Q\to\Bbb R$
\cite{book10,jmp07}.

\begin{rem} \label{0110} \mar{0110}
In fact, the Lagrangian (\ref{Q33}) is the pull-back onto
$J^1V^*Q$ of the form $L_H$ on the product $V^*Q\times_Q J^1Q$.
\end{rem}

The Lagrange operator (\ref{21f11}) associated to the Lagrangian
$L_H$ reads
\mar{3.9}\beq
\cE_H=\dl L_H=[(q^i_t-\dr^i\cH) dp_i -(p_{ti}+\dr_i\cH) dq^i]\w
dt. \label{3.9}
\eeq
The corresponding Lagrange equation (\ref{21f50}) is of first
order, and it coincides with the Hamilton equation (\ref{z20a}) --
(\ref{z20b}) on $J^1V^*Q$.

Due to this fact, the Lagrangian $L_H$ (\ref{Q33}) plays a
prominent role in Hamiltonian mechanics.

In particular, let $u$ (\ref{z372}) be a vector field  on a
configuration space $Q$. Its functorial lift (\ref{l27'}) onto the
cotangent bundle $T^*Q$ is
\mar{gm513}\beq
\wt u=u^t\dr_t + u^i\dr_i - p_j\dr_i u^j \dr^i \label{gm513}
\eeq
This vector field is identically projected onto a vector field,
also given by the expression (\ref{gm513}), on the phase space
$V^*Q$ as a base of the trivial fibre bundle (\ref{z11'}). Then we
have the equality
\mar{mm24}\beq
\bL_{\wt u}H= \bL_{J^1\wt u}L_H= (-u^t\dr_t\cH+p_i\dr_tu^i
-u^i\dr_i\cH +p_i\dr_j u^i\dr^j\cH)dt. \label{mm24}
\eeq
This equality enables us to study conservation laws in Hamiltonian
mechanics similarly to those in Lagrangian mechanics (Section
3.5).

\section{Associated Lagrangian and Hamiltonian systems}

As was mentioned above, Lagrangian and Hamiltonian formulations of
mechanics fail to be equivalent. The comprehensive relations
between Lagrangian and Hamiltonian systems can be established in
the case of almost regular Lagrangians
\cite{book10,book98,mang00,sard98}. This is a particular case of
the relations between Lagrangian and Hamiltonian theories on fibre
bundles \cite{jpa99,book09}.

In order to compare Lagrangian and Hamiltonian formalisms, we are
based on the facts that:

(i) every first order Lagrangian $L$ (\ref{21f10}) on a velocity
space $J^1Q$ induces the Legendre map (\ref{a303}) of this
velocity space to a phase space $V^*Q$;

(ii) every Hamiltonian form $H$ (\ref{b4210}) on a phase space
$V^*Q$ yields the Hamiltonian map
\mar{ws531b}\beq
\wh H: V^*Q\ar_Q J^1Q, \qquad q^i_t\circ\wh H=\dr^i\cH
\label{ws531b}
\eeq
of this phase space to a velocity space $J^1Q$.

\begin{rem}
A Hamiltonian form $H$ is called regular if the Hamiltonian map
$\wh H$ (\ref{ws531b}) is regular, i.e., a local diffeomorphism.
\end{rem}

\begin{rem} It is readily observed that a section $r$ of a fibre bundle $V^*Q\to \Bbb R$
is a solution of the Hamilton equation (\ref{z20a}) --
(\ref{z20b}) for the Hamiltonian form $H$ iff it obeys the
equality
\mar{N10}\beq
J^1(\pi_\Pi\circ r)=\wh H\circ r, \label{N10}
\eeq
where $\pi_\Pi:V^*Q\to Q$.
\end{rem}

Given a Lagrangian $L$, the Hamiltonian form $H$ (\ref{b4210}) is
said to be associated with $L$ if $H$ satisfies the relations
\mar{d2.30}\ben
&&\wh L\circ\wh H\circ \wh L=\wh L,\label{d2.30a} \\
&&\wh H^*L_H=\wh H^*L, \label{d2.30b}
\een
where $L_H$ is the Lagrangian (\ref{Q33}).

A glance at the equality (\ref{d2.30a}) shows that $\wh L\circ\wh
H$ is the projector of $V^*Q$ onto the Lagrangian constraint space
$N_L$ which is given by the coordinate conditions
\mar{060'}\beq
p_i=\pi_i(t,q^j,\dr^j\cH(t,q^j,p_j)). \label{060'}
\eeq

The relation (\ref{d2.30b}) takes the coordinate form
\mar{b481}\beq
\cH=p_i\dr^i\cH-\cL(t,q^j,\dr^j\cH). \label{b481}
\eeq
Acting on this equality by the exterior differential, we obtain
the relations
\mar{2.31,'}\ben
&& \dr_t\cH(p) =-(\dr_t\cL)\circ \wh H(p), \qquad p\in N_L,
\nonumber\\
&&  \dr_i\cH(p) =-(\dr_i\cL)\circ \wh H(p), \qquad
p\in N_L, \label{2.31}\\
&& (p_i-(\dr_i\cL)
(t,q^j,\dr^j\cH))\dr^i\dr^a\cH=0.\label{2.31'}
\een
The relation (\ref{2.31'}) shows that an $L$-associated
Hamiltonian form $H$ is not regular outside the Lagrangian
constraint space $N_L$.

For instance, let $L$ be a hyperregular Lagrangian, i.e., the
Legendre map $\wh L$ (\ref{a303}) is a diffeomorphism. It follows
from the relation (\ref{d2.30a}) that, in this case, $\wh H=\wh
L^{-1}$. Then the relation (\ref{b481}) takes the form
\mar{cc311}\beq
\cH=p_i\wh L^{-1i} - \cL(t, q^j,\wh L^{-1j}). \label{cc311}
\eeq
It defines a unique Hamiltonian form associated with a
hyperregular Lagrangian. Let $s$ be a solution of the Lagrange
equation (\ref{b327}) for a Lagrangian $L$. A direct computation
shows that $\wh L\circ J^1s$ is a solution of the Hamilton
equation (\ref{z20a}) -- (\ref{z20b}) for the Hamiltonian form $H$
(\ref{cc311}). Conversely, if $r$ is a solution of the Hamilton
equation (\ref{z20a}) -- (\ref{z20b}) for the Hamiltonian form $H$
(\ref{cc311}), then $s=\pi_\Pi\circ r$ is a solution of the
Lagrange equation (\ref{b327}) for $L$ (see the equality
(\ref{N10})). It follows that, in the case of hyperregular
Lagrangians, Hamiltonian formalism is equivalent to Lagrangian
one.

If a Lagrangian is not regular, an associated Hamiltonian form
need not exist.

A Hamiltonian form is called {\it weakly associated} with a
Lagrangian $L$ if the condition (\ref{d2.30b}) (namely, the
condition (\ref{2.31'}) holds on the Lagrangian constraint space
$N_L$.

For instance, any Hamiltonian form is weakly associated with the
Lagrangian $L=0$, while the associated Hamiltonian forms are only
$H_\G$ (\ref{ws515}).

A hyperregular Lagrangian $L$ has a unique weakly associated
Hamiltonian form (\ref{cc311}) which also is $L$-associated. In
the case of a regular Lagrangian $L$, the Lagrangian constraint
space $N_L$ is an open subbundle of the vector Legendre bundle
$V^*Q\to Q$. If $N_L\neq V^*Q$, a weakly associated Hamiltonian
form fails to be defined everywhere on $V^*Q$ in general. At the
same time, $N_L$ itself can be provided with the pull-back
symplectic structure with respect to the imbedding $N_L\to V^*Q$,
so that one may consider Hamiltonian forms on $N_L$.

Note that, in contrast with associated Hamiltonian forms, a weakly
associated Hamiltonian form may be regular.

In order to say something more, let us restrict our consideration
to almost regular Lagrangians $L$ (Definition \ref{d11})
\cite{book10,book98,mang00}.

\begin{lem} \label{d3.22} \mar{d3.22}
The Poincar\'e--Cartan form $H_L$ (\ref{303}) of an almost regular
Lagrangian $L$ is constant on the inverse image $\wh L^{-1}(z)$ of
any point $z\in N_L$.
\end{lem}

A corollary of Lemma \ref{d3.22} is the following.

\begin{theo} \label{d3.22'}  All Hamiltonian forms
weakly associated with an almost regular Lagrangian $L$ coincide
with each other on the Lagrangian constraint space $N_L$, and the
Poincar\'e--Cartan form $H_L$ (\ref{303}) of $L$ is the pull-back
\mar{d2.32}\beq
H_L=\wh L^*H, \qquad \pi_iq^i_t-\cL=\cH(t,q^j,\pi_j),
\label{d2.32}
\eeq
of such a Hamiltonian form $H$.
\end{theo}

It follows that, given Hamiltonian forms $H$ and $H'$ weakly
associated with an almost regular Lagrangian $L$, their difference
is a density
\be
H'-H=(\cH-\cH')dt
\ee
vanishing on the Lagrangian constraint space $N_L$. However, $\wh
H|_{N_L}\neq \wh H'|_{N_L}$ in general. Therefore, the Hamilton
equations for $H$ and $H'$ do not necessarily coincide on the
Lagrangian constraint space $N_L$.

Theorem \ref{d3.22'} enables us to relate the Lagrange equation
for an almost regular Lagrangian $L$ with the Hamilton equation
for Hamiltonian forms weakly associated to $L$.

\begin{theo}\label{d3.23} \mar{d3.23}
Let a section $r$ of $V^*Q\to \Bbb R$ be a  solution of the
Hamilton equation (\ref{z20a}) -- (\ref{z20b}) for a Hamiltonian
form $H$ weakly associated with an almost regular Lagrangian $L$.
If $r$ lives in the Lagrangian constraint space $N_L$, the section
$s=\pi\circ r$ of $\pi:Q\to \Bbb R$ satisfies the Lagrange
equation (\ref{b327}), while $\ol s=\wh H\circ r$ obeys the Cartan
equation (\ref{b336a}) -- (\ref{b336b}).
\end{theo}

The proof is based on the relation
\be
\wt L=(J^1\wh L)^*L_H,
\ee
where $\wt L$ is the Lagrangian (\ref{cmp80}), while $L_H$ is the
Lagrangian (\ref{Q33}). This relation is derived from the equality
(\ref{d2.32}). The converse assertion is more intricate.

\begin{theo}\label{d3.24} \mar{d3.24} Given an almost regular
Lagrangian $L$, let a section $\ol s$ of the jet bundle $J^1Q\to
\Bbb R$ be a solution of the Cartan equation  (\ref{b336a}) --
(\ref{b336b}). Let $H$ be a Hamiltonian form weakly associated
with $L$,  and let $H$ satisfy the relation
\mar{2.36'}\beq
\wh H\circ \wh L\circ \ol s=J^1s, \label{2.36'}
\eeq
where $s$ is the projection of $\ol s$ onto $Q$. Then the section
$r=\wh L\circ \ol s$ of a fibre bundle $V^*Q\to \Bbb R$ is a
solution of the Hamilton equation (\ref{z20a}) -- (\ref{z20b}) for
$H$.
\end{theo}

We say that a set of Hamiltonian forms $H$ weakly associated with
an almost regular Lagrangian $L$ is complete if, for each solution
$s$ of the Lagrange equation, there exists a solution $r$ of the
Hamilton equation for a Hamiltonian form $H$ from this set such
that $s=\pi_\Pi\circ r$. By virtue of Theorem \ref{d3.24}, a set
of weakly associated Hamiltonian forms is complete if, for every
solution $s$ of the Lagrange equation for $L$, there exists a
Hamiltonian form $H$ from this set which fulfills the relation
(\ref{2.36'}) where $\ol s=J^1s$, i.e.,
\mar{072}\beq
\wh H\circ \wh L\circ J^1s=J^1s. \label{072}
\eeq

In the case of almost regular Lagrangians, one can formulate the
following necessary and sufficient conditions of the existence of
weakly associated Hamiltonian forms.

\begin{theo} \label{mm71} \mar{mm71}
A Hamiltonian form $H$ weakly associated with an almost regular
Lagrangian $L$ exists iff the fibred manifold (\ref{cmp12}):
\mar{cmp12'}\beq
\wh L: J^1Q\to N_L, \label{cmp12'}
\eeq
admits a global section.
\end{theo}

In particular, any point of $V^*Q$ possesses an open neighborhood
$U$ such that there exists a complete set of local Hamiltonian
forms on $U$ which are weakly associated with an almost regular
Lagrangian $L$. Moreover, one can construct a complete set of
local $L$-associated Hamiltonian forms on $U$ \cite{sard95}.

\section{Hamiltonian conservation laws}

As was mentioned above, integrals of motion in Lagrangian
mechanics usually come from variational symmetries of a Lagrangian
(Theorem \ref{035'}), though not all integrals of motion are of
this type (Section 2.4). In Hamiltonian mechanics, all integrals
of motion are conserved generalized symmetry currents (Theorem
\ref{0150} below).

An integral of motion of a Hamiltonian system $(V^*Q,H)$ is
defined as a smooth real function $F$ on $V^*Q$ which is an
integral of motion of the Hamilton equation (\ref{z20a}) --
(\ref{z20b}) (Section 1.10). Its Lie derivative
\mar{ws516}\beq
\bL_{\g_H} F=\dr_tF +\{\cH,F\}_V \label{ws516}
\eeq
along the Hamilton vector field $\g_H$ (\ref{z3}) vanishes in
accordance with the equation (\ref{0116}). Given the Hamiltonian
vector field $\vt_F$ of $F$ with respect to the Poisson bracket
(\ref{m72}), it is easily justified that
\mar{092}\beq
[\g_H,\vt_F]=\vt_{\bL_{\g_H} F}. \label{092}
\eeq
Consequently, the Hamiltonian vector field of an integral of
motion is a symmetry of the Hamilton equation (\ref{z20a}) --
(\ref{z20b}).

One can think of the formula (\ref{ws516}) as being the evolution
equation of Hamiltonian mechanics.

Given a Hamiltonian system $(V^*Q,H)$, let $(T^*Q,\cH^*)$ be an
equivalent homogeneous Hamiltonian system. It follows from the
equality (\ref{ws525}) that
\mar{077a}\beq
\zeta^*(\bL_{\g_H}F)=\{\cH^*,\zeta^*F\}_T =\zeta^*(\dr_tF
+\{\cH,F\}_V) \label{077a}
\eeq
for any function $F\in C^\infty(V^*Z)$. This formula is equivalent
to the evolution equation (\ref{ws516}). It is called the
homogeneous evolution equation.

\begin{prop} \label{075} \mar{075} A function $F\in
C^\infty(V^*Q)$ is an integral of motion of a Hamiltonian system
$(V^*Q,H)$ iff its pull-back $\zeta^*F$ onto $T^*Q$ is an integral
of motion of a homogeneous Hamiltonian system $(T^*Q,\cH^*)$.
\end{prop}

\begin{proof} It follows from the equality (\ref{077a}) that
\mar{077}\beq
\{\cH^*,\zeta^*F\}_T=\zeta^*(\bL_{\g_H}F)=0. \label{077}
\eeq
\end{proof}

\begin{prop} \label{076} \mar{076} If $F$ and $F'$ are
integrals of motion of a Hamiltonian system, their Poisson bracket
$\{F,F'\}_V$ also is an integral of motion.
\end{prop}

Consequently, integrals of motion of a Hamiltonian system
$(V^*Q,H)$ constitute a real Lie subalgebra of the Poisson algebra
$C^\infty(V^*Q)$.

Let us turn to Hamiltonian conservation laws. We are based on the
fact that the Hamilton equation (\ref{z20a}) -- (\ref{z20b}) also
is the Lagrange equation of the Lagrangian $L_H$ (\ref{Q33}).
Therefore, one can study conservation laws in Hamiltonian
mechanics similarly to those in Lagrangian mechanics
\cite{book10,jmp07}.

Since the Hamilton equation (\ref{z20a}) -- (\ref{z20b}) is of
first order, we restrict our consideration to classical
symmetries, i.e., vector fields on $V^*Q$. In this case, all
conserved generalized symmetry currents are integrals of motion.

Let
\mar{0121}\beq
\up=u^t\dr_t + \up^i\dr_i + \up_i\dr^i, \qquad u^t=0,1,
\label{0121}
\eeq
be a vector field on a phase space $V^*Q$. Its prolongation onto
$V^*Q\times_Q J^1Q$ (Remark \ref{0110}) reads
\be
J^1\up= u^t\dr_t + \up^i\dr_i + \up_i\dr^i + d_t\up^i\dr_i^t.
\ee
Then the first variational formula (\ref{J4}) for the Lagrangian
$L_H$ (\ref{Q33}) takes the form
\mar{0122}\ben
&& -u^t\dr_t\cH - \up^i\dr_i\cH +\up_i(q^i_t -\dr^i\cH) +p_id_t\up^i =
\label{0122}\\
&& \qquad -(\up^i-q^i_tu^t)(p_{ti}+\dr_i\cH)+ (\up_i-p_{ti}u^t)(q^i_t-\dr^i\cH)
 \nonumber\\
&& \qquad + d_t(p_i\up^i-u^t\cH).\nonumber
\een

If $\up$ (\ref{0121}) is a variational symmetry, i.e.,
\be
\bL_{J^1\up} L_H=d_H\si,
\ee
we obtain the weak conservation law, called the Hamiltonian
conservation law,
\mar{0171}\beq
0\ap d_t \gT_\up \label{0171}
\eeq
of the generalized symmetry current (\ref{m225'}) which reads.
\mar{0141}\beq
\gT_\up=p_i\up^i-u^t\cH-\si. \label{0141}
\eeq
This current is an integral of motion of a Hamiltonian system.

The converse also is true. Let $F$ be an integral of motion, i.e.,
\mar{0145}\beq
\bL_{\g_H} F=\dr_tF +\{\cH,F\}_V=0. \label{0145}
\eeq
We aim to show that there is a variational symmetry $\up$ of $L_H$
such that $F=\gT_\up$ is a conserved generalized symmetry current
along $\up$.

In accordance with Proposition \ref{0133}, the vector field $\up$
(\ref{0121}) is a variational symmetry iff
\mar{0140}\beq
\up^i(p_{ti}+\dr_i\cH) -\up_i(q^i_t-\dr^i\cH) +
u^t\dr_t\cH=d_t(\gT_u +u^t\cH). \label{0140}
\eeq
A glance at this equality shows the following.

\begin{prop} \label{0146} \mar{0146}
The vector field $\up$ (\ref{0121}) is a variational symmetry only
if
\mar{0147}\beq
\dr^i\up_i=-\dr_i\up^i. \label{0147}
\eeq
\end{prop}

For instance, if the vector field $\up$ (\ref{0121}) is
projectable onto $Q$ (i.e., its components $\up^i$ are independent
of momenta $p_i$), we obtain that $u_i=-p_j\dr_iu^j$.
Consequently, $\up$ is the canonical lift $\wt u$ (\ref{gm513})
onto $V^*Q$ of the vector field $u$ (\ref{z372}) on $Q$. Moreover,
let $\wt u$ be a variational symmetry of a Lagrangian $L_H$. It
follows at once from the equality (\ref{0140}) that $\wt u$ is an
exact symmetry of $L_H$. The corresponding conserved symmetry
current reads
\mar{0150}\beq
\gT_{\wt up}=p_iu^i-u^t\cH. \label{0150}
\eeq

We agree to call the vector field $u$ (\ref{z372}) the Hamiltonian
symmetry if its canonical lift $\wt u$ (\ref{gm513}) onto $V^*Q$
is a variational (consequently, exact) symmetry of the Lagrangian
$L_H$ (\ref{Q33}).  If a Hamiltonian symmetry is vertical, the
corresponding conserved symmetry current $\gT_{\wt u}=p_iu^i$ is
called the Noether current.

\begin{prop} \label{0160} \mar{0160}
The Hamilton vector field $\g_H$ (\ref{z3}) is a unique
variational symmetry of $L_H$ whose conserved generalized symmetry
current equals zero.
\end{prop}

It follows that, given a non-vertical variational symmetry $\up$,
$u^t=1$, of a Lagrangian $L_H$, there exists a vertical
variational symmetry $\up-\g_H$ possessing the same generalized
conserved symmetry current $\gT_\up=\gT_{\up-\g_H}$ as $\up$.

\begin{theo} \label{0150'} \mar{0150'}
Any integral of motion $F$ of a Hamiltonian system $(V^*Q,H)$ is a
generalized conserved current $F=\gT_{\vt_F}$ of the Hamiltonian
vector field
\be
\vt_F\dr^iF\dr_F -\dr_i F\dr^F
\ee
of $F$.
\end{theo}

\begin{proof}
If $\up=\vt_F$ and $\gT_{\vt_F}=F$, the relation (\ref{0140}) is
satisfied owing to the equality (\ref{0145}).
\end{proof}

It follows from Theorem \ref{0150'} that the Lie algebra of
integrals of motion of a Hamiltonian system in Proposition
\ref{076} coincides with the Lie algebra of conserved generalized
symmetry currents with respect to the bracket
\be
\{F,F'\}_V=\{\gT_{\vt_F},\gT_{\vt_{F'}}\}_V=\gT_{[\vt_F,\vt_{F'}]}.
\ee

In accordance with Theorem \ref{0150'}, any integral of motion of
a Hamiltonian system can be treated as a conserved generalized
current along a vertical variational symmetry. However, this is
not convenient for the study of energy conservation laws.

Let $\cE_\G$ (\ref{xx60}) be the energy function of a Hamiltonian
system relative to a reference frame $\G$. Given bundle
coordinates adapted to $\G$, its evolution equation (\ref{ws516})
takes the form
\mar{0173}\beq
\bL_{\g_H}\cE_\G=\dr_t\cE_\G=\dr_t\cH. \label{0173}
\eeq
It follows that, an energy of a Hamiltonian system relative to a
reference frame $\G$ is an integral of motion iff a Hamiltonian,
written with respect to the coordinates adapted to $\G$, is
time-independent. By virtue of Theorem \ref{0150'}, if $\cE_\G$ is
an integral of motion, it is a conserved generalized symmetry
current of the variational symmetry
\be
\g_H + \vt_{\cE_\G}=-(\dr_t +\G^i\dr_i -p_j\dr_i\G^j\dr^i)=-\wt\G.
\ee
This is the canonical lift (\ref{gm513}) onto $V^*Q$ of the vector
field $-\G$ (\ref{a1.10}) on $Q$. Consequently, $-\wt\G$ is an
exact symmetry, and $-\G$ is a Hamiltonian symmetry.

\begin{ex} \label{0152} \mar{0152} Let us consider the Kepler system
on the configuration space $Q$ (\ref{0155}) in Example
\ref{048}. Its phase space is
\be
V^*Q=\Bbb R\times\Bbb R^6
\ee
coordinated by $(t,q^i,p_i)$. The Lagrangian (\ref{042}) and
(\ref{049}) of the Kepler system is hyperregular. The associated
Hamiltonian form reads
\mar{0157}\beq
H=p_idq^i-\left[\frac12 \left(\op\sum_i(p_i)^2\right)
-\frac1r\right]dt. \label{0157}
\eeq
The corresponding Lagrangian $L_H$ (\ref{Q33}) is
\mar{0158}\beq
L_H=\left[p_iq^i_t - \frac12 \left(\op\sum_i(p_i)^2\right)
+\frac1r\right]dt. \label{158}
\eeq
The Kepler system possesses the following integrals of motion:

$\bullet$ an energy function $\cE=\cH$;

$\bullet$ orbital momenta
\mar{0159}\beq
M^a_b =q^a p_b - q^bp_a \label{0159}
\eeq

$\bullet$ components of the Rung--Lenz vector
\mar{0160'}\beq
A^a=\op\sum_b(q^ap_b -q^bp_a)p_b -\frac{q^a}{r}. \label{0160'}
\eeq
These integrals of motions are the conserved currents of:

$\bullet$ the exact symmetry $\dr_t$,

$\bullet$ the exact vertical symmetries
\mar{0161}\beq
\up^a_b=q^a \dr_b - q^b\dr_a - p_b\dr^a + p_a\dr^b, \label{0161}
\eeq

$\bullet$ the variational vertical symmetries
\mar{0162}\beq
\up^a=\op\sum_b[p_b\up^a_b + (q^ap_b -q^bp_a)\dr_b]
+\dr_b\left(\frac{q^a}{r}\right)\dr^b, \label{0162}
\eeq
respectively. Note that the vector fields $\up^a_b$ (\ref{0161})
are the canonical lift (\ref{gm513}) onto $V^*Q$ of the vector
fields
\be
u^a_b=q^a \dr_b - q^b\dr_a
\ee
on $Q$. Thus, these vector fields are vertical Hamiltonian
symmetries, and integrals of motion $M^a_b$ (\ref{0159}) are the
Noether currents.
\end{ex}

Let us remind that, in contrast with the Rung--Lenz vector
(\ref{0162}) in Hamiltonian mechanics, the Rung--Lenz vector
(\ref{050}) in Lagrangian mechanics fails to come from variational
symmetries of a Lagrangian. There is the following relation
between Lagrangian and Hamiltonian symmetries if they are the same
vector fields on a configuration space $Q$.

\begin{theo}\label{hamlaw} \mar{hamlaw} Let a
Hamiltonian form $H$ be associated with an almost regular
Lagrangian $L$. Let $r$ be a solution of the Hamilton equation
(\ref{z20a}) -- (\ref{z20b}) for $H$ which lives in the Lagrangian
constraint space $N_L$. Let $s=\pi_\Pi\circ r$ be the
corresponding solution  of the Lagrange equation for $L$ so that
the relation (\ref{072}) holds. Then, for any vector field $u$
(\ref{z372}) on a fibre bundle $Q\to \Bbb R$, we have
\mar{Q10'}\beq
\gT_{\wt u} (r)=\gT_u( \pi_\Pi\circ r),\qquad \gT_{\wt u} (\wh
L\circ J^1s) =\gT_u(s), \label{Q10'}
\eeq
where $\gT_u$ is the symmetry current (\ref{m225}) on $J^1Y$  and
$\gT_{\wt u}$ is the symmetry current (\ref{0150}) on $V^*Q$.
\end{theo}

By virtue of Theorems \ref{d3.23} -- \ref{d3.24}, it follows that:

$\bullet$ if $\gT_u$ in Theorem \ref{hamlaw} is a conserved
symmetry current, then the symmetry current $\gT_{\wt u}$
(\ref{Q10'}) is conserved on solutions of the Hamilton equation
which live in the Lagrangian constraint space;

$\bullet$ if $\gT_{\wt u}$ in Theorem \ref{hamlaw} is a conserved
symmetry current, then the symmetry current $\gT_u$ (\ref{Q10'})
is conserved on solutions $s$ of the Lagrange equation which obey
the condition (\ref{072}).

\chapter{Appendixes}

For the sake of convenience of the reader, this Chapter summarizes
the relevant material on differential geometry of fibre bundles
and modules over commutative rings \cite{book09,gre,book00,ste}.

\section{Geometry of fibre bundles}

Throughout this Section, all morphisms are smooth (i.e., of class
$C^\infty$), and manifolds are smooth real and finite-dimensional.
A smooth manifold is customarily assumed to be Hausdorff and
second-countable. Consequently, it is locally compact and
paracompact. Unless otherwise stated, manifolds are assumed to be
connected (and, consequently, arcwise connected).

Given a smooth manifold $Z$, by $\pi_Z:TZ\to Z$ is denoted its
tangent bundle. Given manifold coordinates $(z^\al)$ on $Z$, the
tangent bundle $TZ$ is equipped with the holonomic coordinates
\be
(z^\la,\dot z^\la), \qquad \dot z'^\la= \frac{\dr z'^\la}{\dr
z^\mu}\dot z^\m,
\ee
with respect to the holonomic frames $\{\dr_\la\}$ in the tangent
spaces to $Z$. Any manifold morphism $f:Z\to Z'$ yields the
tangent morphism
\be
Tf:TZ\to TZ', \qquad \dot z'^\la\circ Tf = \frac{\dr f^\la}{\dr
z^\m}\dot z^\m,
\ee
of their tangent bundles.

\subsection{Fibred manifolds}

Let $M$ and $N$ be smooth manifolds and $f:M\to N$ a manifold
morphism. Its rank ${\rm rank}_pf$ at a point $p\in M$ is defined
as the rank of the tangent map
\be
T_pf:T_pM\to T_{f(p)}N, \qquad p\in M.
\ee
Since the function $p \to {\rm rank}_pf$ is lower semicontinuous,
a manifold morphism $f$ of maximal rank at a point $p$ also is of
maximal rank on some open neighborhood of $p$. A morphism $f$ is
said to be an immersion if $T_pf$, $p\in M$, is injective and a
submersion if $T_pf$, $p\in M$, is surjective. Note that a
submersion is an open map (i.e., an image of any open set is
open).

If $f:M\to N$ is an injective immersion, its range is called a
submanifold of $N$. A submanifold is said to be imbedded if it
also is a topological subspace. In this case, $f$ is called an
imbedding.  For the sake of simplicity, we usually identify
$(M,f)$ with $f(M)$. If $M\subset N$, its natural injection is
denoted by $i_M:M\to N$.  There are the following criteria for a
submanifold to be imbedded.

\begin{theo}\label{subman3} \mar{subman3}
Let $(M,f)$ be a submanifold of $N$.

(i) A map $f$ is an imbedding iff, for each point $p\in M$, there
exists a (cubic) coordinate chart $(V,\psi)$ of $N$ centered at
$f(p)$ so that $f(M)\cap V$ consists of all points of $V$ with
coordinates $(x^1,\ldots,x^m,0,\ldots,0)$.

(ii) Suppose that $f:M\to N$ is a proper map, i.e., the inverse
images of compact sets are compact. Then $(M,f)$ is a closed
imbedded submanifold of $N$. In particular, this occurs if $M$ is
a compact manifold.

(iii) If $\di M =\di N$, then $(M,f)$ is an open imbedded
submanifold of $N$.
\end{theo}

If a manifold morphism
\mar{11f1}\beq
\p :Y\to X, \qquad \di X=n>0, \label{11f1}
\eeq
is a surjective submersion, one says that: (i) its domain $Y$ is a
fibred manifold, (ii) $X$ is its base, (iii) $\p$ is a fibration,
and (iv) $Y_x=\p^{-1}(x)$ is a fibre over $x\in X$.

By virtue of the inverse function theorem \cite{war}, the
surjection (\ref{11f1}) is a fibred manifold iff a manifold $Y$
admits an atlas of fibred coordinate charts $(U_Y; x^\la, y^i)$
such that $(x^\la)$ are coordinates on $\p(U_Y)\subset X$ and
coordinate transition functions read
\be
x'^\la =f^\la(x^\m), \qquad y'^i=f^i(x^\m,y^j).
\ee

The surjection $\p$ (\ref{11f1}) is a fibred manifold iff, for
each point $y\in Y$, there exists a local section $s$ of $Y\to X$
passing through $y$. Recall that by a local section of the
surjection (\ref{11f1}) is meant an injection $s:U\to Y$ of an
open subset $U\subset X$ such that $\p\circ s=\id U$, i.e., a
section sends any point $x\in X$ into the fibre $Y_x$ over this
point. A local section also is defined over any subset $N\in X$ as
the restriction to $N$ of a local section over an open set
containing $N$. If $U=X$, one calls $s$ the global section. A
range $s(U)$ of a local section $s:U\to Y$ of a fibred manifold
$Y\to X$ is an imbedded submanifold of $Y$. A local section is a
closed map, sending closed subsets of $U$ onto closed subsets of
$Y$. If  $s$ is a global section, then $s(X)$ is a closed imbedded
submanifold of $Y$. Global sections of a fibred manifold need not
exist.

\begin{theo} \label{mos9} \mar{mos9}
Let $Y\to X$ be a fibred manifold whose fibres are diffeomorphic
to $\Bbb R^m$.  Any its section over a closed imbedded submanifold
(e.g., a point) of $X$ is extended to a global section \cite{ste}.
In particular, such a fibred manifold always has a global section.
\end{theo}

Given fibred coordinates $(U_Y;x^\la,y^i)$, a section $s$ of a
fibred manifold $Y\to X$ is represented by collections of local
functions $\{s^i=y^i\ \circ s\}$ on $\p(U_Y)$.

Morphisms of fibred manifolds, by definition, are fibrewise
morphisms, sending a fibre to a fibre. Namely, a fibred morphism
of a fibred manifold $\pi:Y\to X$ to a fibred manifold $\pi':
Y'\to X'$ is defined as a pair $(\Phi,f)$ of manifold morphisms
which form a commutative diagram
\be
\begin{array}{rcccl}
& Y &\ar^\Phi & Y'&\\
_\pi& \put(0,10){\vector(0,-1){20}} & & \put(0,10){\vector(0,-1){20}}&_{\pi'}\\
& X &\ar^f & X'&
\end{array}, \qquad \pi'\circ\Phi=f\circ\pi.
\ee
Fibred injections and surjections are called monomorphisms and
epimorphisms, respectively. A fibred diffeomorphism is called an
isomorphism or an automorphism if it is an isomorphism to itself.
For the sake of brevity, a fibred morphism over $f=\id X$ usually
is said to be a fibred morphism over $X$, and is denoted by
$Y\ar_XY'$. In particular, a fibred automorphism over $X$ is
called a vertical automorphism.

\subsection{Fibre bundles}

A fibred manifold $Y\to X$ is said to be trivial if $Y$ is
isomorphic to the product $X\times V$. Different trivializations
of $Y\to X$ differ from each other in surjections $Y\to V$.

A fibred manifold $Y\to X$ is called a fibre bundle if it is
locally trivial, i.e., if it admits a fibred coordinate atlas
$\{(\pi^{-1}(U_\xi); x^\la, y^i)\}$ over a cover
$\{\pi^{-1}(U_\xi)\}$ of $Y$ which is the inverse image of a cover
$\gU=\{U_\xi\}$ of $X$. In this case, there exists a manifold $V$,
called a typical fibre, such that $Y$ is locally diffeomorphic to
the  splittings
\mar{mos02}\beq
\psi_\xi:\pi^{-1}(U_\xi) \to U_\xi\times V, \label{mos02}
\eeq
glued together by means of transition functions
\mar{mos271}\beq
\vr_{\xi\zeta}=\psi_\xi\circ\psi_\zeta^{-1}: U_\xi\cap
U_\zeta\times V \to  U_\xi\cap U_\zeta\times V \label{mos271}
\eeq
on overlaps $U_\xi\cap U_\zeta$. Transition functions
$\vr_{\xi\zeta}$ fulfil the cocycle condition
\mar{+9}\beq
\vr_{\xi\zeta}\circ\vr_{\zeta\iota}=\vr_{\xi\iota} \label{+9}
\eeq
on all overlaps $U_\xi\cap U_\zeta\cap U_\iota$. Restricted to a
point $x\in X$, trivialization morphisms $\psi_\xi$ (\ref{mos02})
and transition functions $\vr_{\xi\zeta}$ (\ref{mos271}) define
diffeomorphisms of fibres
\mar{sp21,2}\ben
&&\psi_\xi(x): Y_x\to V, \qquad x\in U_\xi,\label{sp21}\\
&& \vr_{\xi\zeta}(x):V\to V, \qquad x\in U_\xi\cap U_\zeta. \label{sp22}
\een
Trivialization charts $(U_\xi, \psi_\xi)$ together with transition
functions $\vr_{\xi\zeta}$ (\ref{mos271}) constitute a bundle
atlas
\mar{sp5}\beq
\Psi = \{(U_\xi, \psi_\xi), \vr_{\xi\zeta}\} \label{sp5}
\eeq
of a fibre bundle $Y\to X$. Two bundle atlases are said to be
equivalent if their union also is a bundle atlas, i.e., there
exist transition functions between trivialization charts of
different atlases. All atlases of a fibre bundle are equivalent.

Given a bundle atlas $\Psi$ (\ref{sp5}), a fibre bundle $Y$ is
provided with the fibred coordinates
\be
x^\la(y)=(x^\la\circ \pi)(y), \qquad y^i(y)=(y^i\circ\psi_\xi)(y),
\qquad y\in \pi^{-1}(U_\xi),
\ee
called the bundle coordinates, where $y^i$ are coordinates on a
typical fibre $V$.

A fibre bundle $Y\to X$ is uniquely defined by a bundle atlas.
Given an atlas $\Psi$ (\ref{sp5}), there exists a unique manifold
structure on $Y$ for which $\p:Y\to X$ is a fibre bundle with a
typical fibre $V$ and a bundle atlas $\Psi$.

There are the following useful criteria for a fibred manifold to
be a fibre bundle.

\begin{theo} \label{1110} \mar{1110}
If a fibration $\p:Y\to X$ is a proper map, then $Y\to X$ is a
fibre bundle. In particular, a compact fibred manifold is a fibre
bundle.
\end{theo}

\begin{theo} \label{11t2} \mar{11t2} A fibred manifold whose
fibres are diffeomorphic either to a compact manifold or $\Bbb
R^r$ is a fibre bundle \cite{meig}.
\end{theo}

A comprehensive relation between fibred manifolds and fibre
bundles is given in Remark \ref{Ehresmann}. It involves the notion
of an Ehresmann connection.

Forthcoming Theorems \ref{11t3} -- \ref{sp2} describe the
particular covers which one can choose for a bundle atlas
\cite{gre}.

\begin{theo} \label{11t3} \mar{11t3} Any fibre
bundle over a contractible base is trivial.
\end{theo}

Note that a fibred manifold over a contractible base need not be
trivial. It follows from Theorem \ref{11t3} that any cover of a
base $X$ by domains (i.e., contractible open subsets) is a bundle
cover.

\begin{theo} \label{sp1} \mar{sp1}
Every fibre bundle $Y\to X$ admits a bundle atlas over a countable
cover $\gU$ of $X$ where each member $U_\xi$ of $\gU$ is a domain
whose closure $\ol U_\xi$ is compact.
\end{theo}

If a base $X$ is compact, there is a bundle atlas of $Y$ over a
finite cover of $X$ which obeys the condition of Theorem
\ref{sp1}.

\begin{theo} \label{sp2} \mar{sp2}
Every fibre bundle $Y\to X$ admits a bundle atlas over a finite
cover $\gU$ of $X$, but its members need not be contractible and
connected.
\end{theo}

A fibred morphism of fibre bundles is called a bundle morphism. A
bundle monomorphism $\Phi:Y\to Y'$ over $X$ onto a submanifold
$\Phi(Y)$ of $Y'$ is called a subbundle of a fibre bundle $Y'\to
X$. There is the following useful criterion for an image and an
inverse image of a bundle morphism to be subbundles.

\begin{theo}\label{pomm} \mar{pomm}
Let $\Phi: Y\to Y'$ be a bundle morphism over $X$. Given a global
section $s$ of the fibre bundle $Y'\to X$ such that $s(X)\subset
\F(Y)$, by the kernel of a bundle morphism $\F$ with respect to a
section $s$ is meant the inverse image
\be
\Ker_s\F = \F^{-1}(s(X))
\ee
of $s(X)$ by $\F$. If $\Phi: Y\to Y'$ is a bundle morphism of
constant rank over $X$, then $\Phi(Y)$ and $\Ker_s\F$ are
subbundles of $Y'$ and $Y$, respectively.
\end{theo}

The following are the standard constructions of new fibre bundles
from old ones.

$\bullet$ Given a fibre bundle $\pi:Y\to X$ and a manifold
morphism $f: X'\to X$, the pull-back of $Y$ by $f$ is called the
manifold
\mar{mos106}\beq
f^*Y =\{(x',y)\in X'\times Y \,: \,\, \pi(y) =f(x')\}
\label{mos106}
\eeq
together with the natural projection $(x',y)\to x'$. It is a fibre
bundle over $X'$ such that the fibre of $f^*Y$ over a point $x'\in
X'$ is that of $Y$ over the point $f(x')\in X$. There is the
canonical bundle morphism
\mar{mos81}\beq
f_Y:f^*Y\ni (x',y)|_{\pi(y) =f(x')} \to y\in Y. \label{mos81}
\eeq
Any section $s$ of a fibre bundle $Y\to X$ yields the pull-back
section
\be
f^*s(x')=(x',s(f(x'))
\ee
of $f^*Y\to X'$.

$\bullet$ If $X'\subset X$ is a submanifold of $X$ and $i_{X'}$ is
the corresponding natural injection, then the pull-back bundle
\be
i_{X'}^*Y=Y|_{X'}
\ee
is called the restriction of a fibre bundle $Y$ to the submanifold
$X'\subset X$. If $X'$ is an imbedded submanifold, any section of
the pull-back bundle
\be
Y|_{X'}\to X'
\ee
is the restriction to $X'$ of some section of $Y\to X$.

$\bullet$ Let $\pi:Y\to X$ and $\pi':Y'\to X$ be fibre bundles
over the same base $X$. Their bundle product $Y\times_X Y'$ over
$X$ is defined as the pull-back
\be
Y\op\times_X Y'=\pi^*Y'\quad {\rm or} \quad Y\op\times_X
Y'={\pi'}^*Y
\ee
together with its natural surjection onto $X$.  Fibres of the
bundle product $Y\times Y'$ are the Cartesian products $Y_x\times
Y'_x$ of fibres of fibre bundles $Y$ and $Y'$.

$\bullet$ Let us consider the composite fibre bundle
\mar{1.34}\beq
Y\to \Si\to X. \label{1.34}
\eeq
It is provided with bundle coordinates $(x^\la,\si^m,y^i)$, where
$(x^\la,\si^m)$ are bundle coordinates on a fibre bundle $\Si\to
X$, i.e., transition functions of coordinates $\si^m$ are
independent of coordinates $y^i$. Let $h$ be a global section of a
fibre bundle $\Si\to X$. Then the restriction $Y_h=h^*Y$ of a
fibre bundle $Y\to\Si$ to $h(X)\subset \Si$ is a subbundle of a
fibre bundle $Y\to X$.

\subsection{Vector and affine bundles}

A fibre bundle $\pi:Y\to X$ is called a vector bundle if both its
typical fibre and fibres are finite-dimensional real vector
spaces, and if it admits a bundle atlas whose trivialization
morphisms and transition functions are linear isomorphisms. Then
the corresponding bundle coordinates on $Y$ are linear bundle
coordinates $(y^i)$ possessing linear transition functions
$y'^i=A^i_j(x)y^j$. We have
\mar{trt}\beq
y=y^ie_i(\pi(y))=y^i \psi_\xi(\pi(y))^{-1}(e_i), \qquad \pi(y)\in
U_\xi, \label{trt}
\eeq
where $\{e_i\}$ is a fixed basis for a typical fibre $V$ of $Y$
and $\{e_i(x)\}$ are the fibre bases (or the frames) for the
fibres $Y_x$ of $Y$ associated to a bundle atlas $\Psi$.

By virtue of Theorem \ref{mos9}, any vector bundle has a global
section, e.g., the canonical global zero-valued section $\wh
0(x)=0$.

\begin{theo} \label{12t10} \mar{12t10} Let a vector bundle $Y\to
X$ admit $m=\di V$ nowhere vanishing global sections $s_i$ which
are linearly independent, i.e., $\op\w^m s_i\neq 0$. Then $Y$ is
trivial.
\end{theo}

Global sections of a vector bundle $Y\to X$ constitute a
projective $C^\infty(X)$-module $Y(X)$ of finite rank. It is
called the structure module of a vector bundle. The well-known
Serre--Swan theorem \cite{book05} states the categorial
equivalence between the vector bundles over a smooth manifold $X$
and projective $C^\infty(X)$-modules of finite rank.

There are the following particular constructions of new vector
bundles from the old ones.

$\bullet$ Let $Y\to X$ be a vector bundle with a typical fibre
$V$. By $Y^*\to X$ is denoted the dual vector bundle with the
typical fibre $V^*$, dual of $V$. The interior product of $Y$ and
$Y^*$ is defined as a fibred morphism
\be
\rfloor: Y\otimes Y^*\ar_X X\times \Bbb R.
\ee

$\bullet$ Let $Y\to X$ and $Y'\to X$ be vector bundles with
typical fibres $V$ and $V'$, respectively. Their Whitney sum
$Y\oplus_X Y'$ is a vector bundle over $X$ with the typical fibre
$V\oplus V'$.

$\bullet$ Let $Y\to X$ and $Y'\to X$ be vector bundles with
typical fibres $V$ and $V'$, respectively. Their tensor product
$Y\ot_X Y'$ is a vector bundle over $X$ with the typical fibre
$V\ot V'$. Similarly, the exterior product of vector bundles
$Y\w_X Y'$ is defined. The exterior product
\mar{ss12f11}\beq
\w Y=X\times \Bbb R \op\oplus_X Y \op\oplus_X \op\w^2
Y\op\oplus_X\cdots\oplus \op\w^k Y, \qquad k=\di Y-\di X,
\label{ss12f11}
\eeq
is called the exterior bundle.

$\bullet$ If $Y'$ is a subbundle of a vector bundle $Y\to X$, the
factor bundle $Y/Y'$ over $X$ is defined as a vector bundle whose
fibres are the quotients $Y_x/Y'_x$, $x\in X$.

By a morphism of vector bundles is meant a linear bundle morphism,
which is a linear fibrewise map whose restriction to each fibre is
a linear map.

Given a linear bundle morphism $\Phi: Y'\to Y$ of vector bundles
over $X$, its kernel Ker$\,\Phi$ is defined as the inverse image
$\Phi^{-1}(\wh 0(X))$ of the canonical zero-valued section $\wh
0(X)$ of $Y$. By virtue of Theorem \ref{pomm}, if $\Phi$ is of
constant rank, its kernel and its range are vector subbundles of
the vector bundles $Y'$ and $Y$, respectively. For instance,
monomorphisms and epimorphisms of vector bundles fulfil this
condition.

\begin{rem}\label{mos30} \mar{mos30}
Given vector bundles $Y$ and $Y'$ over the same base $X$, every
linear bundle morphism
\be
\Phi: Y_x\ni \{e_i(x)\}\to \{\Phi^k_i(x)e'_k(x)\}\in Y'_x
\ee
over $X$ defines a global section
\be
\Phi: x\to \Phi^k_i(x)e^i(x)\ot e'_k(x)
\ee
of the tensor product $Y\ot Y'^*$, and {\it vice versa}.
\end{rem}

A sequence $Y'\ar^i Y\ar^j Y''$ of vector bundles over the same
base $X$ is called exact at $Y$ if Ker$\,j=\im i$. A sequence of
vector bundles
\mar{sp10}\beq
0\to Y'\ar^i Y\ar^j Y'' \to 0 \label{sp10}
\eeq
over $X$ is said to be a short exact sequence if it is exact at
all terms $Y'$, $Y$, and $Y''$. This means that $i$ is a bundle
monomorphism, $j$ is a bundle epimorphism, and Ker$\,j=\im i$.
Then $Y''$ is isomorphic to a factor bundle $Y/Y'$. Given an exact
sequence of vector bundles (\ref{sp10}), there is the exact
sequence of their duals
\be
0\to Y''^*\ar^{j^*} Y^*\ar^{i^*} Y'^* \to 0.
\ee
One says that the exact sequence (\ref{sp10}) is split if there
exists a bundle monomorphism $s:Y''\to Y$ such that $j\circ s=\id
Y''$ or, equivalently,
\be
Y=i(Y')\oplus s(Y'')= Y'\oplus Y''.
\ee

\begin{theo} \label{sp11} \mar{sp11}
Every exact sequence of vector bundles (\ref{sp10}) is split
\cite{hir}.
\end{theo}

The tangent bundle $TZ$ and the cotangent bundle $T^*Z$ of a
manifold $Z$ exemplify vector bundles. Given an atlas $\Psi_Z
=\{(U_\iota,\phi_\iota)\}$ of a manifold $Z$, the tangent bundle
is provided with the holonomic bundle atlas
\mar{mos150}\beq
\Psi_T =\{(U_\iota, \psi_\iota = T\phi_\iota)\}. \label{mos150}
\eeq
The associated linear bundle coordinates are holonomic coordinates
$(\dot z^\la)$.

The cotangent bundle of a manifold $Z$ is the dual $T^*Z\to Z$ of
the tangent bundle $TZ\to Z$. It is equipped with the holonomic
coordinates
\be
(z^\la,\dot z_\la). \qquad \dot z'_\la=\frac{\dr z^\m}{\dr
z'^\la}\dot z_\m,
\ee
with respect to the coframes $\{dz^\la\}$ for $T^*Z$ which are the
duals of $\{\dr_\la\}$.

The tensor product of tangent and cotangent bundles
\mar{sp20}\beq
T=(\op\ot^mTZ)\ot(\op\ot^kT^*Z), \qquad m,k\in \Bbb N,
\label{sp20}
\eeq
is called a tensor bundle, provided with holonomic bundle
coordinates $\dot z^{\al_1\cdots\al_m}_{\bt_1\cdots\bt_k}$
possessing transition functions
\be
\dot z'^{\al_1\cdots\al_m}_{\bt_1\cdots\bt_k}=\frac{\dr
z'^{\al_1}}{\dr z^{\m_1}}\cdots\frac{\dr z'^{\al_m}}{\dr
z^{\m_m}}\frac{\dr z^{\nu_1}}{\dr z'^{\bt_1}}\cdots\frac{\dr
z^{\nu_k}}{\dr z'^{\bt_k}} \dot
z^{\m_1\cdots\m_m}_{\nu_1\cdots\nu_k}.
\ee

Let $\pi_Y:TY\to Y$ be the tangent bundle of a fibred manifold
$\pi: Y\to X$. Given fibred coordinates $(x^\la,y^i)$ on $Y$, it
is equipped with the holonomic coordinates $(x^\la,y^i,\dot x^\la,
\dot y^i)$. The tangent bundle $TY\to Y$ has the subbundle $VY =
\Ker (T\pi)$, which consists of the vectors tangent to fibres of
$Y$. It is called the vertical tangent bundle of $Y$, and it is
provided with the holonomic coordinates $(x^\la,y^i,\dot y^i)$
with respect to the vertical frames $\{\dr_i\}$. Every fibred
morphism $\Phi: Y\to Y'$ yields the linear bundle morphism over
$\Phi$ of the vertical tangent bundles
\mar{ws538}\beq
V\Phi: VY\to VY', \qquad \dot y'^i\circ V\Phi=\frac{\dr
\Phi^i}{\dr y^j}\dot y^j. \label{ws538}
\eeq
It is called the vertical tangent morphism.

In many important cases, the vertical tangent bundle $VY\to Y$ of
a fibre bundle $Y\to X$ is trivial, and it is isomorphic to the
bundle product
\mar{48'}\beq
VY= Y\op\times_X\ol Y, \label{48'}
\eeq
where $\ol Y\to X$ is some vector bundle. One calls (\ref{48'})
the vertical splitting. For instance, every vector bundle $Y\to X$
admits the canonical vertical splitting
\mar{12f10}\beq
VY= Y\op\oplus_X Y. \label{12f10}
\eeq

The vertical cotangent bundle $V^*Y\to Y$ of a fibred manifold
$Y\to X$ is defined as the dual of the vertical tangent bundle
$VY\to Y$. It is not a subbundle of the cotangent bundle $T^*Y$,
but there is the canonical surjection
\mar{z11z}\beq
\zeta: T^*Y\ni \dot x_\la dx^\la +\dot y_i dy^i \to \dot y_i \ol
dy^i\in V^*Y, \label{z11z}
\eeq
where the bases $\{\ol dy^i\}$, possessing transition functions
\be
\ol dy'^i=\frac{\dr y'^i}{\dr y^j}\ol dy^j,
\ee
are the duals of the vertical frames $\{\dr_i\}$ of the vertical
tangent bundle $VY$.

For any fibred manifold $Y$, there exist the exact sequences of
vector bundles
\mar{1.8a,b}\ben
&& 0\to VY\ar TY\ar^{\pi_T} Y\op\times_X TX\to 0,
\label{1.8a} \\
&& 0\to Y\op\times_X T^*X\to T^*Y\to V^*Y\to 0.
\label{1.8b}
\een
Their splitting, by definition, is a connection on $Y\to X$
(Section 4.3.1).

Let us consider the tangent bundle $TT^*X$ of $T^*X$ and the
cotangent bundle $T^*TX$ of $TX$. Relative to  coordinates
$(x^\la, p_\la=\dot x_\la)$ on $T^*X$ and $(x^\la,\dot x^\la)$ on
$TX$, these fibre bundles are provided with the  coordinates
$(x^\la, p_\la, \dot x^\la, \dot p_\la)$ and  $(x^\la, \dot
x^\la,\dot x_\la, \ddot x_\la)$, respectively. By inspection of
the coordinate transformation laws, one can show that there is an
isomorphism
\mar{z124}\beq
\al:TT^*X= T^*TX, \qquad p_\la\llra\ddot x_\la, \quad \dot p_\la
\llra\dot x_\la \label{z124}
\eeq
of these bundles over $TX$. Given a fibred manifold $Y\to X$,
there is the similar isomorphism
\mar{m6}\beq
\al_V:VV^*Y = V^*VY, \qquad p_i\llra\ddot y_i, \quad \dot
p_i\llra\dot y_i \label{m6}
\eeq
over $VY$, where $(x^\la, y^i, p_i,\dot y^i, \dot p_i)$ and
$(x^\la, y^i, \dot y^i,\dot y_i, \ddot y_i)$ are  coordinates on
$VV^*Y$ and $V^*VY$, respectively.

Let $\ol\pi:\ol Y\to X$ be a vector bundle with a typical fibre
$\ol V$. An affine bundle modelled over the vector bundle $\ol
Y\to X$ is a fibre bundle $\pi:Y\to X$ whose typical fibre $V$ is
an affine space modelled over $\ol V$, all the fibres $Y_x$ of $Y$
are affine spaces modelled over the corresponding fibres $\ol Y_x$
of the vector bundle $\ol Y$, and there is an affine bundle atlas
\be
\Psi=\{(U_\al,\psi_\chi),\vr_{\chi\zeta}\}
\ee
of $Y\to X$ whose local trivializations morphisms $\psi_\chi$
(\ref{sp21}) and transition functions $\vr_{\chi\zeta}$
(\ref{sp22}) are affine isomorphisms.

Dealing with affine bundles, we use only affine bundle coordinates
$(y^i)$ associated to an affine bundle atlas $\Psi$. There are the
bundle morphisms
\be
&&Y\op\times_X\ol Y\ar_X Y,\qquad (y^i, \ol y^i)\to  y^i +\ol y^i,\\
&&Y\op\times_X Y\ar_X \ol Y,\qquad (y^i, y'^i)\to  y^i - y'^i,
\ee
where $(\ol y^i)$ are linear coordinates on a vector bundle $\ol
Y$.

By virtue of Theorem \ref{mos9}, affine bundles  have  global
sections, but in contrast with vector bundles, there is no
canonical global section of an affine bundle. Let $\pi:Y\to X$ be
an affine bundle. Every global section $s$ of an affine bundle
$Y\to X$ modelled over a vector bundle $\ol Y\to X$ yields the
bundle morphisms
\mar{mos31,'}\ben
&& Y\ni y\to y-s(\pi(y))\in \ol Y, \label{mos31}\\
&& \ol Y\ni \ol y\to s(\pi(y))+\ol y\in Y. \label{mos31'}
\een
In particular, every vector bundle $Y$ has a natural structure of
an affine bundle due to the morphisms (\ref{mos31'}) where $s=\wh
0$ is the canonical zero-valued section of $Y$.

\begin{theo} \label{11t60} \mar{11t60}
Any affine bundle $Y\to X$ admits bundle coordinates $(x^\la, \wt
y^i)$ possessing linear transition functions $\wt y'^i=A^i_j(x)\wt
y^j$ \cite{book09}.
\end{theo}

By a morphism of affine bundles is meant a bundle morphism
$\Phi:Y\to Y'$ whose restriction to each fibre of $Y$ is an affine
map. It is called an affine bundle morphism. Every affine bundle
morphism $\Phi:Y\to Y'$ of an affine bundle $Y$ modelled over a
vector bundle $\ol Y$ to an affine bundle $Y'$ modelled over a
vector bundle $\ol Y'$ yields an unique linear bundle morphism
\mar{1355'}\beq
\ol \Phi: \ol Y\to \ol Y', \qquad \ol y'^i\circ \ol\Phi=
\frac{\dr\Phi^i}{\dr y^j}\ol y^j, \label{1355'}
\eeq
called the linear derivative of $\Phi$.

Every affine bundle $Y\to X$ modelled over a vector bundle $\ol
Y\to X$ admits the canonical vertical splitting
\mar{48}\beq
VY= Y\op\times_X\ol Y. \label{48}
\eeq

\subsection{Vector and multivector fields}

Vector fields on a manifold $Z$ are global sections of the tangent
bundle $TZ\to Z$.

The set $\cT_1(Z)$ of vector fields on $Z$ is both a
$C^\infty(Z)$-module and a real Lie algebra with respect to the
Lie bracket
\be
&& u=u^\la\dr_\la, \qquad v=v^\la\dr_\la,\\
&& [v,u] = (v^\la\dr_\la u^\m - u^\la\dr_\la v^\m)\dr_\m.
\ee

Given a vector field $u$ on $X$, a curve
\be
c:\Bbb R\supset (,)\to Z
\ee
in $Z$ is said to be an integral curve of $u$ if $Tc=u(c)$. Every
vector field $u$ on a manifold $Z$ can be seen as an infinitesimal
generator of a local one-parameter group of local diffeomorphisms
(a flow),  and {\it vice versa} \cite{kob}. One-dimensional orbits
of this group are integral curves of $u$.

A vector field is called complete if its flow is a one-parameter
group of diffeomorphisms of $Z$.

\begin{theo} \label{10b3} \mar{10b3}
Any vector field on a compact manifold is complete.
\end{theo}

A vector field $u$ on a fibred manifold $Y\to X$ is called
projectable  if it is projected onto a vector field on $X$, i.e.,
there exists a vector field $\tau$ on $X$ such that
\be
\tau\circ\pi= T\pi\circ u.
\ee
A projectable vector field takes the coordinate form
\mar{11f30}\beq
u=u^\la(x^\m) \dr_\la + u^i(x^\m,y^j) \dr_i, \qquad
\tau=u^\la\dr_\la. \label{11f30}
\eeq
A projectable vector field is called vertical if its projection
onto $X$ vanishes, i.e., if it lives in the vertical tangent
bundle $VY$.

A vector field $\tau=\tau^\la\dr_\la$ on a base $X$ of a fibred
manifold $Y\to X$ gives rise to a vector field on $Y$ by means of
a connection on this fibre bundle (see the formula (\ref{b1.85})).
Nevertheless, every tensor bundle (\ref{sp20}) admits the
functorial lift of vector fields
\mar{l28}\beq
\wt\tau = \tau^\m\dr_\m + [\dr_\nu\tau^{\al_1}\dot
x^{\nu\al_2\cdots\al_m}_{\bt_1\cdots\bt_k} + \ldots
-\dr_{\bt_1}\tau^\nu \dot
x^{\al_1\cdots\al_m}_{\nu\bt_2\cdots\bt_k} -\ldots]\dot \dr
_{\al_1\cdots\al_m}^{\bt_1\cdots\bt_k}, \label{l28}
\eeq
where we employ the compact notation
\mar{vvv}\beq
\dot\dr_\la = \frac{\dr}{\dr\dot x^\la}. \label{vvv}
\eeq
This lift is an $\Bbb R$-linear monomorphism of the Lie algebra
$\cT_1(X)$ of vector fields on $X$ to the Lie algebra $\cT_1(Y)$
of vector fields on $Y$. In particular, we have the functorial
lift
\mar{l27}\beq
\wt\tau = \tau^\m\dr_\m +\dr_\nu\tau^\al\dot
x^\nu\frac{\dr}{\dr\dot x^\al} \label{l27}
\eeq
of vector fields on $X$ onto the tangent bundle $TX$ and their
functorial lift
\mar{l27'}\beq
\wt\tau = \tau^\m\dr_\m -\dr_\bt\tau^\nu\dot
x_\nu\frac{\dr}{\dr\dot x_\bt} \label{l27'}
\eeq
onto the cotangent bundle $T^*X$.

Let $Y\to X$ be a vector bundle. Using the canonical vertical
splitting (\ref{12f10}), we obtain the canonical vertical vector
field
\mar{z112'}\beq
u_Y=y^i\dr_i \label{z112'}
\eeq
on $Y$, called the Liouville vector field. For instance, the
Liouville vector field on the tangent bundle $TX$ reads
\mar{z112}\beq
u_{TX}=\dot x^\la\dot\dr_\la. \label{z112}
\eeq
Accordingly, any vector field $\tau=\tau^\la\dr_\la$ on a manifold
$X$ has the canonical vertical lift
\mar{z111}\beq
\tau_V=\tau^\la\dot\dr_\la \label{z111}
\eeq
onto the tangent bundle $TX$.

A multivector field $\vt$ of degree $\nm\vt=r$ (or, simply, an
$r$-vector field) on a manifold $Z$ is a section
\mar{cc6}\beq
\vt =\frac{1}{r!}\vt^{\la_1\dots\la_r} \dr_{\la_1}\w\cdots\w
\dr_{\la_r} \label{cc6}
\eeq
of the exterior product $\op\w^r TZ\to Z$.  Let $\cT_r(Z)$ denote
the $C^\infty(Z)$-module space of $r$-vector fields on $Z$. All
multivector fields on a manifold $Z$ make up the graded
commutative algebra $\cT_*(Z)$ of global sections of the exterior
bundle $\w TZ$ (\ref{ss12f11}) with respect to the exterior
product $\w$.

Given an $r$-vector field $\vt$ (\ref{cc6}) on a manifold $Z$, its
tangent lift $\wt \vt$ onto the tangent bundle $TZ$ of $Z$ is
defined by the relation
\mar{gm60'}\beq
\wt \vt(\wt\si^r,\ldots,\wt\si^1) =\wt{\vt(\si^r,\ldots,\si^1)}
\label{gm60'}
\eeq
where \cite{grab}:

$\bullet$ $\si^k=\si^k_\la dz^\la$ are arbitrary one-forms on a
manifold $Z$,

$\bullet$ by
\be
\wt \si^k =\dot z^\m\dr_\m\si^k_\la dz^\la + \si^k_\la d\dot z^\la
\ee
are meant their tangent lifts (\ref{gm61}) onto the tangent bundle
$TZ$ of $Z$,

$\bullet$ the right-hand side of the equality (\ref{gm60'}) is the
tangent lift (\ref{gm62}) onto $TZ$ of the function
$\vt(\si^r,\ldots,\si^1)$ on $Z$.

The tangent lift (\ref{gm60'}) takes the coordinate form
\mar{gm60}\ben
&& \wt\vt =\frac{1}{r!}[\dot z^\m\dr_\m\vt^{\la_1\dots\la_r}
\dot\dr_{\la_1}\w\cdots\w \dot\dr_{\la_r} + \label{gm60}\\
&& \qquad \vt^{\la_1\dots\la_r}\op\sum_{i=1}^r
\dot\dr_{\la_1}\w\cdots\w\dr_{\la_i}\w\cdots\w
\dot\dr_{\la_r}].\nonumber
\een
In particular, if $\tau$ is a vector field on a manifold $Z$, its
tangent lift (\ref{gm60}) coincides with the functorial lift
(\ref{l27}).

\subsection{Differential forms}

An exterior $r$-form on a manifold $Z$ is a section
\be
\f =\frac{1}{r!}\f_{\la_1\dots\la_r} dz^{\la_1}\w\cdots\w
dz^{\la_r}
\ee
of the exterior product $\op\w^r T^*Z\to Z$, where
\be
&& dz^{\la_1}\w\cdots\w dz^{\la_r}=
\frac{1}{r!}\e^{\la_1\ldots\la_r}{}_{\m_1\ldots\m_r}dz^{\m_1}\ot\cdots\ot
dz^{\m_r},\\
&& \e^{\ldots \la_i\ldots\la_j\ldots}{}_{\ldots
\m_p\ldots\m_k\ldots}= -\e^{\ldots
\la_j\ldots\la_i\ldots}{}_{\ldots \m_p\ldots\m_k\ldots} = -
\e^{\ldots \la_i\ldots\la_j\ldots}{}_{\ldots
\m_k\ldots\m_p\ldots}, \\
&& \e^{\la_1\ldots\la_r}{}_{\la_1\ldots\la_r}=1.
\ee
Sometimes, it is convenient to write
\be
\f =\f'_{\la_1\dots\la_r} dz^{\la_1}\w\cdots\w dz^{\la_r}
\ee
without the coefficient $1/r!$.

Let $\cO^r(Z)$ denote the $C^\infty(Z)$-module of exterior
$r$-forms on a manifold $Z$. By definition, $\cO^0(Z)=C^\infty(Z)$
is the ring of smooth real functions on $Z$. All exterior forms on
$Z$ constitute the graded algebra $\cO^*(Z)$ of global sections of
the exterior bundle $\w T^*Z$ (\ref{ss12f11}) endowed with the
exterior product
\be
&&\f=\frac{1}{r!}\f_{\la_1\dots\la_r} dz^{\la_1}\w\cdots\w
dz^{\la_r}, \qquad \si= \frac{1}{s!}\si_{\m_1\dots\m_s}
dz^{\m_1}\w\cdots\w dz^{\m_s},\\
&& \f\w\si=\frac{1}{r!s!}\f_{\nu_1\ldots\nu_r}\si_{\nu_{r+1}\ldots\nu_{r+s}}
dz^{\nu_1}\w\cdots\w
dz^{\nu_{r+s}}=\\
&& \qquad
\frac{1}{r!s!(r+s)!}\e^{\nu_1\ldots\nu_{r+s}}{}_{\al_1\ldots\al_{r+s}}
\f_{\nu_1\ldots\nu_r}\si_{\nu_{r+1}\ldots\nu_{r+s}}dz^{\al_1}\w\cdots\w
dz^{\al_{r+s}},
\ee
such  that
\be
\f\w\si=(-1)^{|\f||\si|}\si\w\f,
\ee
where the symbol $|\f|$ stands for the form degree. The algebra
$\cO^*(Z)$ also is provided with the exterior differential
\be
d\f= dz^\m\w \dr_\m\f=\frac{1}{r!} \dr_\m\f_{\la_1\ldots\la_r}
dz^\m\w dz^{\la_1}\w\cdots \w dz^{\la_r}
\ee
which obeys the relations
\be
d\circ d=0, \qquad d(\f\w\si)= d(\f)\w \si +(-1)^{|\f|}\f\w
d(\si).
\ee
The exterior differential $d$ makes $\cO^*(Z)$ into a differential
graded algebra, called the exterior algebra.

Given a manifold morphism $f:Z\to Z'$, any exterior $k$-form $\f$
on $Z'$ yields the pull-back exterior form $f^*\f$ on $Z$ given by
the condition
\be
f^*\f(v^1,\ldots,v^k)(z) = \f(Tf(v^1),\ldots,Tf(v^k))(f(z))
\ee
for an arbitrary collection of tangent vectors $v^1,\cdots, v^k\in
T_zZ$. We have the relations
\be
f^*(\f\w\si) =f^*\f\w f^*\si, \qquad df^*\f =f^*(d\f).
\ee

In particular, given a fibred manifold $\pi:Y\to X$, the pull-back
onto $Y$ of exterior forms on $X$ by $\pi$ provides the
monomorphism of graded commutative algebras $\cO^*(X)\to
\cO^*(Y)$. Elements of its range $\pi^*\cO^*(X)$ are called basic
forms. Exterior forms
\be
\phi : Y\to\op\w^r T^*X, \qquad \phi
=\frac{1}{r!}\phi_{\la_1\ldots\la_r}dx^{\la_1}\w\cdots\w
dx^{\la_r},
\ee
on $Y$ such that $u\rfloor\f=0$ for an arbitrary vertical vector
field $u$ on $Y$ are said to be horizontal forms. Horizontal forms
of degree $n=\dim X$ are called densities.

In the case of the tangent bundle $TX\to X$, there is a different
way to lift exterior forms on $X$ onto $TX$ \cite{grab,leon}. Let
$f$ be a function on $X$. Its tangent lift onto $TX$ is defined as
the function
\mar{gm62}\beq
\wt f=\dot x^\la\dr_\la f. \label{gm62}
\eeq
Let $\si$ be an $r$-form on $X$. Its tangent lift onto $TX$ is
said to be the $r$-form $\wt \si$ given by the relation
\mar{z115}\beq
\wt \si(\wt\tau_1,\ldots,\wt\tau_r)=
\wt{\si(\tau_1,\ldots,\tau_r)}, \label{z115}
\eeq
where $\tau_i$ are arbitrary vector fields on $X$ and $\wt\tau_i$
are their functorial lifts (\ref{l27}) onto $TX$. We have the
coordinate expression
\mar{gm61}\ben
&& \si=\frac{1}{r!}\si_{\la_1\cdots\la_r}dx^{\la_1}\w\cdots\w dx^{\la_r},
\nonumber\\
&& \wt\si =\frac1{r!}[\dot x^\m\dr_\m
\si_{\la_1\cdots\la_r}dx^{\la_1}\w\cdots\w dx^{\la_r}+ \label{gm61}\\
&&\qquad \op\sum_{i=1}^r
\si_{\la_1\cdots\la_r}dx^{\la_1}\w\cdots\w d\dot
x^{\la_i}\w\cdots\w dx^{\la_r}]. \nonumber
\een
The following equality holds:
\be
d\wt\si=\wt{d\si}.
\ee

The interior product (or contraction) of a vector field $u$ and an
exterior $r$-form $\f$ on a manifold $Z$ is given by the
coordinate expression
\be
&& u\rfloor\f = \op\sum_{k=1}^r \frac{(-1)^{k-1}}{r!} u^{\la_k}
\f_{\la_1\ldots\la_k\ldots\la_r} dz^{\la_1}\w\cdots\w\wh
{dz}^{\la_k}\w\cdots \w dz^{\la_r}= \\
&& \qquad \frac{1}{(r-1)!}u^\m\f_{\m\al_2\ldots\al_r} dz^{\al_2}\w\cdots\w
dz^{\al_r},
\ee
where the caret $\,\wh{}\,$ denotes omission. It obeys the
relations
\mar{031}\ben
&& \f(u_1,\ldots,u_r)=u_r\rfloor\cdots u_1\rfloor\f,\nonumber\\
&& u\rfloor(\f\w\si)= u\rfloor\f\w\si +(-1)^{|\f|}\f\w
u\rfloor\si. \label{031}
\een

A generalization of the interior product to multivector fields is
the left interior product
\be
\vt\rfloor\f=\f(\vt), \qquad \nm\vt\leq\nm\f, \qquad
\f\in\cO^*(Z), \qquad \vt\in \cT_*(Z),
\ee
of multivector fields and exterior forms. It is defined by the
equalities
\be
\f(u_1\w\cdots\w u_r)=\f(u_1,\ldots,u_r), \qquad \f\in\cO^*(Z),
\qquad u_i\in\cT_1(Z),
\ee
and obeys the relation
\be
\vt\rfloor \up\rfloor\f=(\up\w\vt)\rfloor\f=(-1)^{\nm\up\nm\vt}
\up\rfloor \vt\rfloor\f, \qquad \f\in\cO^*(Z), \qquad \vt,\up\in
\cT_*(Z).
\ee

The Lie derivative of an exterior form $\f$ along a vector field
$u$ is
\mar{034,5}\ben
&& \bL_u\f = u\rfloor d\f +d(u\rfloor\f), \label{034}\\
&& \bL_u(\f\w\si)= \bL_u\f\w\si +\f\w\bL_u\si. \label{035}
\een
In particular, if $f$ is a function, then
\be
\bL_u f =u(f)=u\rfloor d f.
\ee
An exterior form $\f$ is invariant under a local one-parameter
group of diffeomorphisms $G_t$ of $Z$ (i.e., $G_t^*\f=\f$) iff its
Lie derivative along the infinitesimal generator $u$ of this group
vanishes, i.e.,
\be
\bL_u\f=0.
\ee
Following physical terminology (Definition \ref{084}), we say that
a vector field $u$ is a symmetry of an exterior form $\f$.

A tangent-valued $r$-form on a manifold $Z$ is a section
\mar{spr611}\beq
\phi = \frac{1}{r!}\phi_{\la_1\ldots\la_r}^\m dz^{\la_1}\w\cdots\w
dz^{\la_r}\ot\dr_\m \label{spr611}
\eeq
of the tensor bundle
\be
\op\w^r T^*Z\ot TZ\to Z.
\ee

\begin{rem}
There is one-to-one correspondence between the tangent-valued
one-forms $\f$ on a manifold $Z$ and the linear bundle
endomorphisms
\mar{29b,b'}\ben
&& \wh\f:TZ\to TZ,\quad
\wh\f: T_zZ\ni v\to v\rfloor\f(z)\in T_zZ, \label{29b} \\
&&\wh\f^*:T^*Z\to T^*Z,\quad \wh\f^*: T_z^*Z\ni v^*\to
\f(z)\rfloor v^*\in T_z^*Z, \label{29b'}
\een
over $Z$ (Remark \ref{mos30}). For instance, the canonical
tangent-valued one-form
\mar{b1.51}\beq
\thh_Z= dz^\la\ot \dr_\la \label{b1.51}
\eeq
on $Z$ corresponds to the identity morphisms (\ref{29b}) and
(\ref{29b'}).
\end{rem}

\begin{rem}
Let $Z=TX$, and let $TTX$ be the tangent bundle of $TX$. It is
called the double tangent bundle. There is the bundle endomorphism
\mar{z117}\beq
J(\dr_\la)= \dot\dr_\la, \qquad J(\dot\dr_\la)=0 \label{z117}
\eeq
of $TTX$ over $X$. It  corresponds to the canonical tangent-valued
form
\mar{z117'}\beq
\thh_J=dx^\la\ot\dot\dr_\la \label{z117'}
\eeq
on the tangent bundle $TX$. It is readily observed that $J\circ
J=0$.
\end{rem}

The space $\cO^*(Z)\ot \cT_1(Z)$ of tangent-valued forms is
provided with the Fr\"olicher--Nijenhuis bracket
\mar{1149}\ben
&& [,]_{\rm FN}:\cO^r(Z)\ot \cT_1(Z)\times \cO^s(Z)\ot \cT_1(Z)
\to\cO^{r+s}(Z)\ot \cT_1(Z), \nonumber \\
&& [\al\ot u,\, \bt\ot v]_{\rm FN} = (\al\w\bt)\ot [u, v] +
(\al\w \bL_u \bt)\ot v - \label{1149} \\
&& \qquad(\bL_v \al\w\bt)\ot u +  (-1)^r (d\al\w u\rfloor\bt)\ot v
+ (-1)^r(v\rfloor\al\w d\bt)\ot u, \nonumber \\
&&  \al\in\cO^r(Z), \qquad \bt\in \cO^s(Z), \qquad u,v\in\cT_1(Z). \nonumber
\een
Its coordinate expression is
\be
&& [\phi,\si]_{\rm FN} = \frac{1}{r!s!}(\phi_{\la_1
\dots\la_r}^\nu\dr_\n\si_{\la_{r+1}\dots\la_{r+s}}^\m -
\si_{\la_{r+1}
\dots\la_{r+s}}^\nu\dr_\nu\phi_{\la_1\dots\la_r}^\m -\\
&& \qquad r\phi_{\la_1\ldots\la_{r-1}\nu}^\m\dr_{\la_r}\si_{\la_{r+1}
\dots\la_{r+s}}^\nu + s \si_{\nu\la_{r+2}\ldots\la_{r+s}}^\m
\dr_{\la_{r+1}}\phi_{\la_1\ldots\la_r}^\nu)\\
&& \qquad dz^{\la_1}\wedge\cdots \wedge dz^{\la_{r+s}}\otimes\dr_\m,\\
&&
\f\in \cO^r(Z)\ot \cT_1(Z), \qquad \si\in \cO^s(Z)\ot \cT_1(Z).
\ee
There are the relations
\mar{1150,'}\ben
&& [\f,\si]_{\rm FN}=(-1)^{|\f||\s|+1}[\si,\f]_{\rm FN}, \label{1150} \\
&& [\f, [\si, \thh]_{\rm FN}]_{\rm FN} = [[\f, \si]_{\rm FN}, \thh]_{\rm
FN} +(-1)^{|\f||\si|}  [\si, [\f,\thh]_{\rm FN}]_{\rm FN}, \label{1150'}\\
&& \f,\si,\thh\in \cO^*(Z)\ot \cT_1(Z). \nonumber
\een

Given a tangent-valued form  $\thh$, the Nijenhuis differential on
$\cO^*(Z)\ot\cT_1(Z)$ is defined as the morphism
\be
d_\thh : \psi\to d_\thh\psi = [\thh,\psi]_{\rm FN}, \qquad
\psi\in\cO^*(Z)\ot\cT_1(Z).
\ee
By virtue of (\ref{1150'}), it has the property
\be
d_\f[\psi,\thh]_{\rm FN} = [d_\phi\psi,\thh]_{\rm FN}+
(-1)^{|\f||\psi|} [\psi,d_\f\thh]_{\rm FN}.
\ee
In particular, if $\thh=u$ is a vector field, the Nijenhuis
differential is the Lie derivative of tangent-valued forms
\be
&& \bL_u\si= d_u\si=[u,\si]_{\rm FN} =\frac{1}{s!}(u^\n\dr_\n\si_{\la_1\ldots\la_s}^\m -
\si_{\la_1\ldots\la_s}^\n\dr_\n u^\m +\\
&& \qquad s\si^\m_{\nu\la_2\ldots\la_s}\dr_{\la_1}u^\nu)dx^{\la_1}
\w\cdots\w dx^{\la_s}\ot\dr_\m, \qquad \si\in\cO^s(Z)\ot\cT_1(Z).
\ee

Let $\Y$ be a fibred manifold. We consider the following subspaces
of the space $\cO^*(Y)\ot \cT_1(Y)$ of tangent-valued forms on
$Y$:

$\bullet$ horizontal tangent-valued forms
\be
&& \phi : Y\to\op\w^r T^*X\op\otimes_Y TY,\\
&& \phi =dx^{\la_1}\wedge\cdots\wedge dx^{\la_r}\otimes
\frac{1}{r!}[\phi_{\la_1\ldots\la_r}^\m(y) \dr_\m
+\phi_{\la_1\ldots\la_r}^i(y) \dr_i],
\ee

$\bullet$ projectable horizontal tangent-valued forms
\be
\phi =dx^{\la_1}\wedge\cdots\wedge dx^{\la_r}\otimes
\frac{1}{r!}[\phi_{\la_1\ldots\la_r}^\m(x)\dr_\m
+\phi_{\la_1\ldots\la_r}^i(y) \dr_i],
\ee

$\bullet$ vertical-valued form
\be
\phi : Y\to\op\w^r T^*X\op\otimes_Y VY,\quad \phi
=\frac{1}{r!}\phi_{\la_1\ldots\la_r}^i(y)dx^{\la_1}\wedge\cdots
\wedge dx^{\la_r}\otimes\dr_i,
\ee

$\bullet$ vertical-valued one-forms, called soldering forms,
\mar{1290}\beq
\si = \si_\la^i(y) dx^\la\otimes\dr_i, \label{1290}
\eeq

$\bullet$ basic soldering forms
\be
\si = \si_\la^i(x) dx^\la\otimes\dr_i.
\ee

\begin{rem} \label{mos161} \mar{mos161}
The tangent bundle $TX$ is provided with the canonical soldering
form $\thh_J$ (\ref{z117'}). Due to the canonical vertical
splitting
\mar{mos163}\beq
VTX=TX\op\times_X TX, \label{mos163}
\eeq
the canonical soldering form (\ref{z117'}) on $TX$ defines the
canonical tangent-valued form $\thh_X$ (\ref{b1.51}) on $X$. By
this reason, tangent-valued one-forms on a manifold $X$ also are
called soldering forms.
\end{rem}

We also mention the  $TX$-valued forms
\mar{1.11}\ben
&&\f:Y\to \op\w^r T^*X\op\ot_Y TX, \label{1.11}\\
&&\phi =\frac{1}{r!}\phi_{\la_1\ldots\la_r}^\m  dx^{\la_1}\w\cdots\w dx^{\la_r}
\otimes \dr_\m,\nonumber
\een
and  $V^*Y$-valued forms
\mar{87}\ben
&& \f :Y\to \op\w^r T^*X\op\ot_Y V^*Y, \label{87}\\
&& \phi =\frac{1}{r!}\phi_{\la_1\ldots\la_ri}dx^{\la_1}\w\cdots\w dx^{\la_r}\ot
\ol dy^i. \nonumber
\een
It should be emphasized that (\ref{1.11}) are not tangent-valued
forms, while (\ref{87}) are not exterior forms. They exemplify
vector-valued forms. Given a vector bundle $E\to X$, by a
$E$-valued $k$-form on $X$, is meant a section of the fibre bundle
\be
(\op\w^k T^*X)\op\ot_X E^*\to X.
\ee

\subsection{Distributions and foliations}

A subbundle $\bT$ of the tangent bundle $TZ$ of a manifold $Z$ is
called a regular distribution (or, simply, a distribution). A
vector field $u$ on $Z$ is said to be subordinate to a
distribution $\bT$ if it lives in $\bT$. A distribution $\bT$ is
called involutive if the Lie bracket of $\bT$-subordinate vector
fields also is subordinate to $\bT$.

A subbundle of the cotangent bundle $T^*Z$ of $Z$ is called a
codistribution $\bT^*$ on a manifold $Z$. For instance, the
annihilator $\rA\bT$ of a distribution $\bT$ is a codistribution
whose fibre over $z\in Z$ consists of covectors $w\in T^*_z$ such
that $v\rfloor w=0$ for all $v\in \bT_z$.

There is the following criterion of an involutive distribution
\cite{war}.

\begin{theo} \label{warr} \mar{warr} Let $\bT$ be a distribution
and $\rA\bT$ its annihilator. Let $\w\rA\bT(Z)$ be the ideal of
the exterior algebra $\cO^*(Z)$ which is generated by sections of
$\rA\bT\to Z$. A distribution $\bT$ is involutive iff the ideal
$\w\rA\bT(Z)$ is a differential ideal, i.e.,
\be
d(\w\rA\bT(Z))\subset\w\rA\bT(Z).
\ee
\end{theo}

The following local coordinates can be associated to an involutive
distribution  \cite{war}.

\begin{theo}\label{c11.0} \mar{c11.0} Let $\bT$ be an involutive
$r$-dimensional distribution on a manifold $Z$, $\di Z=k$. Every
point $z\in Z$ has an open neighborhood $U$ which is a domain of
an adapted coordinate chart $(z^1,\dots,z^k)$ such that,
restricted to $U$, the distribution $\bT$ and its annihilator
$\rA\bT$ are spanned by the local vector fields $\dr/\dr z^1,
\cdots,\dr/\dr z^r$ and the local one-forms $dz^{r+1},\dots,
dz^k$, respectively.
\end{theo}

A connected submanifold $N$ of a manifold $Z$ is called an
integral manifold of a distribution $\bT$ on $Z$ if $TN\subset
\bT$. Unless otherwise stated, by an integral manifold is meant an
integral manifold of dimension of $\bT$. An integral manifold is
called maximal if no other integral manifold contains it. The
following is the classical theorem of Frobenius \cite{kob,war}.

\begin{theo}\label{to.1}  \mar{to.1} Let $\bT$ be an
involutive distribution on a manifold $Z$. For any $z\in Z$, there
exists a unique maximal integral manifold of $\bT$ through $z$,
and any integral manifold through $z$ is its open subset.
\end{theo}

Maximal integral manifolds of an involutive distribution on a
manifold $Z$ are assembled into a regular foliation $\cF$ of $Z$.

A regular $r$-dimensional foliation (or, simply, a foliation)
$\cF$ of a $k$-dimensional manifold $Z$ is defined as a partition
of $Z$ into connected $r$-dimensional submanifolds (the leaves of
a foliation) $F_\iota$, $\iota\in I$, which possesses the
following properties \cite{rei,tam}.

A manifold $Z$ admits an adapted coordinate atlas
\mar{spr850}\beq
\{(U_\xi;z^\la, z^i)\},\quad \la=1,\ldots,k-r, \qquad
i=1,\ldots,r, \label{spr850}
\eeq
such that transition functions of coordinates $z^\la$ are
independent of the remaining coordinates $z^i$. For each leaf $F$
of a foliation $\cF$, the connected components of $F\cap U_\xi$
are given by the equations $z^\la=$const. These connected
components and coordinates $(z^i)$ on them make up a coordinate
atlas of a leaf $F$. It follows that tangent spaces to leaves of a
foliation $\cF$ constitute an involutive distribution $T\cF$ on
$Z$, called the tangent bundle to the foliation $\cF$. The factor
bundle
\be
V\cF=TZ/T\cF,
\ee
called the normal bundle to $\cF$, has transition functions
independent of coordinates $z^i$. Let $T\cF^*\to Z$ denote the
dual of $T\cF\to Z$. There are the exact sequences
\mar{pp2,3}\ben
&& 0\to T\cF \ar^{i_\cF} TX \ar V\cF\to 0, \label{pp2} \\
&& 0\to {\rm Ann}\,T\cF\ar T^*X\ar^{i^*_\cF} T\cF^* \to 0
\label{pp3}
\een
of vector bundles over $Z$.

A pair $(Z,\cF)$, where $\cF$ is a foliation of $Z$, is called a
foliated manifold. It should be emphasized that leaves of a
foliation need not be closed or imbedded submanifolds. Every leaf
has an open saturated neighborhood $U$, i.e., if $z\in U$, then a
leaf through $z$ also belongs to $U$.

Any submersion $\zeta:Z\to M$ yields a foliation
\be
\cF=\{F_p=\zeta^{-1}(p)\}_{p\in \zeta(Z)}
\ee
of $Z$ indexed by elements of $\zeta(Z)$, which is an open
submanifold of $M$, i.e., $Z\to \zeta(Z)$ is a fibred manifold.
Leaves of this foliation are closed imbedded submanifolds. Such a
foliation is called simple. Any (regular) foliation is locally
simple.

\section{Jet manifolds}

This Section collects the relevant material on jet manifolds of
sections of fibre bundles \cite{book09,kol,book00,sau}.

\subsection{First order jet manifolds}

Given a fibre bundle $Y\to X$ with bundle coordinates
$(x^\la,y^i)$, let us consider the equivalence classes $j^1_xs$ of
its sections $s$, which are identified by their values $s^i(x)$
and the values of their partial derivatives $\dr_\mu s^i(x)$ at a
point $x\in X$. They are called the first order jets of sections
at $x$. One can justify that the definition of jets is
coordinate-independent. A key point is that the set $J^1Y$ of
first order jets $j^1_xs$, $x\in X$, is a smooth manifold with
respect to the adapted coordinates $(x^\la,y^i,y_\la^i)$ such that
\mar{50}\beq
y_\la^i(j^1_xs)=\dr_\la s^i(x),\qquad {y'}^i_\la = \frac{\dr
x^\m}{\dr{x'}^\la}(\dr_\m +y^j_\m\dr_j)y'^i.\label{50}
\eeq
It is called the first order jet manifold of a fibre bundle $Y\to
X$. We call $(y_\la^i)$ the jet coordinate.

A jet manifold $J^1Y$ admits the natural fibrations
\mar{1.14,5}\ben
&&\pi^1:J^1Y\ni j^1_xs\to x\in X, \label{1.14}\\
&&\pi^1_0:J^1Y\ni j^1_xs\to s(x)\in Y. \label{1.15}
\een
A glance at the transformation law (\ref{50}) shows that $\pi^1_0$
is an affine bundle modelled over the vector bundle
\mar{cc9}\beq
T^*X \op\otimes_Y VY\to Y.\label{cc9}
\eeq
It is convenient to call $\pi^1$ (\ref{1.14}) the jet bundle,
while $\pi^1_0$ (\ref{1.15}) is said to be the affine jet bundle.

Let us note that, if $Y\to X$ is a vector or an affine bundle, the
jet bundle $\pi_1$ (\ref{1.14}) is so.

Jets can be expressed in terms of familiar tangent-valued forms as
follows. There are the canonical imbeddings
\mar{18,24}\ben
&&\la_{(1)}:J^1Y\op\to_Y
T^*X \op\otimes_Y TY,\nonumber\\
&& \la_{(1)}=dx^\la
\otimes (\dr_\la + y^i_\la \dr_i)=dx^\la\otimes d_\la, \label{18}\\
&&\thh_{(1)}:J^1Y \op\to_Y T^*Y\op\otimes_Y VY,\nonumber\\
&&\thh_{(1)}=(dy^i- y^i_\la dx^\la)\otimes \dr_i=\thh^i \otimes
\dr_i,\label{24}
\een
where $d_\la$ are said to be total derivatives,  and $\thh^i$ are
called contact forms.

We further identify the jet manifold $J^1Y$ with its images under
the canonical morphisms (\ref{18}) and (\ref{24}), and represent
the jets $j^1_xs=(x^\la,y^i,y^i_\m)$ by the tangent-valued forms
$\la_{(1)}$ (\ref{18}) and $\thh_{(1)}$ (\ref{24}).

Sections and morphisms of fibre bundles admit prolongations to jet
manifolds as follows.

Any section $s$ of a fibre bundle $Y\to X$ has the jet
prolongation to the section
\be
(J^1s)(x)= j_x^1s, \qquad y_\la^i\circ J^1s= \dr_\la s^i(x),
\ee
of the jet bundle $J^1Y\to X$. A section of the jet bundle
$J^1Y\to X$ is called integrable if it is the jet prolongation of
some section of a fibre bundle $Y\to X$.

Any bundle morphism $\Phi:Y\to Y'$ over a diffeomorphism $f$
admits a jet prolongation to a bundle morphism of affine jet
bundles
\mar{1.21a}\beq
J^1\Phi : J^1Y \ar_\Phi J^1Y', \qquad {y'}^i_\la\circ
J^1\Phi=\frac{\dr(f^{-1})^\m}{\dr x'^\la}d_\m\Phi^i. \label{1.21a}
\eeq

Any projectable vector field $u$ (\ref{11f30}) on a fibre bundle
$Y\to X$ has a jet prolongation to the projectable vector field
\mar{1.21}\ben
&&J^1u =r_1\circ J^1u: J^1Y\to J^1TY\to TJ^1Y,\nonumber \\
&& J^1u =u^\la\dr_\la + u^i\dr_i + (d_\la u^i
- y_\m^i\dr_\la u^\m)\dr_i^\la, \label{1.21}
\een
on the jet manifold $J^1Y$. In order to obtain (\ref{1.21}), the
canonical bundle morphism
\be
r_1: J^1TY\to TJ^1Y,\qquad \dot y^i_\la\circ r_1 = (\dot
y^i)_\la-y^i_\m\dot x^\m_\la
\ee
is used. In particular, there is the canonical isomorphism
\mar{d020}\beq
VJ^1Y=J^1VY, \qquad \dot y^i_\la=(\dot y^i)_\la.\label{d020}
\eeq

\subsection{Second order jet manifolds}

Taking the first order jet manifold  of the jet bundle $J^1Y\to
X$, we obtain the repeated jet manifold $J^1J^1Y$ provided with
the adapted coordinates
\be
(x^\la ,y^i,y^i_\la ,\wh y_\m^i,y^i_{\m\la})
\ee
possessing transition functions
\be
&& y'^i_\la = \frac{\dr x^\al}{\dr{x'}^\la} d_\al y'^i,
\qquad \wh y'^i_\la = \frac{\dr x^\al}{\dr{x'}^\la} \wh d_\al
y'^i, \qquad
{y'}_{\m\la}^i= \frac{\dr x^\al}{\dr{x'}^\m}\wh d_\al {y'}^i_\la,\\
&& d_\al = \dr_\al + y^j_\al\dr_j +y^j_{\nu\al}\dr^\nu_j,
\qquad \wh d_\al = \dr_\al +\wh y^j_\al\dr_j
+y^j_{\nu\al}\dr^\nu_j.
\ee
There exist two different affine fibrations of $J^1J^1Y$ over
$J^1Y$:

$\bullet$ the familiar affine jet bundle (\ref{1.15}):
\mar{gm213}\beq
\pi_{11}:J^1J^1Y\to J^1Y, \qquad y_\la^i\circ\p_{11} = y_\la^i,
\label{gm213}
\eeq

$\bullet$ the affine bundle
\mar{gm214}\beq
J^1\pi^1_0:J^1J^1Y\to J^1Y,\qquad y_\la^i\circ J^1\pi_0^1 = \wh
y_\la^i. \label{gm214}
\eeq
\noindent In general, there is no canonical identification of
these fibrations. The points $q\in J^1J^1Y$, where
\be
\pi_{11}(q)=J^1\pi^1_0(q),
\ee
form an affine subbundle $\wh J^2Y\to J^1Y$ of $J^1J^1Y$ called
the sesquiholonomic jet manifold. It is given by the coordinate
conditions $\wh y^i_\la= y^i_\la$, and is coordinated by  $(x^\la
,y^i, y^i_\la,y^i_{\m\la})$.

The second order (or holonomic) jet manifold $J^2Y$ of a fibre
bundle $Y\to X$ can be defined as the affine subbundle of the
fibre bundle $\wh J^2Y\to J^1Y$ given by the coordinate conditions
$y^i_{\la\m}=y^i_{\m\la}$. It is modelled over the vector bundle
\be
\op\vee^2T^*X\op\ot_{J^1Y}VY\to J^1Y,
\ee
and is endowed with adapted coordinates $(x^\la ,y^i,
y^i_\la,y^i_{\la\m}=y^i_{\m\la})$, possessing transition functions
\mar{12f80}\beq
y'^i_\la = \frac{\dr x^\al}{\dr{x'}^\la} d_\al y'^i, \qquad
{y'}_{\m\la}^i= \frac{\dr x^\al}{\dr{x'}^\m} d_\al {y'}^i_\la.
\label{12f80}
\eeq

The second order jet manifold $J^2Y$ also can be introduced as the
set of the equivalence classes $j_x^2s$ of sections $s$ of the
fibre bundle $Y\to X$, which are identified by their values and
the values of their first and second order partial derivatives at
points $x\in X$, i.e.,
\be
y^i_\la (j_x^2s)=\dr_\la s^i(x),\qquad
y^i_{\la\m}(j_x^2s)=\dr_\la\dr_\m s^i(x).
\ee
The equivalence classes $j^2_xs$ are called the second order jets
of sections.

Let $s$ be a section of a fibre bundle $Y\to X$, and let $J^1s$ be
its jet prolongation to a section of a jet bundle $J^1Y\to X$. The
latter gives rise to the section $J^1J^1s$ of the repeated jet
bundle $J^1J^1Y\to X$. This section takes its values into the
second order jet manifold $J^2Y$. It is called the second order
jet prolongation of a section $s$, and is denoted by $J^2s$.

\begin{theo}\label{1.5.3} \mar{1.5.3} Let $\ol s$ be a section of the
jet bundle $J^1Y\to X$, and let $J^1\ol s$ be its jet prolongation
to a section of the repeated jet bundle $J^1J^1Y\to X$. The
following three facts are equivalent:

$\bullet$ $\ol s=J^1s$ where $s$ is a section of a fibre bundle
$Y\to X$,

$\bullet$ $J^1\ol s$ takes its values into $\wh J^2Y$,

$\bullet$ $J^1\ol s$ takes its values into $J^2Y$.
\end{theo}

\subsection{Higher order jet manifolds}

The notion of first and second order jet manifolds is naturally
extended to higher order jet manifolds.

The $k$-order jet manifold $J^kY$ of a fibre bundle $Y\to X$
comprises the equivalence classes $j^k_xs$, $x\in X$, of sections
$s$ of $Y$ identified by the $k+1$ terms of their Tailor series at
points $x\in X$. The jet manifold $J^kY$ is provided with the
adapted coordinates
\be
&&(x^\la, y^i,y^i_\la, \ldots, y^i_{\la_k\cdots\la_1}),\\
&& y^i_{\la_l\cdots\la_1}(j^k_xs)= \dr_{\la_l}\cdots \dr_{\la_1}s^i(x),
\qquad 0\leq l\leq k.
\ee
Every section $s$ of a fibre bundle $Y\to X$ gives rise to the
section $J^ks$ of a fibre bundle $J^kY\to X$ such that
\be
y^i_{\la_l\cdots\la_1}\circ J^ks=\dr_{\la_l}\cdots \dr_{\la_1}s^i,
\qquad 0\leq l\leq k.
\ee

The following operators on exterior forms on jet manifolds are
utilized:

$\bullet$ the total derivative operator
\mar{+410}\beq
d_\la =\dr_\la +y^i_\la\dr_i +y^i_{\la\m}\dr_i^\m +\cdots,
\label{+410}
\eeq
obeying the relations
\be
&& d_\la(\f\w\si)=d_\la(\f)\w\si +\f\w d_\la(\si),\\
&& d_\la(d\f)=d(d_\la(\f)),
\ee
in particular,
\be
&& d_\la (f) = \dr_\la f +y^i_\la\dr_if +y^i_{\la\m}\dr_i^\m f +\cdots,
\qquad f\in C^\infty (J^kY), \\
&& d_\la(dx^\m)=0, \qquad d_\la(dy^i_{\la_l\cdots\la_1})=
dy^i_{\la\la_l\cdots\la_1};
\ee

$\bullet$ the horizontal projection $h_0$ given by the relations
\mar{mos40}\beq
h_0(dx^\la)=dx^\la, \qquad h_0(dy^i_{\la_k\cdots\la_1})=
y^i_{\m\la_k\ldots\la_1}dx^\m, \label{mos40}
\eeq
in  particular,
\be
h_0(dy^i)=y^i_\mu dx^\mu, \qquad h_0(dy^i_\la)=
y^i_{\mu\la}dx^\mu;
\ee

$\bullet$ the total differential
\mar{mos41}\beq
d_H(\f)=dx^\la\w d_\la(\f), \label{mos41}
\eeq
possessing the properties
\be
d_H\circ d_H=0,\qquad h_0\circ d=d_H\circ h_0.
\ee

\subsection{Differential operators and differential equations}

Jet manifolds provide the standard language for the theory of
differential equations and differential operators
\cite{bry,book,kras}.

\begin{defi}\label{ch529} \mar{ch529}
Let $Z$ be an $(m+n)$-dimensional manifold. A system of $k$-order
partial differential equations (or, simply, a differential
equation) in $n$ variables on $Z$ is defined to be a closed smooth
submanifold $\gE$ of the $k$-order jet bundle $J^k_nZ$ of
$n$-dimensional submanifolds of $Z$.
\end{defi}

By its solution is meant an $n$-dimensional submanifold $S$ of $Z$
whose $k$-order jets $[S]^k_z$, $z\in S$, belong to $\gE$.

\begin{defi} \label{mos04} \mar{mos04}
A $k$-order differential equation in $n$ variables on a manifold
$Z$ is called a dynamic equation if it can be algebraically solved
for the highest order derivatives, i.e., it is a section of the
fibration $J^k_nZ\to J^{k-1}_nZ$.
\end{defi}

In particular, a first order dynamic equation in $n$ variables on
a manifold $Z$ is a section of the jet bundle $J^1_nZ\to Z$. Its
image in the tangent bundle $TZ\to Z$ is an $n$-dimensional vector
subbundle of $TZ$. If $n=1$, a dynamic equation is given by a
vector field
\mar{mos011}\beq
\dot z^\la(t)=u^\la(z(t)) \label{mos011}
\eeq
on a manifold $Z$. Its solutions are integral curves $c(t)$ of the
vector field $u$.

Let $Y\to X$ be a fibre bundle. There are several equivalent
definitions of (non-linear) differential operators. We start with
the following.

\begin{defi}\label{ch538} \mar{ch538}
Let $E\to X$ be a vector bundle. A $k$-order $E$-valued
differential operator on a fibre bundle $Y\to X$ is defined as a
section $\cE$ of the pull-back bundle
\mar{5.113}\beq
{\rm pr}_1:E^k_Y= J^kY\op\times_X E \to J^kY. \label{5.113}
\eeq
\end{defi}

Given bundle coordinates $(x^\la, y^i)$ on $Y$ and $(x^\la,
\chi^a)$ on $E$, the pull-back (\ref{5.113}) is provided with
coordinates $(x^\la, y^j_\Si,\chi^a)$, $0\leq|\Si|\leq k$. With
respect to these coordinates, a differential operator $\cE$ seen
as a closed imbedded submanifold $\cE\subset E^k_Y$ is given by
the equalities
\mar{5.132}\beq
\chi^a = \cE^a(x^\la, y^j_\Si). \label{5.132}
\eeq

There is obvious one-to-one correspondence between the sections
$\cE$ (\ref{5.132}) of the fibre bundle (\ref{5.113}) and the
bundle morphisms
\mar{5.115}\ben
&&\Phi: J^kY\ar_X E, \label{5.115}\\
&& \Phi= {\rm pr}_2\circ \cE \, \Longleftrightarrow \, \cE
=(\id J^kY, \Phi). \nonumber
\een
Therefore, we come to the following equivalent definition of
differential operators on $Y\to X$.

\begin{defi}\label{oper} \mar{oper}
A fibred morphism
\mar{gm2}\beq
\cE: J^kY\op\to_X E \label{gm2}
\eeq
is called a $k$-order differential operator on the fibre bundle
$Y\to X$. It sends each section $s(x)$ of $Y\to X$ onto the
section $(\cE\circ J^ks)(x)$ of the vector bundle $E\to X$
\cite{bry,kras}.
\end{defi}

The kernel of a differential operator is the subset
\mar{z60}\beq
\Ker\cE=\cE^{-1}(\wh 0(X))\subset J^kY, \label{z60}
\eeq
where $\wh 0$ is the zero section of the vector bundle $E\to X$,
and we assume that $\wh 0(X)\subset \cE(J^kY)$.

\begin{defi}\label{equa} \mar{equa}
A system of $k$-order partial differential equations (or, simply,
a differential equation) on a fibre bundle $Y\to X$ is defined as
a closed subbundle $\gE$ of the jet bundle $J^kY\to X$.
\end{defi}

Its solution is a (local) section $s$ of the fibre bundle $Y\to X$
such that its $k$-order jet prolongation $J^ks$ lives in $\gE$.

For instance, if the kernel (\ref{z60}) of a differential operator
$\cE$ is a closed subbundle of the fibre bundle $J^kY\to X$, it
defines a differential equation
\be
\cE\circ J^ks=0.
\ee

The following condition is sufficient for a kernel of a
differential operator to be a differential equation.

\begin{theo}\label{oper2} \mar{oper2}
Let the morphism (\ref{gm2}) be of constant rank. Its kernel
(\ref{z60}) is a closed subbundle of the fibre bundle $J^kY\to X$
and, consequently, is a $k$-order differential equation.
\end{theo}

\section{Connections on fibre bundles}

There are different equivalent definitions of a connection on a
fibre bundle $Y\to X$. We define it both as a splitting of the
exact sequence (\ref{1.8a}) and a global section of the affine jet
bundle $J^1Y\to $Y \cite{book09,book00,sau}.

\subsection{Connections}

A connection on a fibred manifold $Y\to X$ is defined as a
splitting (called the horizontal splitting)
\mar{150}\ben
&& \G: Y\op\times_X TX\op\to_Y TY, \qquad
   \G: \dot x^\la\dr_\la \mapsto \dot x^\la(\dr_\la+\G^i_\la(y)\dr_i),
\label{150}\\
&& \dot x^\la \dr_\la + \dot y^i \dr_i = \dot
x^\la(\dr_\la + \G_\la^i \dr_i) + (\dot y^i - \dot
x^\la\G_\la^i)\dr_i.  \nonumber
\een
of the exact sequence (\ref{1.8a}). Its range is a subbundle of
$TY\to Y$ called the horizontal distribution. By virtue of Theorem
\ref{sp11}, a connection on a fibred manifold always exists. A
connection $\G$ (\ref{150}) is represented by the horizontal
tangent-valued one-form
\mar{154}\beq
\G = dx^\la\otimes (\dr_\la + \G_\la^i\dr_i) \label{154}
\eeq
on $Y$ which is projected onto the canonical tangent-valued form
$\thh_X$ (\ref{b1.51}) on $X$.

Given a connection $\G$ on a fibred manifold $Y\to X$, any vector
field $\tau$ on a base $X$ gives rise to the projectable vector
field
\mar{b1.85}\beq
\G \tau=\tau\rfloor\G=\tau^\la(\dr_\la +\G^i_\la\dr_i)
\label{b1.85}
\eeq
on $Y$ which lives in the horizontal distribution determined by
$\G$. It is called the horizontal lift of $\tau$ by means of a
connection $\G$.

The splitting (\ref{150}) also is given by the vertical-valued
form
\mar{b1.223}\beq
\G= (dy^i -\G^i_\la dx^\la)\ot\dr_i, \label{b1.223}
\eeq
which yields an epimorphism $TY\to VY$. It provides the
corresponding splitting
\mar{cc3}\ben
&& \G: V^*Y\ni \ol dy^i\mapsto dy^i-\G^i_\la dx^\la\in
T^*Y, \label{cc3}\\
&&\dot x_\la dx^\la + \dot y_i dy^i= (\dot x_\la +\dot
y_i\G^i_\la)dx^\la +\dot y_i(dy^i-\G^i_\la dx^\la), \nonumber
\een
of the dual exact sequence (\ref{1.8b}).

In an equivalent way, connections on a fibred manifold $Y\to X$
are introduced as global sections of the affine jet bundle
$J^1Y\to Y$. Indeed, any global section $\G$ of $J^1Y\to Y$
defines the tangent-valued form $\la_1\circ \G$ (\ref{154}). It
follows from this definition that connections on a fibred manifold
$Y\to X$ constitute an affine space modelled over the vector space
of soldering forms $\si$ (\ref{1290}). One also deduces from
(\ref{50}) the coordinate transformation law of connections
\be
\G'^i_\la = \frac{\dr x^\m}{\dr{x'}^\la}(\dr_\m
+\G^j_\m\dr_j)y'^i.
\ee

\begin{rem} \label{Ehresmann}  Any connection $\G$ on a fibred manifold $Y\to
X$ yields a horizontal lift of a vector field on $X$ onto $Y$, but
need not defines the similar lift of a path in $X$ into $Y$. Let
\be
\Bbb R\supset[,]\ni t\to x(t)\in X, \qquad \Bbb R\ni t\to y(t)\in
Y,
\ee
be smooth paths in $X$ and $Y$, respectively. Then $t\to y(t)$ is
called a horizontal lift of $x(t)$ if
\be
\pi(y(t))= x(t), \qquad \dot y(t)\in H_{y(t)}Y, \qquad t\in\Bbb R,
\ee
where $HY\subset TY$ is the horizontal subbundle associated to the
connection $\G$. If, for each path $x(t)$ $(t_0\leq t\leq t_1)$
and for any $y_0\in\pi^{-1}(x(t_0))$, there exists a horizontal
lift $y(t)$ $(t_0\leq t\leq t_1)$ such that $y(t_0)=y_0$, then
$\G$ is called the Ehresmann connection. A fibred manifold is a
fibre bundle iff it admits an Ehresmann connection \cite{gre}.
\end{rem}

Hereafter, we restrict our consideration to connections on fibre
bundles. The following are two standard constructions of new
connections from old ones.

$\bullet$ Let $Y$ and $Y'$ be fibre bundles over the same base
$X$. Given connections $\G$ on $Y$ and $\G'$ on $Y'$, the bundle
product $Y\op\times_X Y'$  is provided with the product connection
\mar{b1.96}\beq
\G\times\G' = dx^\la\ot\left(\dr_\la +\G^i_\la\frac{\dr}{\dr y^i}
+ \G'^j_\la\frac{\dr}{\dr y'^j}\right). \label{b1.96}
\eeq

$\bullet$ Given a fibre bundle $Y\to X$, let $f:X'\to X$ be a
manifold morphism and $f^*Y$ the pull-back of $Y$ over $X'$. Any
connection $\G$ (\ref{b1.223}) on $Y\to X$ yields the pull-back
connection
\mar{mos82}\beq
f^*\G=\left(dy^i-\G^i_\la(f^\m(x'^\nu),y^j)\frac{\dr f^\la}{\dr
x'^\m}dx'^\m\right)\ot\dr_i \label{mos82}
\eeq
on the pull-back bundle $f^*Y\to X'$.

Every connection $\G$ on a fibre bundle $Y\to X$ defines the first
order differential operator
\mar{2116}\ben
&& D^\G:J^1Y\op\to_Y T^*X\op\otimes_Y VY, \label{2116}\\
&& D^\G=\la_1- \G\circ \pi^1_0 =(y^i_\la -\G^i_\la)dx^\la\otimes\dr_i,
\nonumber
\een
on $Y$ called the covariant differential. If $s:X\to Y$ is a
section, its covariant differential
\mar{+190}\beq
\nabla^\G s= D_\G\circ J^1s = (\dr_\la s^i - \G_\la^i\circ s)
dx^\la\ot \dr_i \label{+190}
\eeq
and its covariant derivative $\nabla_\tau^\G s
=\tau\rfloor\nabla^\G s$ along a vector field $\tau$ on $X$ are
introduced. In particular, a (local) section $s$ of $Y\to X$ is
called an integral section for a connection $\G$ (or parallel with
respect to $\G$) if $s$ obeys the equivalent conditions
\mar{b1.86}\beq
\nabla^\G s=0 \quad {\rm or} \quad J^1s=\G\circ s. \label{b1.86}
\eeq

Let $\G$ be a connection on a fibre bundle $Y\to X$. Given vector
fields $\tau$, $\tau'$ on $X$ and their horizontal lifts $\G\tau$
and $\G\tau'$ (\ref{b1.85}) on $Y$, let us consider the vertical
vector field
\mar{160a,160}\ben
&& R(\tau,\tau')=\G [\tau,\tau'] - [\G \tau, \G \tau']= \tau^\la
\tau'^\m R_{\la\m}^i\dr_i, \label{160a} \\
&& R_{\la\m}^i = \dr_\la\G_\m^i - \dr_\m\G_\la^i +
\G_\la^j\dr_j \G_\m^i - \G_\m^j\dr_j \G_\la^i. \label{160}
\een
It can be seen as the contraction of vector fields $\tau$ and
$\tau'$ with the vertical-valued horizontal two-form
\mar{161'}\beq
R =\frac{1}{2} [\G,\G]_{\rm FN} = \frac12 R_{\la\m}^i dx^\la\wedge
dx^\m\otimes\dr_i \label{161'}
\eeq
on $Y$ called the curvature form of a connection $\G$.

Given a connection $\G$ and a soldering form $\si$, the torsion of
$\G$ with respect to $\si$ is defined as the vertical-valued
horizontal two-form
\mar{1190}\beq
T = [\G, \si]_{\rm FN} = (\dr_\la\si_\m^i + \G_\la^j\dr_j\si_\m^i
- \dr_j\G_\la^i\si_\m^j) dx^\la\w dx^\m\ox \dr_i. \label{1190}
\eeq

\subsection{Flat connections}

A flat (or curvature-free) connection is a connection $\G$ on a
fibre bundle $Y\to X$ which satisfies the following equivalent
conditions:

$\bullet$ its curvature vanishes everywhere on $Y$;

$\bullet$ its horizontal distribution is involutive;

$\bullet$ there exists a local integral section for the connection
$\G$ through any point $y\in Y$.

By virtue of Theorem \ref{to.1}, a flat connection $\G$  yields a
foliation of $Y$ which is transversal to the fibration $Y\to X$.
It called a horizontal foliation. Its leaf through a point $y\in
Y$ is locally defined by an integral section $s_y$ for the
connection $\G$ through $y$. Conversely, let a fibre bundle $Y\to
X$ admit a horizontal foliation such that, for each point $y\in
Y$, the leaf of this foliation through $y$ is locally defined by a
section $s_y$ of $Y\to X$ through $y$. Then the map
\be
\G:Y\ni y\mapsto j^1_{\pi(y)}s_y\in J^1Y
\ee
sets a flat connection on $Y\to X$. Hence, there is one-to-one
correspondence between the flat connections and the horizontal
foliations of a fibre bundle $Y\to X$.

Given a horizontal foliation of a fibre bundle $Y\to X$, there
exists the associated atlas of bundle coordinates $(x^\la, y^i)$
on $Y$ such that every leaf of this foliation is locally given by
the equations $y^i=$const., and the transition functions $y^i\to
{y'}^i(y^j)$ are independent of the base coordinates $x^\la$
\cite{book09}. It is called the atlas of constant local
trivializations. Two such atlases are said to be equivalent if
their union also is an atlas of the same type. They are associated
to the same horizontal foliation. Thus, the following is proved.

\begin{theo} \label{gena113} \mar{gena113}
There is one-to-one correspondence between the flat connections
$\G$ on a fibre bundle $Y\to X$ and the equivalence classes of
atlases of constant local trivializations of $Y$ such that
$\G=dx^\la\ot\dr_\la$ relative to the corresponding atlas.
\end{theo}

\begin{ex} \label{spr852} \mar{spr852}
Any trivial bundle has flat connections corresponding to its
trivializations. Fibre bundles over a one-dimensional base have
only flat connections.
\end{ex}

\subsection{Linear connections}

Let $Y\to X$ be a vector bundle equipped with linear bundle
coordinates $(x^\la,y^i)$. It admits a linear connection
\mar{167}\beq
\G =dx^\la\ot(\dr_\la + \G_\la{}^i{}_j(x) y^j\dr_i). \label{167}
\eeq
There are the following standard constructions of new linear
connections from old ones.

$\bullet$  Any linear connection $\G$ (\ref{167}) on a vector
bundle $Y\to X$ defines the dual linear connection
\mar{spr300}\beq
\G^*=dx^\la\ot(\dr_\la - \G_\la{}^j{}_i(x) y_j\dr^i)
\label{spr300}
\eeq
on the dual bundle $Y^*\to X$.

$\bullet$ Let $\G$ and $\G'$ be linear connections on vector
bundles $Y\to X$ and $Y'\to X$, respectively. The direct sum
connection $\G\oplus\G'$ on the Whitney sum $Y\oplus Y'$ of these
vector bundles is defined as the product connection (\ref{b1.96}).

$\bullet$ Similarly, the tensor product $Y\ot Y'$ of vector
bundles possesses the tensor product connection
\mar{b1.92}\beq
\G\otimes\G'=dx^\la\ot\left[\dr_\la +(\G_\la{}^i{}_j
y^{ja}+\G'_\la{}^a{}_b y^{ib}) \frac{\dr}{\dr y^{ia}}\right].
\label{b1.92}
\eeq

The curvature of a linear connection $\G$ (\ref{167}) on a vector
bundle $Y\to X$ is usually written as a $Y$-valued two-form
\mar{mos4}\ben
&&R=\frac12 R_{\la\m}{}^i{}_j(x)y^j dx^\la\w dx^\m\ot e_i,\label{mos4}\\
&&R_{\la\m}{}^i{}_j = \dr_\la \G_\m{}^i{}_j - \dr_\m
\G_\la{}^i{}_j + \G_\la{}^h{}_j \G_\m{}^i{}_h - \G_\m{}^h{}_j
\G_\la{}^i{}_h, \nonumber
\een
due to the canonical vertical splitting $VY\cong Y\times Y$, where
$\{\dr_i\}=\{e_i\}$. For any two vector fields $\tau$ and $\tau'$
on $X$, this curvature yields the zero order differential operator
\mar{+98}\beq
R(\tau,\tau')
s=([\nabla_\tau^\G,\nabla_{\tau'}^\G]-\nabla_{[\tau,\tau']}^\G)s
\label{+98}
\eeq
on section $s$ of a vector bundle $Y\to X$.

An important example of linear connections is a  connection
\mar{B}\beq
K= dx^\la\otimes (\dr_\la +K_\la{}^\m{}_\n \dot x^\n \dot\dr_\m)
\label{B}
\eeq
on the tangent bundle $TX$ of a manifold $X$. It is called a world
connection or, simply, a connection on a manifold  $X$. The dual
connection (\ref{spr300}) on the cotangent bundle $T^*X$ is
\mar{C}\beq
K^*= dx^\la\otimes (\dr_\la -K_\la{}^\m{}_\n\dot x_\m \dot\dr^\n).
\label{C}
\eeq

The curvature of a world connection $K$ (\ref{B}) reads
\mar{1203}\ben
&& R=\frac12R_{\la\m}{}^\al{}_\bt\dot x^\bt dx^\la\w dx^\m\ot\dr_\al,
 \label{1203}\\
&& R_{\la\m}{}^\al{}_\bt = \dr_\la K_\m{}^\al{}_\bt - \dr_\m
K_\la{}^\al{}_\bt + K_\la{}^\g{}_\bt K_\m{}^\al{}_\g -
K_\m{}^\g{}_\bt K_\la{}^\al{}_\g.\nonumber
\een
Its Ricci tensor $R_{\la\bt}=R_{\la\m}{}^\m{}_\bt$ is introduced.

A torsion of a world connection is defined as the torsion
(\ref{1190}) of the connection $K$ (\ref{B}) on the tangent bundle
$TX$ with respect to the canonical vertical-valued form
$dx^\la\ot\dot\dr_\la$. Due to the vertical splitting of $VTX$, it
also is written as a tangent-valued two-form
\mar{191}\beq
T =\frac12 T_\m{}^\n{}_\la  dx^\la\w dx^\m\ot\dr_\n, \qquad
T_\m{}^\n{}_\la  = K_\m{}^\n{}_\la - K_\la{}^\n{}_\m, \label{191}
\eeq
on $X$. A world connection (\ref{B}) is called symmetric if its
torsion (\ref{191}) vanishes.

For instance, let a manifold $X$ be provided with a non-degenerate
fibre metric
\be
g\in\op\vee^2\cO^1(X), \qquad g=g_{\la\m}dx^\la\ot dx^\m,
\ee
in the tangent bundle $TX$, and with the dual metric
\be
g\in\op\vee^2\cT^1(X), \qquad g=g^{\la\m}\dr_\la\ot \dr_\m,
\ee
in the cotangent bundle $T^*X$. Then there exists a world
connection $K$ such that $g$ is its integral section, i.e.,
\be
\nabla_\la g^{\al\bt}=\dr_\la\, g^{\al \bt} -
g^{\al\g}K_\la{}^\bt{}_\g - g^{\bt\g}K_\la{}^\al{}_\g=0.
\ee
It is called the metric connection. There exists a unique
symmetric metric connection
\mar{b1.400}\beq
K_\la{}^\n{}_\m =
\{_\la{}^\n{}_\m\}=-\frac{1}{2}g^{\n\rho}(\dr_\la g_{\rho\m} +
\dr_\m g_{\rho\la}-\dr_\rho g_{\la\m}). \label{b1.400}
\eeq
This is the Levi--Civita connection, whose components
(\ref{b1.400}) are called Christoffel symbols.

\subsection{Composite connections}

Let us consider the composite bundle $Y\to\Si\to X$ (\ref{1.34}),
coordinated by $(x^\la, \si^m, y^i)$. Let us consider the jet
manifolds $J^1\Si$, $J^1_\Si Y$, and $J^1Y$ of the fibre bundles
$\Si\to X$, $Y\to \Si$ and $Y\to X$, respectively. They are
parameterized respectively by the coordinates
\be
( x^\la ,\si^m, \si^m_\la),\quad ( x^\la ,\si^m, y^i, \wt y^i_\la,
y^i_m),\quad ( x^\la ,\si^m, y^i, \si^m_\la ,y^i_\la).
\ee
There is the canonical map
\mar{1.38}\beq
\vr : J^1\Si\op\times_\Si J^1_\Si Y\ar_Y J^1Y, \qquad
y^i_\la\circ\vr=y^i_m{\si}^m_{\la} +\wt y^i_{\la}. \label{1.38}
\eeq
Using the canonical map (\ref{1.38}), we can consider the
relations between connections on fibre bundles $Y\to X$, $Y\to\Si$
and $\Si\to X$ \cite{book00,sau}.

Connections on fibre bundles $Y\to X$, $Y\to\Si$ and $\Si\to X$
read
\mar{spr290-2}\ben
&& \g=dx^\la\ot (\dr_\la +\g_\la^m\dr_m + \g_\la^i\dr_i), \label{spr290}\\
&&  A_\Si=dx^\la\ot (\dr_\la + A_\la^i\dr_i) +d\si^m\ot (\dr_m + A_m^i\dr_i),
\label{spr291}\\
&& \G=dx^\la\ot (\dr_\la + \G_\la^m\dr_m). \label{spr292}
\een
The canonical map $\vr$ (\ref{1.38}) enables us to obtain a
connection $\g$ on $Y\to X$ in accordance with the diagram
\be
\begin{array}{rcccl}
 & J^1\Si\op\xx_\Si J^1_\Si Y & \ar^\vr & J^1Y & \\
_{(\G,A)} & \put(0,-10){\vector(0,1){20}} & &
\put(0,-10){\vector(0,1){20}} & _{\g} \\
 & \Si\op\xx_X Y & \longleftarrow & Y &
\end{array}
\ee
This connection, called the composite connection, reads
\mar{b1.114}\beq
\g=dx^\la\ot [\dr_\la +\G_\la^m\dr_m + (A_\la^i +
A_m^i\G_\la^m)\dr_i]. \label{b1.114}
\eeq
It is a unique connection such that the horizontal lift $\g\tau$
on $Y$ of a vector field $\tau$ on $X$ by means of the connection
$\g$ (\ref{b1.114}) coincides with the composition $A_\Si(\G\tau)$
of horizontal lifts of $\tau$ onto $\Si$ by means of the
connection $\G$ and then onto $Y$ by means of the connection
$A_\Si$. For the sake of brevity, let us write $\g=A_\Si\circ\G$.

Given the composite bundle $Y$ (\ref{1.34}), there is the exact
sequence
\mar{63a}\beq
0\to V_\Si Y\to VY\to Y\op\times_\Si V\Si\to 0, \label{63a}
\eeq
where $V_\Si Y$ denotes the vertical tangent bundle of a fibre
bundle $Y\to\Si$ coordinated by $(x^\la,\si^m,y^i, \dot y^i)$. Let
us consider the splitting
\mar{63c}\ben
&& B: VY\ni v=\dot y^i\dr_i +\dot \si^m\dr_m \mapsto v\rfloor B= \label{63c}\\
&& \qquad  (\dot y^i -\dot \si^m B^i_m)\dr_i\in V_\Si
Y, \nonumber \\
&& B=(\ol d y^i- B^i_m\ol d\si^m)\ot\dr_i\in
V^*Y\op\ot_Y V_\Si Y, \nonumber
\een
of the exact sequence (\ref{63a}). Then the connection $\g$
(\ref{spr290}) on $Y\to X$ and the splitting $B$ (\ref{63c})
define a connection
\mar{nnn}\ben
&& A_\Si=B\circ \g: TY\to VY \to V_\Si Y, \nonumber\\
&& A_\Si=dx^\la\ot(\dr_\la +(\g^i_\la - B^i_m\g^m_\la)\dr_i)+
\label{nnn}\\
&& \qquad d\si^m\ot (\dr_m +B^i_m\dr_i), \nonumber
\een
on the fibre bundle $Y\to\Si$.

Conversely, every connection $A_\Si$ (\ref{spr291}) on a fibre
bundle $Y\to\Si$ provides the splitting
\mar{46a}\ben
&& VY=V_\Si Y\op\oplus_Y A_\Si(Y\op\times_\Si V\Si),\label{46a}\\
&& \dot y^i\dr_i + \dot\si^m\dr_m= (\dot y^i -A^i_m\dot\si^m)\dr_i
+ \dot\si^m(\dr_m+A^i_m\dr_i), \nonumber
\een
of the exact sequence (\ref{63a}). Using this splitting, one can
construct the first order differential operator
\mar{7.10}\beq
\wt D: J^1Y\to T^*X\op\otimes_Y V_\Si Y, \qquad \wt D=
dx^\la\otimes(y^i_\la- A^i_\la -A^i_m\si^m_\la)\dr_i, \label{7.10}
\eeq
called the vertical covariant differential, on the composite fibre
bundle $Y\to X$.

The vertical covariant differential (\ref{7.10}) possesses the
following important property. Let $h$ be a section of a fibre
bundle $\Si\to X$, and let $Y_h\to X$ be the restriction of a
fibre bundle $Y\to\Si$ to $h(X)\subset \Si$.  This is a subbundle
$i_h:Y_h\to Y$ of a fibre bundle $Y\to X$. Every connection
$A_\Si$ (\ref{spr291}) induces the pull-back connection
\mar{mos83}\beq
A_h=i_h^*A_\Si=dx^\la\ot[\dr_\la+((A^i_m\circ h)\dr_\la h^m
+(A\circ h)^i_\la)\dr_i] \label{mos83}
\eeq
on $Y_h\to X$. Then the restriction of the vertical covariant
differential $\wt D$ (\ref{7.10}) to $J^1i_h(J^1Y_h)\subset J^1Y$
coincides with the familiar covariant differential $D^{A_h}$
(\ref{2116}) on $Y_h$ relative to the pull-back connection $A_h$
(\ref{mos83}).

\bibliographystyle{alpha}
\bibliographystyle{plain}
\addcontentsline{toc}{chapter}{Bibliography}
\bibliography{conn99}

\begin{thebibliography}{ederf}

\bibitem{abr} Abraham, R.  and Marsden, J. (1978). \emph{Foundations of Mechanics}
(Benjamin / Cummings Publ. Comp., London).

\bibitem{arn} Arnold, V. (Ed.) (1988). \emph{Dynamical Systems III}
(Springer, Berlin).

\bibitem{bry} Bryant, R., Chern, S., Gardner, R., Goldschmidt, H. and
Griffiths, P. (1991). \emph{Exterior Differential Systems}
(Springer-Verlag, Berlin).

\bibitem{dew} Dewisme, A. and Bouquet, S. (1993). First integrals and
symmetries of time-dependent Hamiltonian systems, \emph{J. Math.
Phys} \textbf{34}, 997.

\bibitem{eche95} Echeverr\'{\i}a Enr\'{\i}quez, A., Mu\~noz Lecanda, M. and
Rom\'an Roy, N. (1995). Non-standard connections in classical
mechanics, \emph{J. Phys. A} \textbf{28}, 5553.

\bibitem{eche} Echeverr\'{\i}a Enr\'{\i}quez, A., Mu\~noz Lecanda, M. and
Rom\'an Roy, N. (1991). Geometrical setting of time-dependent
regular systems. Alternative models, \emph{Rev. Math. Phys.}
\textbf{3}, 301.

\bibitem{fat} Fatibene, L., Ferraris, M., Francaviglia, M.
and  McLenaghan, R. (2002). Generalized symmetries in mechanics
and field theories, \emph{J. Math. Phys.} \textbf{43}, 3147.


\bibitem{book}  Giachetta, G., Mangiarotti, L. and Sardanashvily, G. (1997).
\emph{New Lagrangian and Hamiltonian Methods in Field Theory}
(World Scientific, Singapore).

\bibitem{jpa99} Giachetta, G., Mangiarotti, L. and Sardanashvily, G. (1999). Covariant
Hamilton equations for field theory, \emph{J. Phys. A}
\textbf{32}, 6629.

\bibitem{jmp02} Giachetta, G., Mangiarotti, L. and Sardanashvily, G.
(2002). Geometric quantization of mechanical systems with
time-dependent parameters, \emph{J. Math. Phys} \textbf{43}, 2882.

\bibitem{cmp04} Giachetta, G., Mangiarotti, L. and Sardanashvily, G.
(2005). Lagrangian supersymmetries depending on derivatives.
Global analysis and cohomology, \emph{Commun. Math. Phys.}
\textbf{259}, 103.

\bibitem{book05} Giachetta, G., Mangiarotti, L. and Sardanashvily, G.
(2005). \emph{Geometric and Topological Algebraic Methods in
Quantum Mechanics} (World Scientific, Singapore).

\bibitem{jmp09} Giachetta, G., Mangiarotti, L. and Sardanashvily, G.
(2009). On the notion of gauge symmetries of generic Lagrangian
field theory, \emph{J. Math. Phys.} \textbf{50}, 012903.

\bibitem{book09} Giachetta, G., Mangiarotti, L. and Sardanashvily, G.
(2009). \emph{Advanced Classical Field Theory} (World Scientific,
Singapore).

\bibitem{book10} Giachetta, G., Mangiarotti, L. and Sardanashvily, G.
(2010). \emph{Advanced Mechanics: Classical and Quantum} (World
Scientific, Singapore).

\bibitem{got91} Gotay, M. (1991). A multisymplectic framework for classical field
theory and the calculus of variations. I. Covariant Hamiltonian
formalism, In \emph{Mechanics, Analysis and Geometry: 200 Years
after Lagrange} (North Holland, Amsterdam) p. 203.

\bibitem{grab} Grabowski, J. and Urba\'nski, P. (1995). Tangent lifts of Poisson and
related structures, \emph{J. Phys. A} \textbf{28}, 6743.

\bibitem{gre} Greub, W.,  Halperin, S. and Vanstone, R. (1972). \emph{Connections, Curvature
and Cohomology} (Academic Press, New York).

\bibitem{hir} Hirzebruch, F. (1966). \emph{Topological Methods in Algebraic Geometry}
(Springer-Verlag, Berlin).

\bibitem{ibr} Ibragimov, N. (1985). \emph{Transformation Groups Applied
to Mathematical Physics} (Riedel, Boston).

\bibitem{kamb} Kamber, F. and Tondeur, P. (1975). \emph{Foliated Bundles and
Characteristic Classes}, Lect. Notes in Mathematics \textbf{493}
(Springer, Berlin).


\bibitem{kob} Kobayashi, S. and Nomizu, K. (1963). \emph{Foundations of Differential
Geometry}  (John Wiley, New York - Singapore).

\bibitem{kol} Kol\'a\v{r}, I.,  Michor, P. and Slov\'ak, J. (1993). \emph{Natural Operations
in Differential Geometry} (Springer-Verlag, Berlin).

\bibitem{kras} Krasil'shchik, I., Lychagin, V. and Vinogradov, A. (1985). \emph{Geometry of
Jet Spaces and Nonlinear Partial Differential Equations} (Gordon
and Breach, Glasgow).


\bibitem{leon} De Le\'on, M. and Rodrigues, P. (1989). \emph{Methods of Differential
Geometry in Analytical Mechanics} (North-Holland, Amsterdam).

\bibitem{libe} Libermann, P. and Marle, C-M. (1987). \emph{Symplectic Geometry and
Analytical Mechanics} (D.Reidel Publishing Company, Dordrecht).


\bibitem{book98} Mangiarotti, L. and Sardanashvily, G. (1998). \emph{Gauge
Mechanics} (World Scientific, Singapore).

\bibitem{jmp00} Mangiarotti, L. and Sardanashvily, G. (2000).
On the geodesic form of second order dynamic equations, \emph{J.
Math. Phys.} \textbf{41}, 835.

\bibitem{mang00} Mangiarotti, L. and Sardanashvily, G. (2000). Constraints
in Hamiltonian time-dependent mechanics, \emph{J. Math. Phys.}
\textbf{41}, 2858.

\bibitem{book00} Mangiarotti, L. and Sardanashvily, G. (2000). \emph{Connections in
Classical and Quantum Field Theory} (World Scientific, Singapore).

\bibitem{jmp07} Mangiarotti, L. and Sardanashvily, G. (2007). Quantum
mechanics with respect to different reference frames, \emph{J.
Math. Phys.} \textbf{48}, 082104.


\bibitem{marl97} Marle, C.-M. (1997). The Schouten--Nijenhuis bracket and interior
products, \emph{J. Geom. Phys.} \textbf{23}, 350.

\bibitem{massa} Massa, E. and Pagani, E. (1994). Jet bundle geometry, dynamical
connections and the inverse problem of Lagrangian mechanics,
\emph{Ann. Inst. Henri Poincar\'e} \textbf{61}, 17.


\bibitem{meig} Meigniez, G. (2002). Submersions, fibrations and
bundles, \emph{Trans. Amer. Math. Soc.} \textbf{354}, 3771.

\bibitem{mora} Morandi, G., Ferrario, C., Lo Vecchio, G., Marmo, G. and Rubano, C. (1990). The
inverse problem in the calculus of variations and the geometry of
the tangent bundle, \emph{Phys. Rep.} \textbf{188}, 147.


\bibitem{olv} Olver, P. (1986). \emph{Applications of Lie Groups to
Differential Equations} (Springer-Verlag, Berlin).


\bibitem{rei} Reinhart, B. (1983). \emph{Differential Geometry and Foliations}
(Springer-Verlag, Berlin).

\bibitem{rie} Riewe, F. (1996). Nonconservative Lagrangian and Hamiltonian mechanics,
\emph{Phys. Rev. E} \textbf{53}, 1890.

\bibitem{sard95} Sardanashvily, G. (1995). \emph{Generalized Hamiltonian Formalism for
Field Theory. Constraint Systems.} (World Scientific, Singapore).

\bibitem{sard98} Sardanashvily, G. (1998). Hamiltonian time-dependent mechanics,
\emph{J. Math. Phys.} \textbf{39}, 2714.


\bibitem{sard08} Sardanashvily, G. (2008). Classical field theory. Advanced mathematical formulation,
\emph{Int. J. Geom. Methods Mod. Phys.} \textbf{5}, 1163.

\bibitem{ijgmmp09}  Sardanashvily, G. (2009). Gauge conservation laws in a general
setting. Superpotential, \emph{Int. J. Geom. Methods Mod. Phys.}
\textbf{6}, 1046.


\bibitem{sau} Saunders, D. (1989). \emph{The Geometry of Jet Bundles}
(Cambridge Univ. Press, Cambridge).

\bibitem{ste} Steenrod, N. (1972). \emph{The Topology of Fibre Bundles} (Princeton Univ.
Press, Princeton).


\bibitem{tam} Tamura, I. (1992). \emph{Topology of Foliations: An Introduction},
Transl. Math. Monographs \textbf{97} (AMS, Providence).

\bibitem{vais} Vaisman, I. (1994). \emph{Lectures on the Geometry of Poisson Manifolds}
(Birkh\"auser Verlag, Basel).


\bibitem{war} Warner, F. (1983). \emph{Foundations of Differential Manifolds and Lie
Groups} (Springer-Verlag, Berlin).


\end{thebibliography}

\end{document}